\documentclass[11pt]{article}
\usepackage[margin=25mm]{geometry}
\usepackage{amsmath,amssymb,amsthm,mathtools,booktabs,array,etoolbox,graphicx,longtable}
\usepackage[colorlinks=true,linkcolor=blue,citecolor=blue,urlcolor=blue,unicode]{hyperref}
\allowdisplaybreaks

\theoremstyle{definition}
\newtheorem{theorem}{Theorem}[section]
\newtheorem{lemma}[theorem]{Lemma}
\newtheorem{proposition}[theorem]{Proposition}
\newtheorem{corollary}[theorem]{Corollary}

\newtheorem*{remark*}{Remark}
\newtheorem{definition}[theorem]{Definition}
\newtheorem{example}[theorem]{Example}
\newtheorem{conjecture}[theorem]{Conjecture}
\AtBeginEnvironment{theorem}{\pushQED{\qed}}
\AtEndEnvironment{theorem}{\popQED}
\AtBeginEnvironment{lemma}{\pushQED{\qed}}
\AtEndEnvironment{lemma}{\popQED}
\AtBeginEnvironment{proposition}{\pushQED{\qed}}
\AtEndEnvironment{proposition}{\popQED}
\AtBeginEnvironment{corollary}{\pushQED{\qed}}
\AtEndEnvironment{corollary}{\popQED}
\AtBeginEnvironment{remark}{\pushQED{\qed}}
\AtEndEnvironment{remark}{\popQED}
\AtBeginEnvironment{remark*}{\pushQED{\qed}}
\AtEndEnvironment{remark*}{\popQED}
\AtBeginEnvironment{definition}{\pushQED{\qed}}
\AtEndEnvironment{definition}{\popQED}
\AtBeginEnvironment{example}{\pushQED{\qed}}
\AtEndEnvironment{example}{\popQED}

\newcommand{\engappendixsection}[2]{%
  \refstepcounter{section}%
  \section*{Appendix \thesection\quad #1}%
  \label{#2}%
}

\newcommand{\F}{\mathbb F}
\newcommand{\EE}{\mathbb E}
\newcommand{\PP}{\mathbb P}
\newcommand{\rank}{\operatorname{rank}}

\newcommand{\wt}{\operatorname{wt}}
\newcommand{\Row}{\operatorname{Row}}
\newcommand{\Ker}{\operatorname{Ker}}
\newcommand{\bu}{\mathbf u}
\newcommand{\bp}{\mathbf p}
\newcommand{\bx}{\mathbf x}
\newcommand{\bv}{\mathbf v}
\newcommand{\be}{\mathbf e}
\newcommand{\bq}{\mathbf q}
\newcommand{\bs}{\mathbf s}
\newcommand{\bw}{\mathbf w}
\newcommand{\bt}{\mathbf t}
\newcommand{\bff}{\mathbf f}
\newcommand{\zeros}{\mathbf 0}
\newcommand{\ones}{\mathbf 1}
\newcommand{\h}{h_2}
\newcommand{\CZ}{C_Z}
\newcommand{\CX}{C_X}

\newcommand{\NX}{N_X}

\title{\textbf{Finite-Degree Quantum LDPC Codes Reaching the Gilbert--Varshamov Bound}}
\author{Kenta Kasai\\Institute of Science Tokyo\\ \texttt{kenta@ict.eng.isct.ac.jp}}
\date{}

\begin{document}
\maketitle

\begin{abstract}
We construct asymptotically good nested Calderbank--Shor--Steane (CSS) code pairs from Hsu--Anastasopoulos codes and MacKay--Neal codes. For fixed balanced triples with even component degrees and \(j_Z\ge4\), we prove that the coding rate stays bounded away from zero and that the relative distances on both sides stay bounded away from zero with probability tending to one as the blocklength grows. Moreover, within an explicit low-degree search window, we prove that all 56 balanced triples in that window with even common row degree, column degree at least four, and positive design quantum rate attain the classical Gilbert--Varshamov (GV) bound on both constituent sides, and consequently the CSS GV bound at fixed finite degree.
\end{abstract}

\section{Introduction}

Recent breakthroughs established asymptotically good Calderbank--Shor--Steane (CSS) low-density parity-check (LDPC) codes with simultaneous rate, distance, and sparsity~\cite{PK21,DHLV22,LZ22}. Those constructions, however, are not organized around an explicit fixed-degree finite-parameter landscape. The question studied here is whether one can exhibit quantum LDPC families that keep the degrees genuinely small and nevertheless reach the Gilbert--Varshamov (GV) benchmark already at fixed finite degree. On the classical side, this question belongs to the sparse-graph coding tradition initiated by Tanner's graph-based formulation~\cite{Tanner81}.

At a high level, our answer is affirmative in a concrete balanced regular regime. For every fixed balanced triple with even component degrees and \(4\le j_Z<k/2\), we prove relative linear distance. Inside the specified low-degree search window, we prove finite-degree attainment of the classical GV distance on both constituent sides for all 56 balanced triples with even \(k\) and \(j_Z\ge4\), including 24 triples with odd \(j_Z\). Consequently, the corresponding nested CSS family attains the CSS GV distance at fixed finite degree throughout that certified region.

Our starting point is a classical pair of sparse-graph families with complementary strengths. MacKay--Neal (MN) codes~\cite{MN96} have very sparse graphical representations and near-capacity performance, while Hsu--Anastasopoulos (HA) codes~\cite{HA05} retain bounded graphical complexity and admit maximum-likelihood capacity theorems. In both cases, the visible code is a punctured sparse-graph code: it is first described by an extended parity-check matrix with hidden variables and then equivalently by a compressed parity-check matrix obtained after eliminating those variables. In the homogeneous regular specialization studied later, the degree triple \((j_Z,j_X,k)\) records the column socket degrees of the constituent matrices \(A_Z\) and \(A_X\), while the common parameter \(k\) is their row socket degree and also the row and column socket degree of the square sparse map \(B\). A direct dual pairing of these families is useless for CSS coding because it gives zero quantum rate. The main structural step of this paper is instead to impose a nested relation on the MN side and combine it with the HA side so as to obtain a balanced CSS family with positive design quantum rate. We sample the matrices from the socket-based random graph ensemble of~\cite[Secs.~3.3--3.4]{RU08}, which is natural both for the exact-enumerator analysis carried out here and for future density-evolution-based decoding analyses~\cite[Sec.~3.9 and App.~B]{RU08}.

From the broader theory literature, this paper is closest in spirit to two lines of work. On the classical side, Sipser and Spielman's expander codes~\cite{SipserS94} showed that constant-degree sparse graphs can already support linear distance and efficient decoding. On the quantum side, the pre-breakthrough benchmark of Tillich and Z{\'e}mor~\cite{TZ09} gave positive rate together with square-root distance, while quantum expander codes~\cite{LeverrierTZ15}, decodable quantum LDPC codes from high-dimensional expanders~\cite{EvraKZ20}, quantum Tanner codes~\cite{LZ22}, balanced product quantum codes~\cite{BreuckmannE21}, and the recent asymptotically good constructions~\cite{PK21,DHLV22} established a modern trajectory toward sparse quantum codes with simultaneously strong asymptotic parameters. Our result is complementary to those works: rather than building asymptotically good families from expansion or product constructions, we isolate a concrete finite-degree regime and prove that the specified sparse quantum CSS construction already attains the GV benchmark there.

The sparse family has an explicit random sampling rule; the distance statements below are high-probability ensemble results, rather than a deterministic algorithm selecting a good code at each blocklength. On the HA side, we adapt the visible-weight analysis of~\cite{HA05} to the socket-based regular ensemble of~\cite[Secs.~3.3--3.4]{RU08}. On the MN side, we analyze the true stacked ensemble \(A_X=[A_Z;A_\Delta]\) rather than a homogeneous surrogate, which requires a refined exact-enumerator calculation for a genuinely nonstandard sparse matrix model.

The main contributions are:
\begin{enumerate}
\item For every fixed balanced degree triple \((j_Z,j_X,k)\) with \(j_Z\ge4\), even \(j_Z\), and even \(j_\Delta\), we prove the existence of positive constants \(\delta_Z\) and \(\delta_X\) such that \(\PP[d_Z^{\mathrm{rel}}\le \delta_Z n]\to 0\) and \(\PP[d_X^{\mathrm{rel}}\le \delta_X n]\to 0\).
\item Inside the stated low-degree search window, we prove that all 56 balanced triples with even \(k\) and \(j_Z\ge4\) attain the classical GV distance at finite degree on both constituent sides.
\item Consequently, the corresponding nested CSS family attains the CSS GV distance at fixed finite degree throughout that region.
\end{enumerate}

The certified set is a sufficient parameter region for GV attainment. No converse excluding attainment outside this set is asserted. Of its 56 triples, 32 have even \(j_Z\); the remaining 24 use an odd-syndrome complement argument.

This finite-degree theorem is genuinely different from large-degree limit statements. In particular, the HA result of~\cite{HA05} is an \(\varepsilon\)--\(K\) theorem: it shows approach to the GV bound as \(k\to\infty\), not finite-degree GV attainment for each fixed degree triple. Our proof closes the positive-rate finite-degree cases by reducing the exponent inequalities to scalar entropy and convexity inequalities, followed by exact rational endpoint comparisons. The zero-rate boundary cases use rational interval subdivision with explicit logarithm remainders~\cite{MKC09,Tucker11}.

At a high level, the proof has three layers. First, exact ensemble formulas reduce the distance problem to exponent inequalities for visible-weight enumerators. Second, almost all of the parameter space is handled analytically by trial-point bounds and pairing estimates. Third, finite scalar comparisons are verified using exact rational arithmetic; the separate zero-rate boundary proof certifies its residual scalar intervals by rational subdivision. On the MN side, entropy concavity reduces the stacked exponent to one hidden-weight variable. On the HA side, entropy conjugacy removes the visible weight. The accompanying standard-library verifiers check every finite comparison and every interval in the stated degree ranges.

The remainder of this paper is organized as follows. Section~2 defines the family and its basic properties. Sections~3 and~4 state the HA-side and MN-side distance theorems. Section~5 deduces the CSS relative-distance and CSS-GV consequences. Section~6 records the finite-degree parameter landscape. The appendices contain the deferred proofs, including the certification arguments and the convergence of actual rates to the design rates.

\section{Construction and Basic Properties}

We work over the binary field \(\F_2\).
\(\Ker M\) denotes the kernel of a matrix \(M\), and \(\Row(M)\) its row space.
The binary entropy function is \(\h(x):=-x\log_2 x-(1-x)\log_2(1-x)\)
for \(0<x<1\), extended by \(\h(0)=\h(1)=0\).
The inverse \(\h^{-1}\) always denotes the branch with values in
\([0,1/2]\). For a linear code \(C\), write
\(d(C):=\min\{\wt(\bv):\bv\in C\setminus\{0\}\}\), with
\(d(\{0\})=+\infty\).
This section first gives in Section 2.1 the general framework of the proposed code pair under the sole nested assumption \(\Row(A_Z)\subseteq\Row(A_X)\) for arbitrary matrices \(A_Z,A_X,B\), together with the CSS condition and the dimension formulas for the actual rates. Section 2.2 then specializes this framework to regular LDPC matrices and a square regular sparse map. Section 2.3 states the design rates, the balanced condition, and an illustrative concrete example. Finally, Section 2.4 explains how to represent compressed syndromes by sparse affine systems.

\subsection{General Definition}

\begin{definition}[General Framework of the Proposed Construction]\label{def:framework-jp}
Let \(A_Z\in\F_2^{m_Z\times n}\), \(A_X\in\F_2^{m_X\times n}\), and \(B\in\F_2^{n\times n}\) be arbitrary matrices satisfying \(\Row(A_Z)\subseteq \Row(A_X)\).
For example, this condition is automatically satisfied if \(A_Z\) is chosen as a row submatrix of \(A_X\).
In what follows, the rightmost \(n\) coordinates correspond to the visible variable \(\bv\in\F_2^n\), while the left block corresponds to hidden variables. More precisely, the left block is \(\bu\in\F_2^n\) for \(H_Z'\) and \(\bw\in\F_2^{m_X}\) for \(H_X'\).
Define the extended parity-check matrices with hidden variables by
\begin{equation*}
\begin{aligned}
H_Z'
:=
\begin{bmatrix}
A_Z & 0\\
B   & I_n
\end{bmatrix},
\qquad
H_X'
:=
\begin{bmatrix}
A_X^T & B^T
\end{bmatrix}
\end{aligned}
\end{equation*}
The relevant codes \(\CZ\) and \(\CX\) are obtained by puncturing the left hidden-variable part of \(\Ker H_Z'\) and \(\Ker H_X'\), respectively. Concretely,
\(\CZ=\{\bv\in\F_2^n:\exists \bu\in\F_2^n,\ H_Z'(\bu,\bv)^T=\zeros\}\)
and
\(\CX=\{\bv\in\F_2^n:\exists \bw\in\F_2^{m_X},\ H_X'(\bw,\bv)^T=\zeros\}\)
define the Z-side and X-side codes.
\end{definition}

In this paper, we refer to the Z-side code \(\CZ\) and the X-side code \(\CX\) defined above, and later to their regular sparse specializations, as HA-type and MN-type codes in a broad sense, respectively. This terminology is broader than in the original papers~\cite{MN96,HA05}: there the ensembles are introduced through more specific graphical constructions and decoding viewpoints, whereas here we abstract them as punctured hidden-variable codes represented by extended parity-check matrices and, equivalently, by compressed parity-check matrices tailored to the nested CSS construction.

In what follows, especially when we specialize to regular sparse families, these matrices are sampled according to the socket-based random graph ensemble of~\cite[Secs.~3.3--3.4]{RU08}. This choice is natural for the exact-enumerator analysis carried out here and is also consistent with future decoding analyses based on density evolution~\cite[Sec.~3.9 and App.~B]{RU08}.

This hidden-variable representation immediately yields the reduced forms used later in the analysis: \(\CZ=B(\Ker A_Z)\), \(C_Z(A_X):=B(\Ker A_X)\), and \(\CX=\{\bv\in\F_2^n:\exists \bw,\ A_X^T \bw+B^T \bv=0\}\). Indeed, \(H_Z'(\bu,\bv)^T=\zeros\) is equivalent to \(A_Z\bu=\zeros\) and \(\bv=B\bu\), while \(H_X'(\bw,\bv)^T=\zeros\) is equivalent to \(A_X^T\bw+B^T\bv=\zeros\). The auxiliary code \(C_Z(A_X)\), used later, is obtained from the Z-side expression for \(\CZ\) by replacing \(A_Z\) with \(A_X\).

Up to this point, all matrices, codes, and parity-check representations are deterministic objects attached to a fixed choice of \(A_Z,A_X,B\). In particular, the general framework of Section~2.1, the compressed parity-check representatives, and the dimension identities below are pointwise algebraic statements for each realization, and no randomness has yet been introduced.

\begin{theorem}[Nested CSS pair]\label{prop:css-jp}
Under Definition~\ref{def:framework-jp}, \(\CZ^\perp\subseteq \CX\) holds. In particular, \((\CX,\CZ)\) forms a CSS pair.
\end{theorem}

\begin{proof}
For any subspace \(U\subseteq\F_2^n\),
\[
\bv\in (BU)^\perp
\iff
\langle \bv,B\bu\rangle =0\ \forall\,\bu\in U
\iff
\langle B^T\bv,\bu\rangle =0\ \forall\,\bu\in U
\iff
B^T\bv\in U^\perp.
\]
First,
\(\CZ^\perp=\{\bv\in\F_2^n:B^T\bv\in(\Ker A_Z)^\perp=\Row(A_Z)\}\),
and similarly
\(\bigl(C_Z(A_X)\bigr)^\perp=\{\bv\in\F_2^n:B^T\bv\in(\Ker A_X)^\perp=\Row(A_X)\}=\CX\).
The last equality holds because \(B^T\bv\in\Row(A_X)\) can be written as
\(A_X^T\bw+B^T\bv=0\) for some \(\bw\).
On the other hand, \(\Row(A_Z)\subseteq \Row(A_X)\) implies \(\Ker A_X\subseteq \Ker A_Z\), so
\(
C_Z(A_X)=B(\Ker A_X)\subseteq B(\Ker A_Z)=\CZ
\).
Taking orthogonal complements gives
\(
\CZ^\perp\subseteq \bigl(C_Z(A_X)\bigr)^\perp=\CX
\)
as claimed.
\end{proof}

\begin{definition}[Compressed parity-check matrices]\label{def:compressed-pcm-jp}
For the general framework of Definition~\ref{def:framework-jp}, let \(\Row(H_Z)=\CZ^\perp\) and \(\Row(H_X)=\CX^\perp\), and call any visible-variable matrices \(H_Z\in\F_2^{r_Z\times n}\) and \(H_X\in\F_2^{r_X\times n}\) compressed parity-check matrices for \(\CZ\) and \(\CX\), respectively, if they satisfy these row-space conditions.
By the proof of Theorem~\ref{prop:css-jp}, one may take
\(\Row(H_Z)=\{\bv\in\F_2^n:\exists \bx\in\F_2^{m_Z},\ A_Z^T\bx+B^T\bv=0\}\)
and
\(\Row(H_X)=C_Z(A_X)=B(\Ker A_X)\).
In particular, if \(K_X\) is a basis matrix of \(\Ker A_X\), then \(H_X=K_XB^T\) is one possible choice of \(H_X\), and a basis of the kernel of \([A_Z^T\ B^T]\), projected to the visible component, gives one possible choice of \(H_Z\).
These matrices are not unique in general; different row-equivalent representatives define the same code.
\end{definition}

To distinguish them from the design rates introduced later in Definition~\ref{def:rates-jp}, we call the following quantities the actual coding rates:
\(R_Z:=\frac{\dim \CZ}{n}\), \(R_X:=\frac{\dim \CX}{n}\), and \(R_Q:=R_X+R_Z-1\),
where \(R_Q\) is the quantum rate. This is the standard CSS dimension formula
\(\dim Q=\dim \CX+\dim \CZ-n\) under \(\CZ^\perp\subseteq \CX\); see, for example,
\cite{CS96,Ste96}.

\begin{proposition}[Dimension formulas for the actual rates]\label{prop:actual-rates-jp}
\(L_Z:=\dim(\Ker A_Z\cap \Ker B)\) and \(L_X:=\dim(\Ker A_X\cap \Ker B)\).
Then \(R_Z=\frac{n-\rank A_Z-L_Z}{n}\), \(R_X=\frac{\rank A_X+L_X}{n}\), \(R_Q=\frac{\rank A_X-\rank A_Z+L_X-L_Z}{n}\), and moreover \(L_X\le L_Z\) holds.
\end{proposition}

\begin{proof}
By the rank-nullity formula for the image of a linear map, \(\dim \CZ=\dim B(\Ker A_Z)=\dim \Ker A_Z-\dim(\Ker A_Z\cap\Ker B)=n-\rank A_Z-L_Z\), which yields \(R_Z=\frac{n-\rank A_Z-L_Z}{n}\) for the Z-side rate.
Also, \(\CX=\bigl(B(\Ker A_X)\bigr)^\perp\), so
\(\dim \CX=n-\dim B(\Ker A_X)=n-(n-\rank A_X-L_X)=\rank A_X+L_X\),
and therefore \(R_X=\frac{\rank A_X+L_X}{n}\) follows.
Hence
\(
R_Q=R_X+R_Z-1=\frac{\rank A_X-\rank A_Z+L_X-L_Z}{n}
\)
as stated.
Finally, \(\Ker A_X\subseteq \Ker A_Z\) implies
\(\Ker A_X\cap\Ker B\subseteq \Ker A_Z\cap\Ker B\),
and hence \(L_X\le L_Z\).
\end{proof}

\subsection{Definition Using LDPC Matrices}

In this subsection, we specialize the general framework of Section~2.1 to a family defined by regular sparse matrices.

\begin{definition}[Nested regular sparse family]\label{def:family-jp}
The regular sparse specialization of the general framework in Definition~\ref{def:framework-jp} is defined as follows. All degrees below are socket degrees; on simple graphs they equal binary row or column weights.
Take positive integers \(k,j_Z,j_X,k_Z,k_\Delta\) such that \(1\le j_Z<j_X\), and set
\(j_\Delta:=j_X-j_Z\).

Let \(m_Z,m_\Delta,m_X\) denote the numbers of rows of \(A_Z,A_\Delta,A_X\), respectively.
First sample \(A_Z\in\F_2^{m_Z\times n}\) from the standard \((j_Z,k_Z)\)-regular LDPC ensemble~\cite[Secs.~3.3--3.4]{RU08}.
Next sample \(A_\Delta\in\F_2^{m_\Delta\times n}\) independently from the standard \((j_\Delta,k_\Delta)\)-regular LDPC ensemble, and define \(A_X:=\begin{bmatrix}A_Z\\ A_\Delta\end{bmatrix}\in\F_2^{m_X\times n}\) and \(m_X:=m_Z+m_\Delta\).

In a regular Tanner graph, the number of edges on the variable-node side must agree with that on the check-node side, so
\(j_Z n = k_Z m_Z\) and \(j_\Delta n = k_\Delta m_\Delta\)
must hold. In particular,
\(m_Z=\frac{j_Z}{k_Z}n\) and \(m_\Delta=\frac{j_\Delta}{k_\Delta}n\)
follow. Throughout we assume \(k_Z\mid j_Z n\) and \(k_\Delta\mid j_\Delta n\) so that these quantities are integers.

Finally, sample \(B\in\F_2^{n\times n}\) independently from the square \((k,k)\)-regular sparse ensemble~\cite[Secs.~3.3--3.4]{RU08}.
Since \(A_X=[A_Z;A_\Delta]\), the inclusion \(\Row(A_Z)\subseteq \Row(A_X)\) holds automatically, and Definition~\ref{def:framework-jp} therefore gives a nested CSS pair \((\CX,\CZ)\).
\end{definition}

From Definition~\ref{def:family-jp} onward, randomness enters through the independent configuration-model draws of \(A_Z\), \(A_\Delta\), and \(B\). These draws determine \(A_X=[A_Z;A_\Delta]\), the extended matrices \(H'_Z,H'_X\), and the codes \(\CZ,\CX\). They also determine compressed representatives \(H_Z,H_X\) and the rates \(R_Z,R_X,R_Q\), constituent distances \(d(\CZ),d(\CX)\), and relative distances \(d_Z^{\mathrm{rel}},d_X^{\mathrm{rel}}\). By contrast, the degree parameters \(j_Z,j_X,j_\Delta,k_Z,k_\Delta,k\), the blocklength \(n\), the row counts \(m_Z,m_\Delta,m_X\), and the design rates \(R_Z^{\mathrm{des}},R_X^{\mathrm{des}},R_Q^{\mathrm{des}}\) are deterministic. Whenever \(H_Z,H_X\) are used in a probabilistic context, they denote representatives chosen deterministically from the sampled row spaces; the probabilistic statements depend only on the underlying codes or row spaces, not on that representative choice.

From this point on, the stacked matrix \(A_X=[A_Z;A_\Delta]\) itself is treated as the object of analysis, without replacing it by a homogeneous regular ensemble.
In this paper, the ``standard \((j,k)\)-regular LDPC ensemble'' means the regular ensemble obtained from the socket-based configuration model of~\cite[Secs.~3.3--3.4]{RU08}, with optional conditioning on the simple-graph event when needed. Likewise, the ``square \((k,k)\)-regular ensemble'' means the square regular case of the same configuration model. That is, one assigns \(j\) (or \(k\)) sockets to each column-side node, \(k\) (or \(j\)) sockets to each row-side node, and then forms the Tanner graph by a uniformly random perfect matching. The coefficient-extraction formulas and pairing bounds below are derived first as exact statements for the unconditioned configuration model, which allows multiple edges. Its binary matrices record edge multiplicities modulo two. Thus prescribed degrees are socket degrees; after conditioning on a simple graph they equal binary row and column weights. On the other hand, for fixed degrees, the simple-graph event has probability bounded away from zero as \(n\to\infty\), so the \(o(1)\) first-moment bounds and negative exponential upper bounds proved in the unconditioned model transfer unchanged to the conditioned simple ensemble. Hence the asymptotic claims on distance and rate are justified for the \(\F_2\)-valued regular matrix ensembles used in the main text.

\subsection{Design Rates and Balanced Conditions}

\begin{definition}[Design rates]\label{def:rates-jp}
The design rates are the rates determined purely by the degrees, obtained from the actual-rate formulas in Proposition~\ref{prop:actual-rates-jp} by formally discarding the finite-length contributions of rank deficiency and kernel overlap. Namely, substitute \(\rank A_Z=m_Z\), \(\rank A_X=m_X\), and \(L_Z=L_X=0\) into the formulas, and define \(R_Z^{\mathrm{des}}:=(n-m_Z)/n\) and \(R_X^{\mathrm{des}}:=m_X/n\).
\end{definition}

For even \(k\), one has \(B\ones_{[n]}=0\), so \(B\) is singular
at every finite blocklength. For \(k\ge3\),
Proposition~\ref{prop:B-rank-jp} gives \(\dim\Ker B=o(n)\) in probability.
Together with the block-rank estimates under the hypotheses of
Theorem~\ref{thm:actual-rates-jp}, this yields convergence of the actual
rates to their design values in that regime.

\begin{proposition}[Formula for the design quantum rate]\label{prop:design-qrate-jp}
\(R_Q^{\mathrm{des}}=R_X^{\mathrm{des}}+R_Z^{\mathrm{des}}-1=(m_X-m_Z)/n=m_\Delta/n=j_\Delta/k_\Delta\) holds.
\end{proposition}

\begin{proof}
By Definitions~\ref{def:rates-jp} and~\ref{def:family-jp}, \(R_X^{\mathrm{des}}=m_X/n=(m_Z+m_\Delta)/n=j_Z/k_Z+j_\Delta/k_\Delta\) and \(R_Z^{\mathrm{des}}=(n-m_Z)/n=1-j_Z/k_Z\). Hence \(R_X^{\mathrm{des}}+R_Z^{\mathrm{des}}-1=m_X/n+(n-m_Z)/n-1=(m_X-m_Z)/n=m_\Delta/n=j_\Delta/k_\Delta\).
\end{proof}

The necessary and sufficient condition for the two classical design rates to coincide is \(R_X^{\mathrm{des}}=R_Z^{\mathrm{des}}\iff m_X/n=1-m_Z/n\iff m_X+m_Z=n\iff 2j_Z/k_Z+j_\Delta/k_\Delta=1\), and we call this the general balanced condition.

\begin{example}[An illustrative example of the general balanced condition]\label{ex:general-balanced-jp}
To illustrate the general balanced condition concretely, consider the example
\(n=40\), \(m_Z=15\), \(m_\Delta=10\), and \(m_X=25\),
with degrees \((j_Z,k_Z,j_\Delta,k_\Delta,k)=(3,8,2,8,2)\).
Then \(m_Z=(3/8)n\), \(m_\Delta=(2/8)n\), and \(m_X=(5/8)n\), so \(m_X+m_Z=n\) indeed holds.
On the other hand, since \(k_Z=k_\Delta=8\) whereas \(k=2\), this example does not belong to the homogeneous specialization introduced immediately below.
Accordingly, Figures~\ref{fig:random-hz-ext-jp} and~\ref{fig:random-hx-ext-jp} visualize a stacked block structure already permitted at the level of the general balanced condition.

\begin{figure}[t]
\centering
\includegraphics[width=\textwidth]{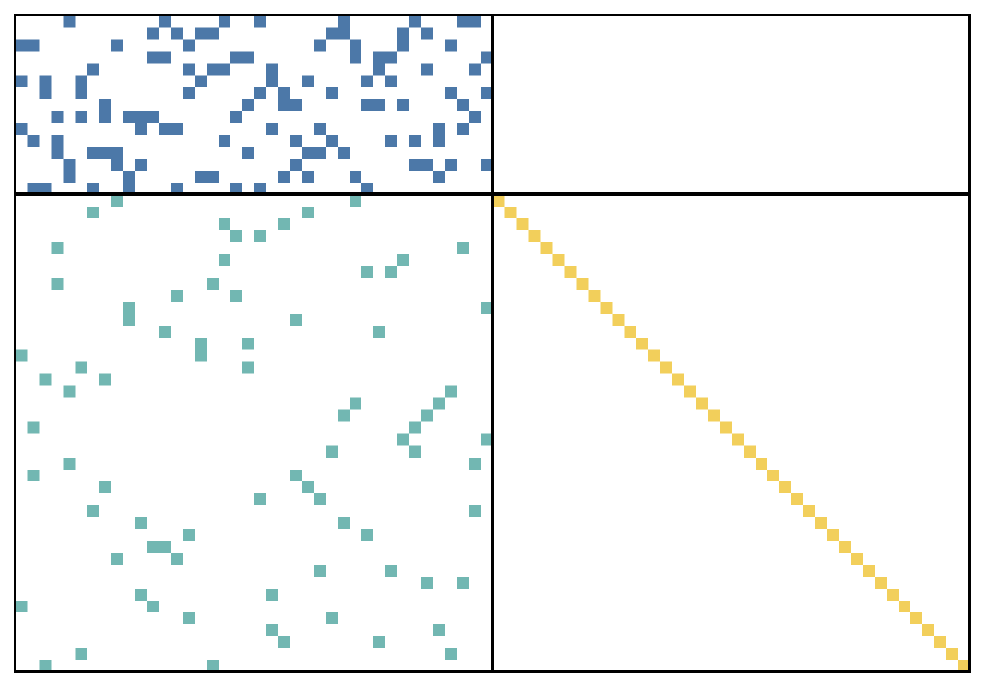}
\caption{Z-side extended parity-check matrix for the illustrative example \((j_Z,k_Z,j_\Delta,k_\Delta,k)=(3,8,2,8,2)\), \(n=40\), \(m_Z=15\), and \(m_\Delta=10\) (hence \(m_X=25\)):
\(
H_Z^\prime=
\begin{bmatrix}
A_Z & 0\\
B   & I_n
\end{bmatrix}
\in \F_2^{55\times 80}.
\)
The blue block in the upper left is the \((3,8)\)-regular matrix \(A_Z\in\F_2^{15\times 40}\), the light-blue block in the lower left is the square \((2,2)\)-regular sparse map \(B\in\F_2^{40\times 40}\), the yellow block in the lower right represents \(I_n\), and the upper-right block is zero. Black lines indicate block boundaries.}
\label{fig:random-hz-ext-jp}
\end{figure}

\begin{figure}[t]
\centering
\includegraphics[width=\textwidth]{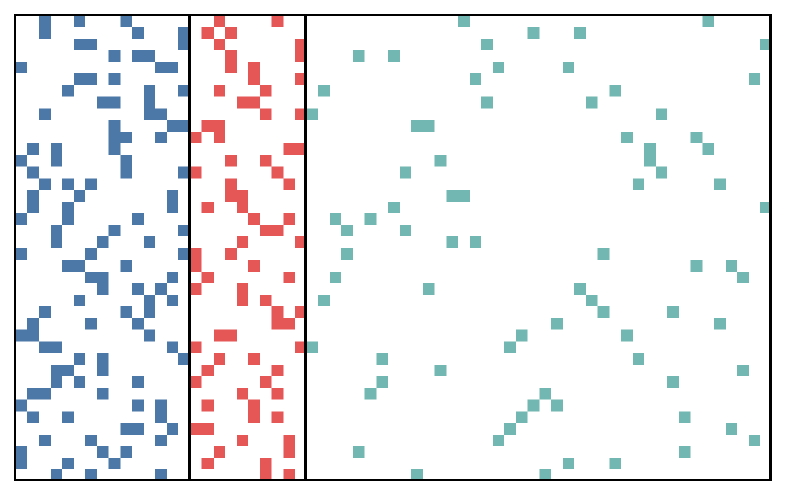}
\caption{X-side extended parity-check matrix for the same example:
\(H_X^\prime=[A_X^T\ B^T]=[A_Z^T\ A_\Delta^T\ B^T]\in \F_2^{40\times 65}\).
From left to right, the blue block is the transpose \(A_Z^T\in\F_2^{40\times 15}\) of the \((3,8)\)-regular block, the red block is the transpose \(A_\Delta^T\in\F_2^{40\times 10}\) of the \((2,8)\)-regular block, and the light-blue block is \(B^T\in\F_2^{40\times 40}\). Thus the stacked structure of \(A_X=[A_Z;A_\Delta]\in\F_2^{25\times 40}\) is visualized directly. Black lines indicate block boundaries.}
\label{fig:random-hx-ext-jp}
\end{figure}

For the same example, Figures~\ref{fig:random-hz-cmpr-jp} and~\ref{fig:random-hx-cmpr-jp} show compressed parity-check matrices \(H_Z,H_X\) obtained without taking reduced row echelon form (RREF).
Here \(H_Z\) is obtained by projecting a basis of the kernel of \([A_Z^T\ B^T]\) to the visible component, and \(H_X\) is defined from a basis matrix \(K_X\) of \(\Ker(A_X)\) by
\(H_X=K_XB^T\).
No final row-reduction by elementary row operations is applied to either visible matrix.
For this generated example, we also verified directly over \(\F_2\) that \(H_XH_Z^T=0\) holds.
The design rates are \(R_Z^{\mathrm{des}}=(n-m_Z)/n=25/40=0.625\), \(R_X^{\mathrm{des}}=m_X/n=25/40=0.625\), and \(R_Q^{\mathrm{des}}=m_\Delta/n=10/40=0.25\). On the other hand, for the generated matrices we have \(\rank H_Z=16\) and \(\rank H_X=15\), so \(R_Z=\dim \CZ/n=(40-16)/40=24/40=0.6\), \(R_X=\dim \CX/n=(40-15)/40=25/40=0.625\), and \(R_Q=(\dim\CX+\dim\CZ-n)/n=(25+24-40)/40=9/40=0.225\).
Here \(\rank A_Z=15\), \(\rank A_X=24\), and
\(L_Z=L_X=1\), which account for these rates through
Proposition~\ref{prop:actual-rates-jp}. The accompanying matrix files and
\path{verify_published_example.py} reproduce the binary-rank and
orthogonality checks.

\begin{figure}[t]
\centering
\includegraphics[width=0.50\textwidth]{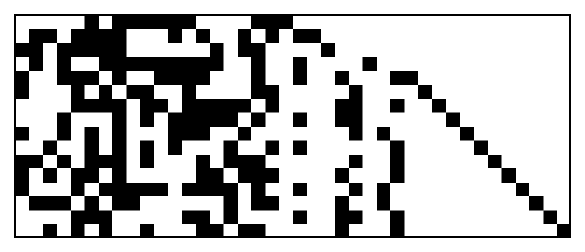}
\caption{Z-side compressed parity-check matrix for the same example:
\(
H_Z\in\F_2^{16\times 40}.
\)
It is obtained by projecting a basis of the kernel of \([A_Z^T\ B^T]\) to the visible component, without taking RREF on the final visible matrix. Thus each row represents an explicit generator of \(\CZ^\perp\).}
\label{fig:random-hz-cmpr-jp}
\end{figure}

\begin{figure}[t]
\centering
\includegraphics[width=0.615\textwidth]{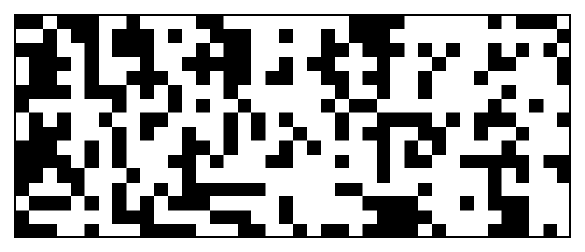}
\caption{X-side compressed parity-check matrix for the same example:
\(
H_X\in\F_2^{16\times 40}.
\)
If \(K_X\) is a basis matrix of \(\Ker(A_X)\), then this figure displays \(H_X=K_XB^T\) directly, again without taking RREF on the final visible matrix. Thus each row represents an explicit generator of \(\CX^\perp\).}
\label{fig:random-hx-cmpr-jp}
\end{figure}
\end{example}

In the distance analysis below, we assume the homogeneous specialization \(k_Z=k_\Delta=k\).
This assumption is made in order to simplify the coefficient-extraction formulas and symmetry statements used in the fixed-degree analysis from Section 3 onward and in the finite-degree GV theorems for explicit triples.
That is, the first half of Section 2 defines the stacked family with general \((j_Z,k_Z)\) and \((j_\Delta,k_\Delta)\)-regular blocks, and from this point on we restrict attention to the homogeneous subclass in which the row degree is also common.
Since \(j_\Delta=j_X-j_Z\), the balanced condition then reduces to \(j_X+j_Z=k\),
and from now on we assume this homogeneous balanced condition and call any such parameter triple \((j_Z,j_X,k)\) a balanced triple. In this case
\(R_Z^{\mathrm{des}}=R_X^{\mathrm{des}}=\frac{j_X}{k}=1-\frac{j_Z}{k}\)
and
\(R_Q^{\mathrm{des}}=\frac{j_\Delta}{k}=1-\frac{2j_Z}{k}>0\)
hold.
Note that the lower bound \(4\le j_Z\) is not needed to define the nested CSS pair itself; it is imposed only later in the fixed-degree distance analysis.

The next theorem states that the design rates introduced here are not merely formal proxies: they asymptotically describe the actual rates of the finite-length codes.
Its significance is that the distance analysis below may therefore use the ratios determined by the design rates without losing contact with the asymptotic parameters of the actual code family.
For the following rate theorem, assume the homogeneous specialization \(k_Z=k_\Delta=k\), the balanced relation \(j_X=k-j_Z\), \(4\le j_Z<k/2\), and even \(k\). Both parities of \(j_Z\) are included; \(j_\Delta=k-2j_Z\) is always even.

The lemmas and proofs needed for the next theorem are deferred to Appendix~\ref{app:actual-rates-jp}.

\begin{theorem}[The actual rates converge in probability to the design rates]\label{thm:actual-rates-jp}
Fix a homogeneous balanced triple with \(k_Z=k_\Delta=k\),
\(j_X=k-j_Z\), \(4\le j_Z<k/2\), and even \(k\).
For the ensemble in Definition~\ref{def:family-jp},
\(R_Z \to R_Z^{\mathrm{des}}=\frac{j_X}{k}\), \(R_X \to R_X^{\mathrm{des}}=\frac{j_X}{k}\), and \(R_Q \to R_Q^{\mathrm{des}}=\frac{j_X-j_Z}{k}\) hold as convergence in probability when \(n\to\infty\).
\end{theorem}

Accordingly, in the distance analysis below, the natural reference quantities are the ratios determined by the design rates, rather than finite-length rank deficiencies.

\subsection{Sparse-Matrix Representation of Syndrome Consistency}

In this subsection, we use the sparse extended matrices \(H_Z'\) and \(H_X'\) introduced in Section~2.2 to rewrite the decoding conditions for the compressed syndromes as affine sparse systems. This also makes explicit that the present family has bounded graphical complexity and admits sum-product (BP) updates with linear per-iteration complexity.

For the nested regular sparse family, the binary Tanner graphs of \(H_Z'\) and \(H_X'\) introduced in Definition~\ref{def:framework-jp} have at most \(j_Zn+kn+n\) and \(j_Xn+kn\) edges, respectively. The bounds are equalities on simple graphs. If \(j_Z,j_X,k_Z,k_\Delta,k\) are fixed, both are \(O(n)\). Therefore the present family has bounded graphical complexity in the sense of~\cite{HA05}.

Although the compressed-syndrome matrices \(H_Z,H_X\) are generally dense, syndrome decoding asks us to estimate noise vectors \(\be_X,\be_Z\) from the measured syndromes \(\bs_Z,\bs_X\) under the compressed syndrome equations \(H_Z\be_X=\bs_Z\) and \(H_X\be_Z=\bs_X\). Once a syndrome representative is injected, this decoding scenario can still be represented by affine sparse systems that retain the sparse structure of \(H_Z'\) and \(H_X'\).

In what follows, let \(H_Z\) be a full-row-rank compressed parity-check matrix in the sense of Definition~\ref{def:compressed-pcm-jp}, and let \(K_X\) be a full-row-rank basis matrix of \(\Ker A_X\). On the Z side, fix any right inverse \(\Gamma_Z\) satisfying \(H_Z\Gamma_Z=I\). On the X side, take the compressed representative \(H_X:=K_XB^T\) and fix any right inverse \(\Gamma_X\) satisfying \(K_X\Gamma_X=I\). Write \(\bt_Z:=\Gamma_Z\bs_Z\) and \(\bt_X:=\Gamma_X\bs_X\).

\begin{theorem}[Sparse affine representations of the syndrome equations]\label{prop:lifted-bp-jp}
The syndrome equations \(H_Z\be_X=\bs_Z\) and \(H_X\be_Z=\bs_X\) hold if and only if there exist \(\bff_X\in\F_2^n\) on the Z side and \(\bff_Z\in\F_2^{m_X}\) on the X side such that
\[
\begin{bmatrix}
A_Z & 0\\
B & I_n
\end{bmatrix}
\begin{bmatrix}
\bff_X\\
\be_X
\end{bmatrix}
=
\begin{bmatrix}
\mathbf 0\\
\bt_Z
\end{bmatrix}
\]
and
\[
\begin{bmatrix}
A_X^T & B^T
\end{bmatrix}
\begin{bmatrix}
\bff_Z\\
\be_Z
\end{bmatrix}
=
\bt_X.
\]
Equivalently, \(A_Z\bff_X=0\), \(\be_X=\bt_Z+B\bff_X\), and \(B^T\be_Z=\bt_X+A_X^T\bff_Z\) hold simultaneously.
\end{theorem}

\begin{proof}
On the Z side, \(A_Z\bff_X=0\) and \(\be_X=\bt_Z+B\bff_X\) imply \(B\bff_X\in B(\Ker A_Z)=\CZ\), and \(\Row(H_Z)=\CZ^\perp\) therefore gives \(H_Z(B\bff_X)=0\). Hence \(H_Z\be_X=H_Z(\bt_Z+B\bff_X)=H_Z\bt_Z=H_Z\Gamma_Z\bs_Z=\bs_Z\). Conversely, \(H_Z\be_X=\bs_Z\) implies \(H_Z(\be_X-\bt_Z)=0\), so \(\be_X-\bt_Z\in\Ker H_Z=\CZ=B(\Ker A_Z)\), and thus there exists \(\bff_X\in\Ker A_Z\) such that \(\be_X=\bt_Z+B\bff_X\).

On the X side, \(B^T\be_Z=\bt_X+A_X^T\bff_Z\) implies \(K_XB^T\be_Z=K_X\bt_X=K_X\Gamma_X\bs_X=\bs_X\) because \(K_XA_X^T=0\). Since \(H_X:=K_XB^T\), this is exactly \(H_X\be_Z=\bs_X\). Conversely, \(H_X\be_Z=\bs_X\) implies \(K_X(B^T\be_Z-\bt_X)=0\). The row space of \(K_X\) is \(\Ker A_X\), so its right kernel equals \(\Row(A_X)\). Therefore there exists \(\bff_Z\in\F_2^{m_X}\) such that \(B^T\be_Z-\bt_X=A_X^T\bff_Z\), namely \(B^T\be_Z=\bt_X+A_X^T\bff_Z\).
\end{proof}

The sparse matrices appearing here are exactly \(\begin{bmatrix}A_Z&0\\ B&I_n\end{bmatrix}=H_Z'\) and \(\begin{bmatrix}A_X^T&B^T\end{bmatrix}=H_X'\). Hence, in the nested regular sparse family, their numbers of nonzero entries are at most \(j_Zn+kn+n\) and \(j_Xn+kn\), respectively, and running BP on the corresponding Tanner graphs requires \(O(n)\) operations per iteration.

However, the sparse object here is the matrix representation of the affine systems after the representatives \(\bt_Z,\bt_X\) have been injected externally. The step that forms \(\bt_Z,\bt_X\) from the compressed syndromes is still generally dense. For simple realizations in the certified degree range, every check node in these Tanner graphs is adjacent to multiple punctured nodes \(\bff_X\) or \(\bff_Z\), so fixing the syndrome bits does not by itself produce nontrivial BP messages on the visible variables \(\be_X,\be_Z\) at the initial update. Thus this theorem should be read as giving a sparse affine representation for a fixed syndrome representative, not a sparse measurement matrix for the compressed syndromes themselves or a directly usable BP decoder.

The same observation also clarifies the role of girth. Bounded degree of the extended matrices \(H'_Z,H'_X\) does not by itself imply any large-girth theorem, and no such theorem is proved in this paper. The point is rather that the three blocks \(A_Z\), \(A_\Delta\), and \(B\) are chosen independently, so finite-length design retains genuine freedom to impose additional cycle constraints on the extended Tanner graphs. This is the precise sense in which the present construction leaves room for large-girth design via general sparse-graph and Tanner-graph generation methods~\cite{BayatiMS09,BayatiKMOS09}, subject of course to the logarithmic-growth limitations implied by the Moore bound~\cite{AlonHL02}; by contrast, structured CSS-LDPC subclasses can exhibit explicit short-cycle obstructions~\cite{AmirzadePS24}.

\section{Distance Analysis on the Hsu--Anastasopoulos Side}

This section states the distance results on the HA side.
The linear minimum-distance property of regular classical LDPC ensembles goes back to Gallager~\cite{Gal63}, and classical distance analysis for HA codes already appears in~\cite{HA05}.
For fixed degrees, a positive linear distance follows from a low-output-weight argument.
All limits are taken along blocklengths satisfying the stated integrality conditions. Unless stated otherwise, every probability and convergence statement in Sections~3--5 is taken with respect to the random draw of \(A_Z,A_\Delta,B\) in Definition~\ref{def:family-jp}.
Although in Section~2 we did not define \(\CZ=B(\Ker A_Z)\) as a concatenated code,
the analysis in this section views it in that way:
\(\Ker A_Z\) plays the role of the outer code, and the map \(\bu\mapsto B\bu\)
is the inner regular sparse map.
That is, we first choose an outer codeword \(\bu\in\Ker A_Z\), and then obtain
\(\bv=B\bu\) through the inner map \(B\).
The proof throughout this section follows the HA-side analysis in~\cite{HA05}.
However, the ensemble itself is not the same: in~\cite{HA05} the outer code is a Gallager ensemble,
whereas here both the inner and outer parts are modeled by the socket-based configuration model of~\cite[Secs.~3.3--3.4]{RU08}.
For this reason, we state only the results in the main text, and move to Appendix~\ref{app:HA-proofs-jp}
the transition enumerator, the first-moment estimates, and the proof of the fixed-degree theorem,
namely the parts of the argument in~\cite{HA05} that must be rewritten for the present ensemble.

\begin{theorem}[Fixed-degree positive linear distance on the HA side]\label{thm:HA-lin-jp}
Assume a fixed even balanced triple satisfying
\(4\le j_Z<\frac{k}{2}\), \(j_Z\equiv 0\pmod 2\), and \(j_\Delta:=j_X-j_Z\equiv 0\pmod 2\).
Then there exists a constant \(\delta_Z^{\mathrm{lin}}>0\), depending only on \((j_Z,k)\), such that for every
\(0<\delta<\delta_Z^{\mathrm{lin}}\),
\(\PP[d(\CZ)\le \delta n]\to 0 \qquad (n\to\infty)\)
holds.
In particular, \(d(\CZ)=\Omega(n)\) holds with high probability.
\end{theorem}

Write \(\alpha_Z=j_Z/k\), \(\alpha_X=j_X/k=1-\alpha_Z\), and \(\alpha_\Delta=j_\Delta/k\). Here the classical GV distance for binary linear codes is the existence curve
\(R \ge 1-\h(\delta)\),
or equivalently
\(\delta_{\mathrm{GV}}(R)=\h^{-1}(1-R)\),
as given in~\cite{Gil52,Var57}.
The design rate of the HA-side constituent code \(\CZ\) considered here is
 \(R_Z^{\mathrm{des}}=\alpha_X=1-\alpha_Z\),
so the corresponding classical GV point is
\(\delta_{\mathrm{GV}}(R_Z^{\mathrm{des}})=\h^{-1}(1-\alpha_X)=\h^{-1}(\alpha_Z)\).
Let the search window be
\[
\mathcal T_{\mathrm{scan}}
:=
\{(j_Z,j_X,k): k\le 30,\ j_Z\le 10,\ j_Z<k/2,\ j_X=k-j_Z\}.
\]
Here we restrict attention to this low-degree region in order to examine finite-degree GV attainment while keeping the degrees genuinely small.
This defines the comparison window. The theorems certify the subset \(\mathcal T_{\mathrm{GV}}\) below, which is displayed in Fig.~\ref{fig:finite-gv-jp}.
Let
\[
\mathcal T_{\mathrm{GV}}
:=
\{(j_Z,j_X,k)\in \mathcal T_{\mathrm{scan}}: j_Z\ge 4,\ k\equiv 0\!\!\!\pmod 2\}
\subset \mathcal T_{\mathrm{scan}}
\]
denote the 56 balanced triples proved to attain GV in Appendices~\ref{app:finite-gv-ha-jp} and~\ref{app:finite-gv-mn-jp}. Here evenness is imposed on \(k\), not on \(j_Z\). Figure~\ref{fig:finite-gv-jp} displays the complete certified set.

\begin{theorem}[Finite-degree GV attainment on the HA side]\label{thm:HA-gv-jp}
For each \((j_Z,j_X,k)\in\mathcal T_{\mathrm{GV}}\), let \(\delta_{\mathrm{GV}}:=\delta_{\mathrm{GV}}(R_Z^{\mathrm{des}})=\h^{-1}(1-R_Z^{\mathrm{des}})\). Then for every \(0<\delta<\delta_{\mathrm{GV}}\),
\[
\PP[d(\CZ)\le \delta n]\to 0
\qquad (n\to\infty)
\]
holds. In this sense, every triple in \(\mathcal T_{\mathrm{GV}}\) attains the classical GV point on the HA side at finite degree.
\end{theorem}

\begin{proof}
The proof is deferred to Appendix~\ref{app:finite-gv-ha-jp}.
\end{proof}

We also stress that~\cite[Theorem~2 and Corollary~1]{HA05} does not prove a finite-degree GV statement of the form above for each fixed triple. More precisely, its result is an \(\varepsilon\)--\(K\) large-degree statement: for any \(\varepsilon>0\), there exists a sufficiently large integer \(K(\varepsilon)\) such that whenever \(k\ge K(\varepsilon)\), the normalized minimum distance can be made within \(\varepsilon\) of the GV bound. Equivalently,~\cite{HA05} shows approach to the GV bound in the limit \(k\to\infty\), rather than finite-degree GV attainment for each fixed triple as in the theorem above.

\section{Distance Analysis on the MacKay--Neal Side}

We now turn to the MN side.
The important point is that the actual model is
\(A_X=\begin{bmatrix}A_Z\\ A_\Delta\end{bmatrix}\),
namely a stacked ensemble, not a standard \((j_X,k)\)-regular ensemble.
Weight-distribution and spectral-shape analyses for irregular and multi-edge ensembles were developed in~\cite{DIRU06,KasaiMET10}. On the MN side, related analyses include the SC-MN weight-enumerator analysis of~\cite{Mitchell12} and the protograph-based MN input-output enumerator analysis of~\cite{Zahr25}, but none of them studies the stacked ensemble treated here itself.
Accordingly, one must build, on this stacked ensemble itself, the stacked refined enumerator, the exact configuration formula, the blockwise complement symmetry coming from even degrees, the trial-point bounds governing the linear-weight regime, and the pairing bound governing the low-weight regime. We defer these details, as well as the proof of the fixed-degree theorem, to Appendix~\ref{app:MN-proofs-jp}, and state only the two resulting theorems in this section.

\begin{theorem}[Fixed-degree positive linear distance on the MN side]\label{thm:MN-lin-jp}
Assume a fixed finite balanced triple satisfying \(4\le j_Z<\frac{k}{2}\), \(j_Z\equiv 0\pmod 2\), and \(j_\Delta:=j_X-j_Z\equiv 0\pmod 2\). Then there exists a constant \(\delta_X^{\mathrm{lin}}>0\) such that for every \(0<\delta<\delta_X^{\mathrm{lin}}\), one has \(\PP[d(\CX)\le \delta n]\to 0\) as \(n\to\infty\).
In particular, \(d(\CX)=\Omega(n)\) holds with high probability.
\end{theorem}

\begin{proof}
The proof is deferred to Appendix~\ref{app:MN-proofs-jp}.
\end{proof}

The finite-degree GV proof keeps the same low-weight / linear-weight decomposition. Appendix~\ref{app:finite-gv-mn-jp} reduces the remaining inequalities analytically to finitely many scalar comparisons, checked with exact rational logarithm bounds.

\begin{theorem}[Finite-degree GV attainment on the MN side]\label{thm:MN-gv-jp}
For each \((j_Z,j_X,k)\in\mathcal T_{\mathrm{GV}}\), define \(\delta_{\mathrm{GV}}:=\delta_{\mathrm{GV}}(R_X^{\mathrm{des}})=\h^{-1}(1-R_X^{\mathrm{des}})\). Then for every \(0<\delta<\delta_{\mathrm{GV}}\), \(\PP[d(\CX)\le \delta n]\to 0\) as \(n\to\infty\). In this sense, every triple in \(\mathcal T_{\mathrm{GV}}\) attains the classical GV point on the MN side at finite degree.
\end{theorem}

\begin{proof}
The proof is deferred to Appendix~\ref{app:finite-gv-mn-jp}.
\end{proof}

\section{Relative Distance of the Nested CSS Pair}

Define the relative distances of the nested CSS pair \((\CX,\CZ)\), with the convention that the value is \(+\infty\) when the corresponding difference set is empty, by \(d_Z^{\mathrm{rel}}:=\min\{\wt(\bv): \bv\in \CZ\setminus C_Z(A_X)\}\) and \(d_X^{\mathrm{rel}}:=\min\{\wt(\bv): \bv\in \CX\setminus \CZ^\perp\}\).
 
\begin{corollary}[Relative linear distance for fixed even degrees]\label{cor:lin-jp}
Assume a fixed even balanced triple satisfying \(4\le j_Z<\frac{k}{2}\), \(j_Z\equiv 0\pmod 2\), and \(j_X-j_Z\equiv 0\pmod 2\). Then the nested CSS pair has relative linear distance. That is, there exist \(\delta_Z^{\mathrm{lin}},\delta_X^{\mathrm{lin}}>0\) such that for every \(0<\delta_Z<\delta_Z^{\mathrm{lin}}\) and \(0<\delta_X<\delta_X^{\mathrm{lin}}\), one has \(\PP[d_Z^{\mathrm{rel}}\le \delta_Z n]\to 0\) and \(\PP[d_X^{\mathrm{rel}}\le \delta_X n]\to 0\).
\end{corollary}

\begin{proof}
From \(C_Z(A_X)\subseteq \CZ\) and \(\CZ^\perp\subseteq \CX\), we obtain \(\CZ\setminus C_Z(A_X)\subseteq \CZ\setminus\{0\}\) and \(\CX\setminus \CZ^\perp\subseteq \CX\setminus\{0\}\). Hence \(d_Z^{\mathrm{rel}}\ge d(\CZ)\) and \(d_X^{\mathrm{rel}}\ge d(\CX)\), and Theorems~\ref{thm:HA-lin-jp} and~\ref{thm:MN-lin-jp} finish the proof.
\end{proof}

Under the balanced condition, the quantum design rate is \(R_Q^{\mathrm{des}}=1-2\alpha_Z\).
On the other hand, the asymptotic existence bound corresponding to
the CSS existence results of Calderbank--Shor and Steane is
 \(R_Q \ge 1-2\h(\delta)\), or equivalently \(\delta_{\mathrm{CSS\text{-}GV}}(R_Q)=\h^{-1}\!\left(\frac{1-R_Q}{2}\right)\), as given in~\cite{CS96,Ste96}.
Therefore, for a finite triple \((j_Z,j_X,k)\), the target value is \(\delta_{\mathrm{CSS\text{-}GV}}(R_Q^{\mathrm{des}})=\h^{-1}\!\left(\frac{1-(1-2j_Z/k)}{2}\right)=\h^{-1}\!\left(\frac{j_Z}{k}\right)\).
Thus the two theorems above imply that this value is attained already at finite degree for every triple in \(\mathcal T_{\mathrm{GV}}\).

\begin{corollary}[Finite-degree CSS GV attainment]\label{cor:finite-gv-jp}
For each \((j_Z,j_X,k)\in\mathcal T_{\mathrm{GV}}\), let \(\delta_{\mathrm{GV}}:=\delta_{\mathrm{CSS\text{-}GV}}(R_Q^{\mathrm{des}})=\h^{-1}\!\left(\frac{1-R_Q^{\mathrm{des}}}{2}\right)\). Then for every \(0<\delta<\delta_{\mathrm{GV}}\), \(\PP[d_Z^{\mathrm{rel}}\le \delta n]\to 0\) and \(\PP[d_X^{\mathrm{rel}}\le \delta n]\to 0\) \((n\to\infty)\) hold. In this sense, every triple in \(\mathcal T_{\mathrm{GV}}\) attains the CSS GV distance at finite degree.
\end{corollary}

\begin{proof}
This follows immediately from Theorems~\ref{thm:HA-gv-jp} and~\ref{thm:MN-gv-jp} together with \(d_Z^{\mathrm{rel}}\ge d(\CZ)\) and \(d_X^{\mathrm{rel}}\ge d(\CX)\).
\end{proof}

\section{Certified Balanced Regular Triples}

The set \(\mathcal T_{\mathrm{GV}}\) contains all 56 balanced triples
with even \(k\), \(10\le k\le30\), \(4\le j_Z\le10\), and
\(j_Z<k/2\). Of these, 32 have even \(j_Z\) and 24 have odd \(j_Z\).
For each triple, Theorem~\ref{thm:actual-rates-jp} and
Corollary~\ref{cor:finite-gv-jp} give the limiting rate and guaranteed
relative-distance threshold
\[
R_Q\longrightarrow 1-\frac{2j_Z}{k},\qquad
\delta_{\mathrm{GV}}=\h^{-1}\!\left(\frac{j_Z}{k}\right).
\]
The distance statement holds for every fixed
\(0<\delta<\delta_{\mathrm{GV}}\). It is a lower-bound assertion;
it does not identify the exact asymptotic minimum distance or exclude
GV attainment by parameters outside the certified set.

\begin{figure}[ht]
\centering
\includegraphics[width=\textwidth]{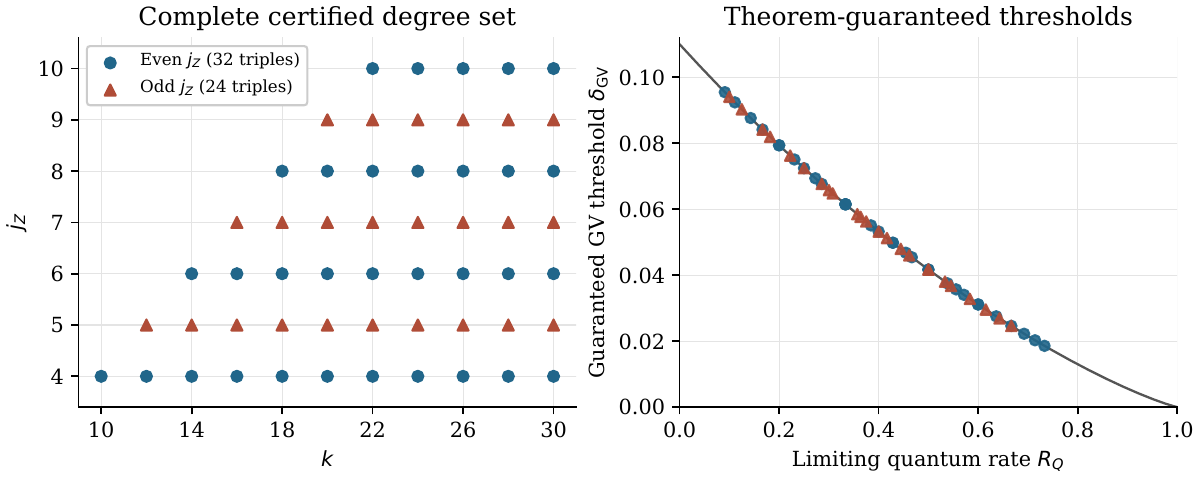}
\caption{The 56 triples covered by the finite-degree GV theorems.
Left: the complete degree set, with \(j_X=k-j_Z\).
Right: the theorem-guaranteed threshold
\(\delta_{\mathrm{GV}}=\h^{-1}(j_Z/k)\) plotted against the limiting
quantum rate \(1-2j_Z/k\). Distinct triples with the same ratio
\(j_Z/k\) overlap in the right panel. The markers distinguish the
32 even-\(j_Z\) and 24 odd-\(j_Z\) triples. The curve and marker heights
represent the GV threshold, rather than measured minimum distances or
zeros of numerically optimized enumerators.}
\label{fig:finite-gv-jp}
\end{figure}

The matrix data, degree tables, and verification programs are available
in the GitHub repository~\cite{KasaiGVData26}.
The file \path{certified_triples.csv} lists all 56 triples.
The script \path{plot_certified_region.py} generates the figure directly
from their integer defining conditions. Floating-point inversion of
binary entropy is used only to position the plotted markers; the
proofs in Appendices~\ref{app:finite-gv-ha-jp} and
\ref{app:finite-gv-mn-jp} use the separate exact-rational verifier
\path{verify_analytic_scalar_bounds.py}.

We also note that the boundary \(j_X=j_Z=k/2\), where the design quantum rate is zero, lies outside the scope of the main theorem. Nevertheless, if one looks only at the minimum distances of the classical constituent codes, then a separate finite-degree GV certification can still be carried out on this boundary. This is formalized in the next proposition. The triple \((3,3,6)\) is outside this certification because the one-block MN low-weight small-support estimate used here does not establish it. Since \(A_X=A_Z\) along this boundary, \(C_X=C_Z^\perp\) and the actual quantum dimension is exactly zero at every blocklength. Thus the proposition below should be read not as a nontrivial quantum relative-distance statement, but as a finite-degree classical constituent example. The accompanying script \path{verify_boundary_bounds.py} checks all 12 boundary triples and records exact rational interval endpoints and upper bounds.

\begin{proposition}[Classical GV certification on the zero-quantum-rate boundary]\label{prop:zero-rate-boundary-jp}
\[
\mathcal T_{\mathrm{GV}}^{(0)}
:=
\{(j_Z,j_X,k)=(j,j,2j): 4\le j\le 15\}
\]
and, for each \((j_Z,j_X,k)\in\mathcal T_{\mathrm{GV}}^{(0)}\), let \(\delta_{\mathrm{GV}}:=\h^{-1}(1/2)\). Then for every \(0<\delta<\delta_{\mathrm{GV}}\),
\[
\PP[d(\CZ)\le \delta n]\to 0,
\qquad
\PP[d(\CX)\le \delta n]\to 0
\qquad (n\to\infty).
\]
Moreover, \(R_Z\to1/2\) and \(R_X\to1/2\) in probability. Thus every triple in \(\mathcal T_{\mathrm{GV}}^{(0)}\) attains the classical GV point on both the HA side and the one-block MN side already at finite degree.
\end{proposition}

\begin{proof}
The proof is deferred to Appendix~\ref{app:zero-rate-boundary-jp}.
\end{proof}

We also note why the present finite-degree certification remains concentrated on even \(k\). On the HA side, Appendix~\ref{app:finite-gv-ha-jp} treats the small-input range \(1\le s\le \beta_Z n/k\) by a pairing bound and then uses the complement symmetry \(N_o(s)=N_o(n-s)\), valid for even \(k\), to control the high-input tail \(n-\beta_Z n/k\le s\le n-1\). For odd \(k\), however, the complement map no longer preserves the kernel and instead moves it to the all-one-syndrome coset. Thus the same reduction is unavailable as it stands. Extending the certification to odd \(k\) would therefore require an odd-syndrome coset enumerator together with a direct estimate of the corresponding high-input tail.

\begin{conjecture}[Balanced triples attaining finite-degree GV]\label{conj:balanced-gv-jp}
Let \((j_Z,j_X,k)\) be any balanced triple satisfying \(4\le j_Z<\frac{k}{2}\) and \(j_X=k-j_Z\).
Then the corresponding HA-side and MN-side classical constituent codes attain GV distance already at fixed finite degree, and hence the associated nested CSS family attains the CSS GV distance.
\end{conjecture}

\section{Conclusion}

We constructed a balanced nested regular CSS family entirely from regular LDPC matrices.
This family simultaneously has bounded graphical complexity, equal classical design rates,
and a positive design quantum rate.
The most important point is that the MN-side parity-check matrix is the stacked matrix \(A_X=\begin{bmatrix}A_Z\\ A_\Delta\end{bmatrix}\), and the proof must therefore be carried out on this stacked ensemble itself.

The final conclusions are twofold.
For fixed balanced triples with \(j_Z\ge4\) and even component degrees, the nested CSS pair has relative distance linear in \(n\) with probability tending to \(1\).
Moreover, within the search window \(\mathcal T_{\mathrm{scan}}\), the classical GV point, and hence by Corollary~\ref{cor:finite-gv-jp} the CSS GV distance, is rigorously certified already at finite degree for the whole set \(\mathcal T_{\mathrm{GV}}\) of 56 balanced triples with even \(k\) and \(j_Z\ge4\), including all 24 odd-\(j_Z\) cases.
Moreover, the actual classical and quantum rates converge in probability to the corresponding design rates.
Thus the balanced nested regular MN/HA construction simultaneously realizes bounded graphical complexity,
a positive design quantum rate, fixed-degree relative linear distance, and explicit finite-degree GV-certified examples.

Another important point is that in this family the three blocks \(A_Z\), \(A_\Delta\), and \(B\)
can be chosen independently, so finite-length design retains the freedom to enlarge the girth of the extended parity-check matrices.
This is significant because large-girth design for fixed-degree sparse graphs and Tanner graphs is a separate finite-length task, with sequential construction methods available~\cite{BayatiMS09,BayatiKMOS09} but only logarithmic girth growth allowed by Moore-bound constraints~\cite{AlonHL02}. The present nested construction is compatible with importing such methods, unlike more rigid CSS-LDPC subclasses where explicit short-cycle obstructions are known~\cite{AmirzadePS24}.

On the other hand, the uncoupled nested construction studied here is primarily an asymptotic guarantee on distance and rate,
and is not itself optimized as a direct target for standard BP decoding.
In fact, the compressed parity-check matrices \(H_Z,H_X\) that provide the visible-variable syndromes are generally dense,
so syndrome measurement is not yet handled as a sparse graph.
However, as seen in Section~2.4, once syndrome representatives \(\bt_Z,\bt_X\) are injected externally, the compressed syndrome equations themselves can be written as sparse affine systems in terms of the sparse extended matrices \(H_Z',H_X'\). What remains unresolved is that the step of forming \(\bt_Z,\bt_X\) from the compressed syndromes is generally dense, and that, on simple realizations in the certified degree range, each check node is adjacent to multiple punctured nodes, so this sparse representation does not by itself yield a directly usable BP rule of the kind employed for classical MN/HA codes. Accordingly, this paper still does not propose an efficient decoding method.

However, on the classical side, bounded-degree spatial coupling achieves capacity on the BEC for multi-edge type LDPC ensembles~\cite{ObataJKP13},
and for the SC-MN / SC-HA family it is known to improve the BP threshold on the BEC~\cite{KasaiSak11},
to preserve distance growth~\cite{Mitchell12},
and even to achieve universal symmetric information rate on generalized erasure channels with memory~\cite{FukushimaOK15}.

On the quantum side as well, spatially coupled quasi-cyclic quantum LDPC-CSS codes were proposed early on in~\cite{HagiwaraKIS11},
and there are now results showing sparse-structure decoding performance approaching
the coding-theoretical bound~\cite{KomotoKasai25}.
This is a plausible direction rather than a consequence of the present analysis.
It is therefore natural to expect that introducing a spatially coupled version of the nested family studied here
could alleviate the issues above, just as in the classical case, and possibly lead to BP-type decoding rules with near-optimal decoding performance.

\appendix
\renewcommand{\thesection}{\Alph{section}}

\engappendixsection{HA-Side Enumerators and Linear-Distance Proof}{app:HA-proofs-jp}

For integers \(N\ge1\), we use the endpoint-uniform binomial bounds
\begin{equation}\label{eq:binomial-uniform-jp}
\frac{2^{N\h(q/N)}}{N+1}\le\binom Nq\le2^{N\h(q/N)}
\qquad(0\le q\le N).
\end{equation}
For \(0<q<N\), the binomial probability at \(q\), with success
probability \(q/N\), is a maximum among its \(N+1\) values, so it is
at least \(1/(N+1)\) and at most \(1\). This proves both bounds;
the endpoints are immediate. We use \(\log_2 0=-\infty\) for a zero
enumerator and the convention \(0\log 0=0\).

For \(A_Z\) sampled from the unconditioned \((j_Z,k)\)-regular configuration-model ensemble of Definition~\ref{def:family-jp} with \(k_Z=k\), let \(m_Z:=j_Zn/k\) and define the average weight distribution of the outer code \(\Ker A_Z\) by
\begin{equation*}
\begin{split}
N_o(s)
&:=\EE\bigl|\{\bu\in\Ker A_Z:\wt(\bu)=s\}\bigr| \\
&=
\binom{n}{s}
\frac{\left[x^{j_Z s}\right]\left(\frac{(1+x)^k+(1-x)^k}{2}\right)^{m_Z}}
{\binom{n j_Z}{j_Z s}}.
\end{split}
\end{equation*}
For a fixed support \(U\subset[n]\) with \(|U|=s\), the \(j_Zs\) sockets attached to \(U\) map to a uniformly random \(j_Zs\)-subset of the \(n j_Z=km_Z\) row-side sockets in the configuration model.
The condition \(A_Z\ones_U=0\) is exactly that every row receives an even number of these sockets, whose one-row generating function is \(\bigl((1+x)^k+(1-x)^k\bigr)/2\).
Multiplying the resulting coefficient probability by the \(\binom ns\) choices of \(U\) gives the displayed form.

Set
\begin{equation*}
\alpha_Z:=\frac{j_Z}{k}.
\end{equation*}
For \(0<\tau<1\), define the upper exponent by a limsup over
admissible blocklengths and integer weights with \(s/n\to\tau\).
This definition includes irrational \(\tau\) and does not require
\(\tau n\) itself to be an integer. All finite-length estimates below
are stated directly for integer \(s\). Coefficient extraction gives
\begin{equation}\label{eq:Wo-jp}
W_o(\tau):=
\limsup_{\substack{n\to\infty\\s/n\to\tau}}\frac1n\log_2 N_o(s)
\le
\alpha_Z \inf_{x>0}\log_2 \frac{(1+x)^k+(1-x)^k}{2x^{\tau k}} - (j_Z-1)\h(\tau).
\end{equation}

To obtain this bound, apply the coefficient-extraction estimate
\([x^{j_Z s}]F(x)\le F(x)/x^{j_Z s}\) to the displayed definition of \(N_o(s)\) for any \(x>0\),
and then use Stirling's formula
\(\frac1n\log_2\binom{n}{s}=\h(s/n)+o(1)\),
\(\frac1n\log_2\binom{n j_Z}{j_Zs}=j_Z\h(s/n)+o(1)\).
For every \(x>0\),
\(\frac1n\log_2 N_o(s)\) is bounded above by
\(\h(\tau)+\alpha_Z\log_2\!\bigl(((1+x)^k+(1-x)^k)/2x^{\tau k}\bigr)-j_Z\h(\tau)+o(1)\).
Taking the limsup with \(s/n\to\tau\), using continuity for each fixed \(x>0\), and then optimizing over \(x>0\) gives \eqref{eq:Wo-jp}.
In particular, substituting \(x=\tau/(1-\tau)\) yields
\begin{equation}\label{eq:Wo-trial-jp}
W_o(\tau)\le W_o^{\mathrm{tr}}(\tau),
\qquad
W_o^{\mathrm{tr}}(\tau):=
\alpha_Z\log_2\!\bigl(1+(1-2\tau)^k\bigr)+\h(\tau)-\alpha_Z.
\end{equation}

Before taking the limsup, use the same trial point and
\eqref{eq:binomial-uniform-jp}. This gives the explicit finite-length bound
\[
N_o(s)\le (nj_Z+1)2^{nW_o^{\mathrm{tr}}(s/n)}
\qquad (1\le s\le n-1).
\]

Next define
\begin{equation*}
f_+(z,k):=\frac{(1+z)^k+(1-z)^k}{2},
\qquad
f_-(z,k):=\frac{(1+z)^k-(1-z)^k}{2}.
\end{equation*}
For a fixed outer matrix \(A_Z\), let
\begin{equation*}
A_o(s;A_Z):=\bigl|\{\bu\in\Ker A_Z:\wt(\bu)=s\}\bigr|.
\end{equation*}
Let \([z^m]f(z)\) denote the coefficient of \(z^m\) in \(f(z)\).
Define
\begin{equation*}
M_k(s,l)
:=
\frac{\binom{n}{l}[z^{ks}]\,f_-(z,k)^l f_+(z,k)^{n-l}}
{\binom{kn}{ks}}.
\end{equation*}
The HA inner transition kernel~\cite{HA05} for the present square \((k,k)\)-regular map is \(M_k(s,l)\): an outer codeword of weight \(s\) produces an output of weight \(l\) through \(B\).

\begin{lemma}[HA transition enumerator for the square regular map]\label{lem:HA-transition-jp}
Assume \(k\) is even, and let \(B\) be sampled from the square \((k,k)\)-regular
configuration-model ensemble of Definition~\ref{def:family-jp}, independently of the fixed matrix \(A_Z\). For \(1\le l\le n\),
\begin{equation*}
\EE_B[A_{\CZ}(l)\mid A_Z]
\le
\sum_{s=\lceil l/k\rceil}^{n-\lceil l/k\rceil} A_o(s;A_Z)\,M_k(s,l)
\end{equation*}
holds.
In particular, taking expectation also over \(A_Z\) gives
\begin{equation}\label{eq:HA-enum-jp}
\EE[A_{\CZ}(l)]
\le
N_Z^{\mathrm{ub}}(l)
:=
\sum_{s=\lceil l/k\rceil}^{n-\lceil l/k\rceil}
N_o(s)\,M_k(s,l).
\end{equation}
Moreover, for each \(s\),
\begin{equation*}
0\le M_k(s,l)\le 1,
\qquad
\sum_{l=0}^n M_k(s,l)=1.
\end{equation*}
\end{lemma}

\begin{proof}
Fix one outer word \(\bu\in\Ker A_Z\) of weight \(s\), and let its support be \(U\subset[n]\).
Then the active column-side sockets attached to \(U\) number exactly \(ks\),
and their images form a uniformly random \(ks\)-subset of the \(kn\) row-side sockets.

Next fix a set \(T\subset[n]\) of weight \(l\).
The condition that row \(i\in T\) receives an odd number of active sockets and row \(i\notin T\) an even number
is equivalent to the statement that the support of \(B\bu\) is exactly \(T\).
Therefore the one-row generating function is \(f_-(z,k)\) for \(i\in T\) and \(f_+(z,k)\) for \(i\notin T\),
and the total number of active row-side subsets satisfying the condition is
\begin{equation*}
[z^{ks}]\,f_-(z,k)^l f_+(z,k)^{n-l}.
\end{equation*}
Hence, for fixed \(T\), the corresponding probability is
\begin{equation*}
\frac{[z^{ks}]\,f_-(z,k)^l f_+(z,k)^{n-l}}{\binom{kn}{ks}},
\end{equation*}
and multiplying by the number \(\binom{n}{l}\) of choices of \(T\) shows that the probability of obtaining an output of weight \(l\)
is \(M_k(s,l)\).

Furthermore, for even \(k\), the degree of \(f_-(z,k)\) is \(k-1\) and the degree of \(f_+(z,k)\) is \(k\),
so the expression is automatically zero unless
\begin{equation*}
l\le ks\le nk-l.
\end{equation*}

Summing over all outer codewords of weight \(s\), we get
\begin{equation*}
\EE_B[A_{\CZ}(l)\mid A_Z]
\le
\sum_s A_o(s;A_Z) M_k(s,l).
\end{equation*}
The inequality, rather than equality, is needed because distinct \(\bu\in\Ker A_Z\) can map to the same \(\bv=B\bu\).
Since the events \(|T|=l\) are disjoint and exhaust all supports, \(\sum_{l=0}^n M_k(s,l)=1\).
Taking expectation also over \(A_Z\) yields \eqref{eq:HA-enum-jp}.
\end{proof}

\begin{lemma}[Half-interval entropy inequality]\label{lem:half-entropy-jp}
For every integer \(L\ge2\),
\begin{equation}\label{eq:scalar-half-ineq-jp}
\h(q)+\log_2\!\bigl(1+(1-2q)^L\bigr)\le1,
\qquad
\frac1L\le q\le\frac12.
\end{equation}
\end{lemma}

\begin{proof}
Set \(z:=1-2q\) and
\[
J_L(z):=1-\h\!\left(\frac{1-z}{2}\right)-\log_2(1+z^L),
\qquad
0\le z\le 1-\frac2L.
\]
For \(L=2\) this interval consists only of \(z=0\), where \(J_2(0)=0\).
Hence assume \(L\ge3\). Then
\[
(\ln2)J_L'(z)
=
\operatorname{artanh}z-\frac{Lz^{L-1}}{1+z^L}.
\]
Define
\[
R_L(z):=
\frac{Lz^{L-1}}{(1+z^L)\operatorname{artanh}z}.
\]
To locate the minima of \(J_L\), first show that \(R_L\) is increasing on
\(0<z\le1-2/L\). Differentiating its logarithm gives
\[
\frac{L-1}{z}
-\frac{Lz^{L-1}}{1+z^L}
-\frac{1}{(1-z^2)\operatorname{artanh}z}.
\]
The estimates
\[
\operatorname{artanh}z\ge z,
\qquad
\frac{z^{L-1}}{1+z^L}\le\frac12,
\qquad
\frac{1}{1-z^2}\le \frac{L^2}{4(L-1)}
\]
hold on this interval. The middle inequality follows from
\(z^{L-1}(2-z)\le1\) on \(0\le z\le1\), and the last one from
\(z\le1-2/L\). Together with \(1/\operatorname{artanh}z\le1/z\), they give
the lower bound
\[
\frac{3L^2-8L+4}{4(L-1)z}-\frac{L}{2},
\]
whose minimum over \(0<z\le1-2/L\) is \(L^2/(4(L-1))>0\).
Thus \(R_L\) is increasing. Since \(R_L(z)\to0\) as \(z\downarrow0\),
the sign of \(J_L'\) is positive until \(R_L\) reaches \(1\), and negative
after that point if such a point exists. Therefore \(J_L\) has no interior
minimum, so its minimum is attained at an endpoint.

The endpoint \(z=0\) gives \(J_L(0)=0\). At the other endpoint
\(z=1-2/L\), equivalently \(q=1/L\), the values of \((\ln2)J_L\) for
\(L=2,3,4\), in that order, are
\[
0,
\qquad
\frac53\ln2-\ln\frac{28}{9},
\qquad
3\ln2+\frac34\ln3-\ln17,
\]
and the last two are strictly positive.  Indeed, after exponentiating
and clearing the fractional powers, the two claims are respectively
equivalent to
\[
2^5 9^3>28^3,
\qquad
8^4 3^3>17^4,
\]
which reduce to the integer comparisons \(23328>21952\) and
\(110592>83521\). For \(L\ge5\), writing
\(H(q):=(\ln2)\h(q)\), the bounds
\[
H(1/L)\le \frac1L\ln(\mathrm{e}L),
\qquad
\left(1-\frac2L\right)^L\le \mathrm{e}^{-2}
\]
give
\[
H(1/L)+\ln(1+(1-2/L)^L)
\le
\frac{\ln(\mathrm{e}L)}{L}+\ln(1+\mathrm{e}^{-2})
<\ln2.
\]
The last expression is decreasing for \(L\ge5\), because
\[
\frac{d}{dL}\frac{\ln(\mathrm{e}L)}{L}
=-\frac{\ln L}{L^2}<0,
\]
and at \(L=5\) it is strictly below \(\ln2\).  One completely rational
verification is
\[
\frac{\ln(5\mathrm e)}5+\ln(1+\mathrm e^{-2})
<\frac{21}{40}+\frac5{36}
=\frac{239}{360}<\frac23<\ln2.
\]
Here \(\ln5<13/8\) follows by retaining the terms through degree five
in the positive series for \(\exp(13/8)\), \(\mathrm e^{-2}<5/36\)
follows from the terms through degree five in the series for
\(\mathrm e^2\), and \(\ln2>2/3\) follows from the first term of the
positive \(\operatorname{artanh}(1/3)\) series. This proves
\eqref{eq:scalar-half-ineq-jp}.
\end{proof}

\begin{lemma}[HA-side analytic inputs]\label{lem:HA-inputs-jp}
For fixed \((j_Z,k)\)-regular outer degrees with \(4\le j_Z<k/2\) and even
\(k\), set
\[
\delta_o:=2\mathrm{e}\,\mathrm{e}^{-(k+j_Z)}.
\]
Then
\begin{equation*}
W_o(\tau)<0 \quad (0<\tau<\delta_o),
\qquad
\sum_{1\le s\le \delta_o n} N_o(s)=o(1).
\end{equation*}
In addition, with \(\alpha_Z=j_Z/k\),
\begin{equation*}
\h(\tau)-\alpha_Z
+(1+\alpha_Z)\log_2\!\left(1+(1-2\tau)^k\right)-1<0
\qquad (\delta_o\le\tau\le1-\delta_o).
\end{equation*}
\end{lemma}

\begin{proof}
The stated value satisfies \(0<\delta_o<1/k\). It also gives
\[
\lambda_o:=
\mathrm{e}^{1-j_Z/2}k^{j_Z/2}\delta_o^{j_Z/2-1}<1.
\]
Indeed,
\[
\ln\lambda_o
=1-\frac{j_Z}{2}+\frac{j_Z}{2}\ln k
+\left(\frac{j_Z}{2}-1\right)
 \bigl(\ln(2\mathrm{e})-k-j_Z\bigr).
\]
Its two partial derivatives are
\[
\frac{\partial}{\partial k}\ln\lambda_o
=\frac{j_Z}{2k}-\left(\frac{j_Z}{2}-1\right),
\qquad
\frac{\partial}{\partial j_Z}\ln\lambda_o
=\frac{1+\ln k+\ln(2\mathrm e)-k-2j_Z}{2}.
\]
The first is at most \(2/10-1<0\) on \(k\ge10\), \(j_Z\ge4\).
For the second, \(\ln k<k-1\), \(\ln(2\mathrm e)<2\), and \(j_Z\ge4\)
give a strict negative upper bound.  Hence \(\ln\lambda_o\) decreases
in both variables.  At the corner,
\[
\ln\lambda_o\big|_{(k,j_Z)=(10,4)}
=2\ln10+\ln(2\mathrm e)-15<6+2-15=-7,
\]
where \(\ln10<3\) and \(\ln(2\mathrm e)<2\).  Thus \(\lambda_o<1\).
For \(1\le s\le\delta_o n\), if \(j_Zs\) is odd then \(N_o(s)=0\). If \(j_Zs\) is even, pair the \(j_Zs\) active sockets. Since \(s\le n/k\), before the second endpoint of any pair is exposed, at most \(j_Zs-1\le m_Z-1\) active row-side sockets have already been exposed. Hence at least \(km_Z-(m_Z-1)\) row-side sockets remain.

Among those remaining sockets, at most \(k-1\) can complete the pair inside the same row. The conditional probability for a prescribed pair to land in one row is therefore at most
\[
\frac{k-1}{km_Z-(m_Z-1)}\le \frac1{m_Z}.
\]

Summing over supports and pairings gives
\[
N_o(s)\le \binom ns (j_Zs-1)!!\,m_Z^{-j_Zs/2}
\]
and the inequalities \(\binom ns\le(\mathrm{e}n/s)^s\) and
\((j_Zs-1)!!\le\sqrt2(j_Zs/\mathrm{e})^{j_Zs/2}\) yield
\[
N_o(s)\le
\sqrt2\left[
\mathrm{e}^{1-j_Z/2}k^{j_Z/2}\left(\frac{s}{n}\right)^{j_Z/2-1}
\right]^s
\le \sqrt2\,\lambda_o^s.
\]
For every fixed \(s\), the bound
\(\binom ns (j_Zs-1)!!\,m_Z^{-j_Zs/2}\), with \(m_Z=j_Zn/k\), has \(n\)-dependence
\(n^s n^{-j_Zs/2}\) up to an \(s\)-dependent constant. Hence
\[
N_o(s)=O\!\left(n^{-s(j_Z/2-1)}\right).
\]
Thus \(W_o(\tau)\le \tau\log_2\lambda_o<0\) for \(0<\tau<\delta_o\). Given \(\varepsilon>0\), choose \(U\) with \(\sum_{s>U}\sqrt2\lambda_o^s<\varepsilon/2\), and then take \(n\) large enough that \(\sum_{s=1}^{U}N_o(s)<\varepsilon/2\). Therefore the small-input sum satisfies
\[
\sum_{1\le s\le\delta_o n}N_o(s)=o(1).
\]
It remains to prove the displayed compact-interval inequality. By evenness
of \(k\) and symmetry of \(\h\), it suffices to consider
\(\delta_o\le\tau\le1/2\). Use natural logarithms, write
\(H(\tau):=(\ln2)\h(\tau)\), put \(x:=k\tau\),
\(x_0:=k\delta_o=2\mathrm{e}k\mathrm{e}^{-(k+j_Z)}\), and set
\(\alpha:=j_Z/k\). When \(x\le1\), the inequalities
\[
H(\tau)\le\tau\ln\frac{\mathrm{e}}{\tau},
\qquad
(1-2\tau)^k\le\mathrm{e}^{-2x},
\qquad
\ln\cosh x\le\frac{x^2}{2}
\]
bound \(\ln2\) times the displayed expression by
\[
x\left[
\frac1k\ln\frac{\mathrm{e}k}{x}
-(1+\alpha)+\frac{1+\alpha}{2}x
\right]
=:xB(x).
\]
The function \(B\) is convex on \([x_0,1]\), so its maximum is attained
at an endpoint. At the lower endpoint,
\[
B(x_0)=-\frac{\ln2}{k}+\frac{1+\alpha}{2}x_0<0.
\]
For the last inequality, the positive term is largest under the stated
degree assumptions at \((k,j_Z)=(10,4)\).  Indeed,
\[
k(1+\alpha)x_0=2\mathrm e\,k(k+j_Z)\mathrm e^{-(k+j_Z)},
\]
and the logarithmic partial derivatives of the expression on the right
(apart from its constant factor) are
\(1/k+1/(k+j_Z)-1<0\) and \(1/(k+j_Z)-1<0\).
At the corner its value is \(280\mathrm e^{-13}<1<2\ln2\), proving
\(B(x_0)<0\) without a decimal approximation. At the upper endpoint,
\[
B(1)
=\frac{\ln(\mathrm{e}k)}{k}-\frac{1+\alpha}{2}
\le \frac{\ln k-1}{k}-\frac12<0
\]
for \(k\ge10\), where \(j_Z\ge4\) was used. Hence the required
expression is strictly negative for \(x_0\le x\le1\).

When \(x\ge1\), Lemma~\ref{lem:half-entropy-jp}, used with \(L=k\), gives
\[
\h(\tau)+\log_2\!\left(1+(1-2\tau)^k\right)\le1.
\]
Because \(\alpha>0\) and
\(\log_2(1+(1-2\tau)^k)<1\) for \(\tau>0\), the displayed expression in
the lemma is again strictly negative. Symmetry completes the proof.
\end{proof}

Write \(\omega=l/n\), \(\tau=s/n\), and \(T=(1-2\tau)^k\). Apply the coefficient-extraction bound to \(M_k(s,l)\) at the trial point \(z=\tau/(1-\tau)\), interpreted by limits at \(\tau=0,1\).

Combine this coefficient bound with the trial upper envelope
\(W_o^{\mathrm{tr}}\) from \eqref{eq:Wo-trial-jp}. The binomial bounds
\eqref{eq:binomial-uniform-jp} give, uniformly for \(1\le s\le n-1\)
and \(0\le l\le n\),
\[
N_o(s)M_k(s,l)\le(nj_Z+1)(nk+1)
2^{n[\h(l/n)+F_Z(s/n,l/n)]},
\]
where \(F_Z\) is defined below. Thus the discarded logarithmic error is
uniformly \(O(\log n)\), and
\begin{align}
\limsup_{\substack{n\to\infty\\l/n\to\omega}}\frac1n\log_2 N_Z^{\mathrm{ub}}(l)
&\le
W_Z^{\mathrm{ub}}(\omega),\\
W_Z^{\mathrm{ub}}(\omega)
&:=
\h(\omega)+\max_{\omega/k\le \tau\le 1-\omega/k} F_Z(\tau,\omega).
\end{align}
Here
\begin{equation}\label{eq:FZ-jp}
F_Z(\tau,\omega)
:=
W_o^{\mathrm{tr}}(\tau)
+
\omega\log_2\frac{1-T}{2}
+
(1-\omega)\log_2\frac{1+T}{2}.
\end{equation}
Equivalently, \(F_Z\) is the explicit finite-length upper exponent obtained
by multiplying the trial bound \eqref{eq:Wo-trial-jp} for \(N_o(s)\) with the coefficient
bound for \(M_k(s,l)\), before the polynomial Stirling factors are discarded.
Moreover,
\begin{equation*}
\omega\log_2\frac{1-T}{2}
+
(1-\omega)\log_2\frac{1+T}{2}
=
\omega\log_2 p+(1-\omega)\log_2(1-p)
\le -\h(\omega),
\qquad
p:=\frac{1-T}{2},
\end{equation*}
so that
\begin{equation*}
F_Z(\tau,\omega)\le W_o^{\mathrm{tr}}(\tau)-\h(\omega)
\end{equation*}
holds.

The inequality is the Bernoulli log-likelihood bound with the continuous endpoint convention:
\(\omega\log_2 p+(1-\omega)\log_2(1-p)\) is maximized over \(0\le p\le1\)
at \(p=\omega\), with value \(-\h(\omega)\).

\begin{remark*}[Complement symmetry of the outer code for even \(k\)]
When \(k\) is even, the sum of every binary row of \(A_Z\) is the
socket degree \(k\) modulo two, and therefore
\begin{equation*}
A_Z\ones_{[n]}=0
\end{equation*}
holds.
Hence \(\ones_{[n]}\in\Ker A_Z\), and the map \(\bu\mapsto \bu+\ones_{[n]}\) gives a bijection
between words of weight \(s\) and words of weight \(n-s\) in \(\Ker A_Z\).
In particular,
\begin{equation*}
N_o(s)=N_o(n-s),
\qquad
W_o(\tau)=W_o(1-\tau)
\end{equation*}
holds.
\end{remark*}

\subsection{Proof of Theorem~\ref{thm:HA-lin-jp}}

\begin{proof}
The first-moment bound \eqref{eq:HA-enum-jp} and the exponent \(F_Z\) in \eqref{eq:FZ-jp} control the low visible-weight range.
Fix the value \(\delta_o>0\) supplied by Lemma~\ref{lem:HA-inputs-jp}.
On the compact interval
\(\delta_o\le\tau\le1-\delta_o\), the exponent \(F_Z\) in \eqref{eq:FZ-jp} is continuous because it is built from the explicit trial upper envelope \(W_o^{\mathrm{tr}}\).
At \(\omega=0\), its explicit form gives
\[
F_Z(\tau,0)
=
\h(\tau)-\alpha_Z
+(1+\alpha_Z)\log_2\!\left(1+(1-2\tau)^k\right)-1.
\]
The compact-interval estimate in Lemma~\ref{lem:HA-inputs-jp} says directly
that this expression is strictly negative throughout
\(\delta_o\le\tau\le1-\delta_o\). Hence there is a margin
\(\gamma_o>0\) such that \(F_Z(\tau,0)\le-2\gamma_o\) throughout
\(\delta_o\le\tau\le1-\delta_o\). Since \(F_Z(\tau,\omega)+\h(\omega)\) is uniformly continuous on
\([\delta_o,1-\delta_o]\times[0,1/2]\), choose
\(\delta_Z^{\mathrm{lin}}>0\) such that
\begin{equation*}
F_Z(\tau,\omega)+\h(\omega)\le-\gamma_o
\qquad
(\delta_o\le\tau\le1-\delta_o,\ 0\le\omega\le\delta_Z^{\mathrm{lin}}).
\end{equation*}
Fix any \(0<\delta<\delta_Z^{\mathrm{lin}}\).

After summing over the output weight, \eqref{eq:HA-enum-jp} gives
\begin{equation*}
\sum_{1\le l\le \delta n}\EE[A_{\CZ}(l)]
\le
\sum_{1\le l\le \delta n}
\sum_{s=\lceil l/k\rceil}^{n-\lceil l/k\rceil}
N_o(s)M_k(s,l).
\end{equation*}
Split the input-weight sum into the low-input ranges
\[
1\le s\le \delta_o n,
\qquad
n-\delta_o n\le s\le n-1,
\]
and the complementary range. For \(1\le s\le\delta_o n\), the identity
\(\sum_{l=0}^n M_k(s,l)=1\) and Lemma~\ref{lem:HA-inputs-jp} give a total contribution at most
\(\sum_{1\le s\le\delta_o n}N_o(s)=o(1)\).

For \(n-\delta_o n\le s\le n-1\), set \(s':=n-s\). The even-\(k\) complement symmetry gives \(N_o(s)=N_o(s')\), and the normalization \(\sum_{l=0}^n M_k(s,l)=1\) holds for every \(s\). Hence, after summing over \(l\), this high-input range is bounded by
\(\sum_{1\le s'\le\delta_o n}N_o(s')=o(1)\). Thus the two edge ranges contribute \(o(1)\) to
\(\sum_{1\le l\le \delta n}\EE[A_{\CZ}(l)]\).

For the complementary range, the coefficient-extraction estimate used in deriving \eqref{eq:FZ-jp}, together with Stirling's formula with endpoint-uniform \(O(\log n)\) errors, gives a constant \(C>0\) such that every summand with \(\delta_o\le\tau=s/n\le1-\delta_o\) and \(1\le l\le\delta n\) satisfies
\[
N_o(s)M_k(s,l)
\le
n^C2^{n(F_Z(\tau,\omega)+\h(\omega)+o_n(1))},
\qquad
\omega:=l/n,
\]
where \(o_n(1)\) is uniform for \(\delta_o\le\tau\le1-\delta_o\) and \(0\le\omega\le\delta\). Since \(F_Z(\tau,\omega)+\h(\omega)\le-\gamma_o\), take \(n\) large so that \(o_n(1)\le\gamma_o/4\) uniformly and \((n+1)n^C\le2^{\gamma_o n/4}\).

Then summing over the admissible \(s\) gives
\begin{equation*}
\sum_{\substack{\lceil l/k\rceil\le s\le n-\lceil l/k\rceil\\ \delta_o n< s< n-\delta_o n}}
N_o(s)\,M_k(s,l)
\le
2^{-\gamma_o n/2}
\qquad
(1\le l\le \delta n)
\end{equation*}
for all sufficiently large \(n\).

The two edge ranges and the complementary range bound give
\begin{equation*}
\sum_{1\le l\le \delta n}\EE[A_{\CZ}(l)]
\le
o(1)+
\sum_{1\le l\le \delta n}2^{-\gamma_o n/2}
=o(1).
\end{equation*}

Set
\begin{equation*}
X_n(\delta):=\sum_{1\le l\le \delta n}A_{\CZ}(l).
\end{equation*}
The event \(d(\CZ)\le \delta n\) is equivalent to \(X_n(\delta)\ge 1\).
Thus Markov's inequality gives
\begin{equation*}
\PP[d(\CZ)\le \delta n]
=
\PP[X_n(\delta)\ge 1]
\le
\EE[X_n(\delta)]
=o(1).
\end{equation*}
\end{proof}

\engappendixsection{MN-Side Enumerators and Linear-Distance Proof}{app:MN-proofs-jp}

The ensemble-level results in this appendix use the homogeneous balanced specialization \(k_Z=k_\Delta=k\) and \(j_X+j_Z=k\); scalar auxiliary lemmas state their own parameters.
Set \(j_\Delta:=j_X-j_Z\) and
\[
\alpha_Z:=\frac{j_Z}{k},\qquad
\alpha_\Delta:=\frac{j_\Delta}{k},\qquad
\alpha_X:=\frac{j_X}{k}=\alpha_Z+\alpha_\Delta=1-\alpha_Z.
\]
Then \(m_Z=\alpha_Z n\), \(m_\Delta=\alpha_\Delta n\), and \(m_X=m_Z+m_\Delta=\alpha_X n\).
For \(\bx=(\bp_Z,\bp_\Delta)\in \F_2^{m_Z}\times \F_2^{m_\Delta}\), the stacked equation is
\[
A_X^T \bx=A_Z^T \bp_Z + A_\Delta^T \bp_\Delta,
\]
where \(\bp_Z\) and \(\bp_\Delta\) are the auxiliary variables corresponding to the two row blocks of \(A_X=[A_Z;A_\Delta]\).

For integers \(t_1,t_\Delta,w\), define
\begin{equation*}
\NX(t_1,t_\Delta,w)
:=
\left|
\left\{
(\bp_Z,\bp_\Delta,\bv):
\begin{array}{l}
\wt(\bp_Z)=t_1,\ \wt(\bp_\Delta)=t_\Delta,\ \wt(\bv)=w,\\
A_Z^T \bp_Z+A_\Delta^T \bp_\Delta+B^T \bv=0
\end{array}
\right\}
\right|
\end{equation*}
The refined enumerator counts each weight-\(w\) candidate \(\bv\) together with a witness \((\bp_Z,\bp_\Delta)\). A single \(\bv\in\CX\) can have several witnesses, so summing over \((t_1,t_\Delta)\) gives the upper bound
\[
A_{\CX}(w)\le \sum_{t_1=0}^{m_Z}\sum_{t_\Delta=0}^{m_\Delta}\NX(t_1,t_\Delta,w).
\]

Lemma~\ref{lem:exact-jp} treats the three socket pools separately by edge type, using the configuration model of~\cite[Secs.~3.3--3.4]{RU08}. Let \(u\), \(v\), and \(r\) mark active sockets in the \(A_Z\)-, \(A_\Delta\)-, and \(B\)-pools, respectively, and define the one-column even-parity generating function
\begin{equation}\label{eq:onecol-gen-jp}
g_{j_Z,j_\Delta,k}(u,v,r)
:=
\frac{(1+u)^{j_Z}(1+v)^{j_\Delta}(1+r)^k+(1-u)^{j_Z}(1-v)^{j_\Delta}(1-r)^k}{2}.
\end{equation}

\begin{lemma}[Exact stacked configuration formula]\label{lem:exact-jp}
For independent configuration-model draws of
\(A_Z\), \(A_\Delta\), and \(B\) from Definition~\ref{def:family-jp}, specialized to
\(k_Z=k_\Delta=k\), for integers
\(0\le t_1\le m_Z\), \(0\le t_\Delta\le m_\Delta\), and \(0\le w\le n\),
\begin{equation}\label{eq:exact-stacked-jp}
\EE[\NX(t_1,t_\Delta,w)]
=
\binom{m_Z}{t_1}\binom{m_\Delta}{t_\Delta}\binom{n}{w}
\frac{[u^{kt_1}v^{kt_\Delta}r^{kw}]\,g_{j_Z,j_\Delta,k}(u,v,r)^n}
{\binom{n j_Z}{kt_1}\binom{n j_\Delta}{kt_\Delta}\binom{nk}{kw}}
\end{equation}
holds.
\end{lemma}

\begin{proof}
Fix supports \(S_1\subset[m_Z]\), \(S_\Delta\subset[m_\Delta]\), and \(T\subset[n]\) of sizes \(t_1,t_\Delta,w\), and let \(\ones_{S_1}, \ones_{S_\Delta}, \ones_T\) denote their indicator vectors. The stacked socket-based configuration model~\cite[Secs.~3.3--3.4]{RU08} has \(kt_1\) active row-side sockets in the \(A_Z\)-pool, \(kt_\Delta\) in the \(A_\Delta\)-pool, and \(kw\) in the \(B\)-pool.

Because \(A_Z\), \(A_\Delta\), and \(B\) are sampled independently, these active row-side sockets induce independent uniformly random active subsets in the column-side socket pools of sizes \(n j_Z\), \(n j_\Delta\), and \(nk\).

For a fixed column, let \(r_Z\), \(r_\Delta\), and \(r_B\) be the active-socket counts in its \(A_Z\)-, \(A_\Delta\)-, and \(B\)-pools. The equation
\[
A_Z^T\ones_{S_1}+A_\Delta^T\ones_{S_\Delta}+B^T\ones_T=0
\]
holds if and only if \(r_Z+r_\Delta+r_B\equiv 0\pmod 2\) in every column. Hence the one-column generating function is \(g_{j_Z,j_\Delta,k}(u,v,r)\), and the number of desired active column-side subset triples is
\([u^{kt_1}v^{kt_\Delta}r^{kw}]\,g_{j_Z,j_\Delta,k}(u,v,r)^n\).
Dividing by the total number of independent subset triples and multiplying by the number of support choices yields the displayed formula for \(\EE[\NX(t_1,t_\Delta,w)]\).
\end{proof}

\begin{proposition}[Blockwise complement symmetry]\label{prop:fold-jp}
Assume \(j_Z\equiv 0\pmod 2\) and \(j_\Delta\equiv 0\pmod 2\). Then \(A_Z^T\ones_{[m_Z]}=0\) and \(A_\Delta^T\ones_{[m_\Delta]}=0\) hold. Consequently \(\NX(t_1,t_\Delta,w)=\NX(m_Z-t_1,t_\Delta,w)=\NX(t_1,m_\Delta-t_\Delta,w)\), and
\begin{equation}\label{eq:fold-bound-jp}
A_{\CX}(w)
\le
4\sum_{0\le t_1\le m_Z/2}\sum_{0\le t_\Delta\le m_\Delta/2}\NX(t_1,t_\Delta,w)
\end{equation}
also holds. The reduced region \(0\le t_1\le m_Z/2\), \(0\le t_\Delta\le m_\Delta/2\) is the folded domain.
\end{proposition}

\begin{proof}
The parity of every binary column sum is the parity of its prescribed
socket degree, even if parallel edges cancel.  Hence
\(A_Z^T\ones_{[m_Z]}=(j_Z\bmod2)\ones_{[n]}=0\) and
\(A_\Delta^T\ones_{[m_\Delta]}=(j_\Delta\bmod2)\ones_{[n]}=0\).
If
\[
A_Z^T \bp_Z + A_\Delta^T \bp_\Delta + B^T \bv = 0,
\]
then complementing either block witness preserves the equation:
\[
A_Z^T(\bp_Z+\ones_{[m_Z]}) + A_\Delta^T \bp_\Delta + B^T \bv = 0,
\qquad
A_Z^T \bp_Z + A_\Delta^T(\bp_\Delta+\ones_{[m_\Delta]}) + B^T \bv = 0.
\]
Complementing \(\bp_Z\) changes the witness weight to \((m_Z-t_1,t_\Delta,w)\), while complementing \(\bp_\Delta\) changes it to \((t_1,m_\Delta-t_\Delta,w)\). Hence \(\NX(t_1,t_\Delta,w)=\NX(m_Z-t_1,t_\Delta,w)=\NX(t_1,m_\Delta-t_\Delta,w)\).

The witness definition of \(\NX\) gives
\[
A_{\CX}(w)\le \sum_{0\le t_1\le m_Z}\sum_{0\le t_\Delta\le m_\Delta}\NX(t_1,t_\Delta,w),
\]

Partition \([0,m_Z]\times[0,m_\Delta]\) by the two involutions
\[
(t_1,t_\Delta)\mapsto (m_Z-t_1,t_\Delta),
\qquad
(t_1,t_\Delta)\mapsto (t_1,m_\Delta-t_\Delta).
\]
Each orbit has size at most \(4\), and the two complement equalities make \(\NX\) constant on each orbit. Taking the representative in
\(0\le t_1\le m_Z/2\), \(0\le t_\Delta\le m_\Delta/2\), the contribution of each orbit is bounded by at most four times its representative value. Thus
\[
A_{\CX}(w)
\le
4\sum_{0\le t_1\le m_Z/2}\sum_{0\le t_\Delta\le m_\Delta/2}\NX(t_1,t_\Delta,w),
\]
which is \eqref{eq:fold-bound-jp}.
\end{proof}

\subsection{Trial-Point Bound Governing the Linear-Weight Regime}

\begin{lemma}[Trial-point coefficient bound]\label{lem:trial-coeff-jp}
In Lemma~\ref{lem:exact-jp}, write
\[
a:=\frac{t_1}{m_Z},\qquad
b:=\frac{t_\Delta}{m_\Delta},\qquad
\omega:=\frac{w}{n}.
\]
Assume first that \(0<a,b,\omega<1\) and set \(s:=a/(1-a)\), \(t:=b/(1-b)\), and \(r:=\omega/(1-\omega)\). Then
\begin{equation}\label{eq:trial-coeff-bound-jp}
\EE[\NX(t_1,t_\Delta,w)]
\le
\binom{m_Z}{t_1}\binom{m_\Delta}{t_\Delta}\binom{n}{w}
\frac{g_{j_Z,j_\Delta,k}(s,t,r)^n}
{\binom{n j_Z}{kt_1}\binom{n j_\Delta}{kt_\Delta}\binom{nk}{kw}s^{kt_1}t^{kt_\Delta}r^{kw}}
\end{equation}
holds.
\end{lemma}

\begin{proof}
Set
\[
P(u,v,r):=g_{j_Z,j_\Delta,k}(u,v,r)^n
=\sum_{p,q,\ell\ge 0} c_{p,q,\ell}u^p v^q r^\ell,
\qquad c_{p,q,\ell}\ge 0.
\]
Then
\[
[u^{kt_1}v^{kt_\Delta}r^{kw}]\,g_{j_Z,j_\Delta,k}(u,v,r)^n
=c_{kt_1,kt_\Delta,kw}.
\]
For \(s,t,r>0\), nonnegativity of the coefficients gives
\(P(s,t,r)\ge c_{kt_1,kt_\Delta,kw}s^{kt_1}t^{kt_\Delta}r^{kw}\), and hence
\(c_{kt_1,kt_\Delta,kw}\le P(s,t,r)/(s^{kt_1}t^{kt_\Delta}r^{kw})\).
Since \(P(s,t,r)=g_{j_Z,j_\Delta,k}(s,t,r)^n\), this coefficient bound gives
\[
[u^{kt_1}v^{kt_\Delta}r^{kw}]\,g_{j_Z,j_\Delta,k}(u,v,r)^n\le \frac{g_{j_Z,j_\Delta,k}(s,t,r)^n}{s^{kt_1}t^{kt_\Delta}r^{kw}},
\]
which is the coefficient estimate used in \eqref{eq:trial-coeff-bound-jp}. Substituting it into Lemma~\ref{lem:exact-jp} gives the displayed bound for \(\EE[\NX(t_1,t_\Delta,w)]\).
\end{proof}

For endpoint values \(a=0\), \(b=0\), or \(\omega=0\) in the half-domain, interpret \(s=a/(1-a)\), \(t=b/(1-b)\), and \(r=\omega/(1-\omega)\) as limits from positive values. The coefficient-extraction bound is valid for positive trial points, and the right-hand side extends continuously to these limits.

The uniform Stirling estimate in Lemma~\ref{lem:trial-exp-jp} includes \(q=0,N\). With \(\h(0)=0\), the exponent bound remains valid at the endpoints, including the zero-output case \(w=0\).

\begin{lemma}[Trial-point exponent estimate]\label{lem:trial-exp-jp}
With \(a=t_1/m_Z\), \(b=t_\Delta/m_\Delta\), \(\omega=w/n\), \(y_1=|1-2a|\), \(y_\Delta=|1-2b|\), and \(\mu=|1-2\omega|^k\),
\begin{equation}\label{eq:trial-jp}
\frac1n\log_2 \EE[\NX(t_1,t_\Delta,w)]
\le
\alpha_Z\h(a)+\alpha_\Delta\h(b)+\h(\omega)-1
+\log_2\!\bigl(1+\mu\, y_1^{j_Z}y_\Delta^{j_\Delta}\bigr)+o(1)
\end{equation}
holds on the half-domain \(0\le a,b,\omega\le 1/2\). More precisely,
the \(o(1)\) term can be replaced by the explicit nonnegative quantity
\[
\varepsilon_n:=\frac1n\log_2\bigl[(nj_Z+1)(nj_\Delta+1)(nk+1)\bigr]
=O\!\left(\frac{\log n}{n}\right),
\]
uniformly on that half-domain, including zero-output and zero-witness
endpoints. If the expectation vanishes, the logarithm on the left is
\(-\infty\).
\end{lemma}

\begin{proof}
Applying \eqref{eq:binomial-uniform-jp} to the three numerator
binomial coefficients and the three denominator coefficients yields
\(\varepsilon_n\) exactly. To display the entropy cancellation, apply
Lemma~\ref{lem:trial-coeff-jp} with
\(s=a/(1-a)\), \(t=b/(1-b)\), and \(r=\omega/(1-\omega)\), and then apply Stirling's formula. Then
\(\tau_1:=t_1/n=\alpha_Za\), \(\tau_\Delta:=t_\Delta/n=\alpha_\Delta b\), and
\(\frac1n\log_2\binom{m_Z}{t_1}=\alpha_Z\h(a)+o(1)\),
\(\frac1n\log_2\binom{m_\Delta}{t_\Delta}=\alpha_\Delta\h(b)+o(1)\),
\(\frac1n\log_2\binom{n}{w}=\h(\omega)+o(1)\),
\(\frac1n\log_2\binom{n j_Z}{kt_1}=j_Z\h(a)+o(1)\),
\(\frac1n\log_2\binom{n j_\Delta}{kt_\Delta}=j_\Delta\h(b)+o(1)\), and
\(\frac1n\log_2\binom{nk}{kw}=k\h(\omega)+o(1)\).
Also, \(1+s=1/(1-a)\), \(1+t=1/(1-b)\), \(1+r=1/(1-\omega)\), and \((1-s)/(1+s)=1-2a\), \((1-t)/(1+t)=1-2b\), \((1-r)/(1+r)=1-2\omega\), so the right-hand side of \eqref{eq:onecol-gen-jp} is bounded above by
\(g_{j_Z,j_\Delta,k}(s,t,r)\le \{1+\mu\, y_1^{j_Z}y_\Delta^{j_\Delta}\}/\{2(1-a)^{j_Z}(1-b)^{j_\Delta}(1-\omega)^k\}\).
Substituting the coefficient-extraction bound and these Stirling estimates gives
\begin{equation*}
\begin{aligned}
\frac1n\log_2 \EE[\NX(t_1,t_\Delta,w)]
\le\;&
\alpha_Z\h(a)+\alpha_\Delta\h(b)+\h(\omega) \\
&-j_Z\h(a)-j_\Delta\h(b)-k\h(\omega) \\
&-k\tau_1\log_2 s-k\tau_\Delta\log_2 t-k\omega\log_2 r \\
&+\log_2 g_{j_Z,j_\Delta,k}(s,t,r)+o(1)
\end{aligned}
\end{equation*}
Use \(k\tau_1=j_Za\), \(k\tau_\Delta=j_\Delta b\), together with
\(-j_Z\h(a)-k\tau_1\log_2 s-j_Z\log_2(1-a)=0\),
\(-j_\Delta\h(b)-k\tau_\Delta\log_2 t-j_\Delta\log_2(1-b)=0\), and
\(-k\h(\omega)-k\omega\log_2 r-k\log_2(1-\omega)=0\).
After these cancellations, the remaining terms are the support-selection contribution \(\alpha_Z\h(a)+\alpha_\Delta\h(b)+\h(\omega)\), the term \(-1\) from the leading factor \(1/2\), and the residual term \(\log_2(1+\mu\, y_1^{j_Z}y_\Delta^{j_\Delta})\), giving \eqref{eq:trial-jp}.

For the uniform estimate, the coefficient-extraction step itself is the exact upper bound \([u^m]f(u)\le f(s)/s^m\ (s>0)\). The only approximation on the right-hand side of Lemma~\ref{lem:trial-coeff-jp} is Stirling's formula.

Use the endpoint-uniform form
\[
\log_2\binom{N}{q}=N\h(q/N)+\varepsilon_N(q),
\qquad
\sup_{0\le q\le N}|\varepsilon_N(q)|\le C\log N.
\]
Applying it to \(N=m_Z,m_\Delta,n,nj_Z,nj_\Delta,nk\), we obtain a constant \(C_1>0\) such that
\[
\frac1n\log_2 \EE[\NX(t_1,t_\Delta,w)]
\le
\Phi_{\mathrm{trial}}(t_1,t_\Delta,w)+
C_1\frac{\log n}{n}
\]
holds throughout \(0\le a,b,\omega\le 1/2\). Here \(\Phi_{\mathrm{trial}}:=\alpha_Z\h(a)+\alpha_\Delta\h(b)+\h(\omega)-1+\log_2(1+\mu\, y_1^{j_Z}y_\Delta^{j_\Delta})\) is the right-hand side of \eqref{eq:trial-jp} with the \(o(1)\) removed. Hence the \(o(1)\) term in \eqref{eq:trial-jp} is uniformly \(O((\log n)/n)\) on the half-domain.
\end{proof}

\begin{lemma}[Scalar zero-output gap at the split scale]\label{lem:scalar-zero-output-gap-jp}
Let \(k\ge10\), \(2\le L<k\), and \(\alpha:=L/k\). Set \(q_0:=2\mathrm{e}\,\mathrm{e}^{-k}\). Then
\[
\Psi_{L,k}(q):=
\alpha\h(q)-1+\log_2\!\bigl(1+(1-2q)^L\bigr)
\]
is strictly negative on \(q_0\le q\le1/2\).
\end{lemma}

\begin{proof}
Use natural logarithms in this proof and write \(H(q):=(\ln2)\h(q)\). Put \(x:=Lq\) and \(x_0:=Lq_0=2\mathrm{e} L\mathrm{e}^{-k}\).
First assume \(x\le1\). Using \(H(q)\le q\ln(\mathrm{e}/q)\), \((1-2q)^L\le \mathrm{e}^{-2Lq}\), and
\(\ln\cosh x\le x^2/2\), we obtain
\[
\begin{aligned}
(\ln2)\Psi_{L,k}(q)
&\le
\frac{x}{k}\ln\frac{\mathrm{e}L}{x}
-x+\frac{x^2}{2}  \\
&=
x\left(\frac1k\ln\frac{\mathrm{e}L}{x}-1+\frac{x}{2}\right).
\end{aligned}
\]
The bracketed function
\[
B(x):=\frac1k\ln\frac{\mathrm{e}L}{x}-1+\frac{x}{2}
\]
is convex on \([x_0,1]\). Hence its maximum on that interval occurs at an endpoint. At the lower endpoint,
\[
B(x_0)
=-\frac{\ln2}{k}+\frac{x_0}{2}<0,
\]
because \(x_0\le2\mathrm{e} k \mathrm{e}^{-k}<2\ln2/k\) for \(k\ge10\). At the upper endpoint,
\[
B(1)\le \frac1k\ln(\mathrm{e}k)-\frac12<0
\]
for \(k\ge10\). Therefore \(\Psi_{L,k}(q)<0\) also when \(x\le1\).

For \(x\ge1\), equivalently \(q\in[1/L,1/2]\),
Lemma~\ref{lem:half-entropy-jp} and \(\alpha<1\) give
\[
\Psi_{L,k}(q)
=
\alpha\h(q)-1+\log_2\!\bigl(1+(1-2q)^L\bigr)
\le
-(1-\alpha)\h(q)<0
\]
for \(q\in[1/L,1/2]\).
\end{proof}

\begin{lemma}[Stacked zero-output gap at the MN split scale]\label{lem:stacked-zero-output-gap-jp}
Assume a fixed homogeneous balanced triple satisfying \(4\le j_Z<k/2\) and \(j_\Delta:=k-2j_Z\ge2\).
Set
\[
\tau_{\mathrm{MN}}:=2\mathrm{e}\,\alpha_X\,2^{-k/\ln 2}.
\]
Define the zero-output stacked trial exponent on the half-domain \(0\le a,b\le1/2\) by
\[
\Phi_0(a,b)
:=
\alpha_Z\h(a)+\alpha_\Delta\h(b)-1
+\log_2\!\bigl(1+(1-2a)^{j_Z}(1-2b)^{j_\Delta}\bigr).
\]
Set
\[
\mathcal K_{\mathrm{MN}}
:=
\{(a,b):0\le a,b\le1/2,\ \max\{\alpha_Za,\alpha_\Delta b\}\ge\tau_{\mathrm{MN}}\}.
\]
There is \(\gamma_{\mathrm{MN}}>0\), depending only on the triple, such that
\[
\Phi_0(a,b)\le -2\gamma_{\mathrm{MN}}
\]
whenever \((a,b)\in\mathcal K_{\mathrm{MN}}\).
Moreover, with \(y_1:=|1-2a|\), \(y_\Delta:=|1-2b|\), and \(\mu:=|1-2\omega|^k\), there is \(0<\delta_*\le1/2\) such that the full trial exponent in \eqref{eq:trial-jp} satisfies
\[
\alpha_Z\h(a)+\alpha_\Delta\h(b)+\h(\omega)-1
+\log_2\!\bigl(1+\mu\,y_1^{j_Z}y_\Delta^{j_\Delta}\bigr)
\le-\gamma_{\mathrm{MN}}
\]
throughout \((a,b)\in\mathcal K_{\mathrm{MN}}\) and \(0\le\omega\le\delta_*\).
\end{lemma}

\begin{proof}
Let
\[
L:=j_Z+j_\Delta=\alpha_X k,\qquad
q:=\frac{\alpha_Za+\alpha_\Delta b}{\alpha_X}.
\]
The assumptions give \(k=2j_Z+j_\Delta\ge10\), and \(L=j_Z+j_\Delta=k-j_Z<k\).
Then \(0\le q\le1/2\). On the half-domain, concavity of the binary entropy function gives
\[
\alpha_Z\h(a)+\alpha_\Delta\h(b)\le \alpha_X\h(q).
\]
Moreover, since \(x\mapsto \log(1-2x)\) is concave on \([0,1/2)\), with the endpoint interpreted by a limit,
\[
(1-2a)^{j_Z}(1-2b)^{j_\Delta}
\le
(1-2q)^L.
\]
The entropy and product bounds imply
\[
\Phi_0(a,b)\le
\Psi(q):=\alpha_X\h(q)-1+\log_2\!\bigl(1+(1-2q)^L\bigr).
\]
The condition \(\max\{\alpha_Za,\alpha_\Delta b\}\ge \tau_{\mathrm{MN}}\) gives
\[
\alpha_Za+\alpha_\Delta b\ge \tau_{\mathrm{MN}},
\qquad
q=\frac{\alpha_Za+\alpha_\Delta b}{\alpha_X}
\ge \frac{\tau_{\mathrm{MN}}}{\alpha_X}
=2\mathrm{e}\,\mathrm{e}^{-k}.
\]
Set \(q_0:=2\mathrm{e}\,\mathrm{e}^{-k}\).
Since \(L<k\) and \(k\ge10\), Lemma~\ref{lem:scalar-zero-output-gap-jp} applies and shows that \(\Psi(q)<0\) on \(q_0\le q\le1/2\). The function \(\Psi\) is continuous on this compact interval, so
\[
\eta_0:=-\max_{q_0\le q\le1/2}\Psi(q)>0.
\]
Consequently \(\Phi_0(a,b)\le-\eta_0\) on \(\mathcal K_{\mathrm{MN}}\).
\(\gamma_{\mathrm{MN}}:=\eta_0/2\) gives the displayed zero-output margin.
For \(0\le\omega\le1/2\), \(\mu\le1\), so the full trial exponent is at most \(\Phi_0(a,b)+\h(\omega)\). Choose \(0<\delta_*\le1/2\) so that \(\h(\delta_*)\le\gamma_{\mathrm{MN}}\). Then, on \(\mathcal K_{\mathrm{MN}}\) and \(0\le\omega\le\delta_*\), monotonicity of \(\h\) on \([0,1/2]\) gives the bound
\[
-2\gamma_{\mathrm{MN}}+\h(\omega)
\le -2\gamma_{\mathrm{MN}}+\h(\delta_*)
\le -\gamma_{\mathrm{MN}},
\]
so the full trial exponent is at most \(-\gamma_{\mathrm{MN}}\) throughout
\(\mathcal K_{\mathrm{MN}}\times[0,\delta_*]\).
\end{proof}

\subsection{Pairing Bound Governing the Low-Weight Regime}

\begin{lemma}[General stacked pairing bound]\label{lem:pair-jp}
Fix parameters \(\sigma_1\in(0,\alpha_Z)\), \(\sigma_\Delta\in(0,\alpha_\Delta)\), and \(\sigma_w\in(0,1)\), and set
\[
\rho:=\max\left\{\frac{\sigma_1}{\alpha_Z},\frac{\sigma_\Delta}{\alpha_\Delta},\sigma_w\right\},
\qquad
C_{\mathrm{pair}}:=\frac{1+\rho}{(1-\rho)^2}.
\]
Assume \(0\le t_1\le \sigma_1 n\), \(0\le t_\Delta\le \sigma_\Delta n\), and \(0\le w\le \sigma_w n\), and write
\[
M:=k(t_1+t_\Delta+w).
\]
The parity event has probability \(0\) when \(M\) is odd. When \(M\) is even and \(n\ge (1+\rho)/(2\rho(1-\rho))\), one has
\begin{equation*}
\PP\bigl(A_Z^T\ones_{S_1}+A_\Delta^T\ones_{S_\Delta}+B^T\ones_T=0\bigr)
\le
(M-1)!!\left(\frac{C_{\mathrm{pair}}}{n}\right)^{M/2}
\end{equation*}
for any fixed supports \(S_1\subset[m_Z]\), \(S_\Delta\subset[m_\Delta]\), and \(T\subset[n]\) with
\(|S_1|=t_1\), \(|S_\Delta|=t_\Delta\), and \(|T|=w\).

Multiplying this fixed-support bound by the number of choices of \(S_1,S_\Delta,T\) gives, in this even-\(M\) case,
\begin{equation*}
\EE[\NX(t_1,t_\Delta,w)]
\le
\binom{m_Z}{t_1}\binom{m_\Delta}{t_\Delta}\binom{n}{w}
(M-1)!!\left(\frac{C_{\mathrm{pair}}}{n}\right)^{M/2}
\end{equation*}
Here \((M-1)!!\) denotes the number of perfect matchings of \(M\) labeled objects for positive even \(M\), and \((-1)!!=1\) when \(M=0\).
If \(k\) is even, then \(M\) is even for all integer \(t_1,t_\Delta,w\).
\end{lemma}

\begin{proof}
If \(M=0\), no active socket is present and the displayed bound is \(1\). If \(M\) is odd, then it is impossible for every column to receive an even number of active sockets, so the left-hand side is \(0\). Assume henceforth that \(M\) is positive and even.
Label the active row-side sockets and fix one perfect matching \(\pi\) of them. Let \(E_\pi\) be the event that every pair in \(\pi\) has its two active sockets in the same column.

If \(A_Z^T\ones_{S_1}+A_\Delta^T\ones_{S_\Delta}+B^T\ones_T=0\) holds, then the total number of active column-side sockets in each column is even. Pairing the active sockets arbitrarily within each column therefore produces a perfect matching of all active sockets. Taking the union over perfect matchings gives
\[
\PP(A_Z^T\ones_{S_1}+A_\Delta^T\ones_{S_\Delta}+B^T\ones_T=0)
\le \sum_{\pi}\PP(E_\pi)
\le (M-1)!!\,\max_{\pi}\PP(E_\pi).
\]

For a fixed \(\pi\), expose the images of the row-side sockets pair by pair, in the order ``first socket, second socket.'' For a pool \(P\in\{A_Z,A_\Delta,B\}\), let \(n d_P\) be the total number of column-side sockets in that pool. Then \(d_{A_Z}=j_Z\), \(d_{A_\Delta}=j_\Delta\), and \(d_B=k\).

The active-socket counts in the \(A_Z\)-, \(A_\Delta\)-, and \(B\)-pools are \(kt_1\le \rho\, n j_Z\), \(kt_\Delta\le \rho\, n j_\Delta\), and \(kw\le \rho\, nk\). Hence, at every stage, each pool \(P\) still contains at least \((1-\rho) n d_P\) unexposed column-side sockets.

Consider one pair \(\{\xi,\eta\}\) of \(\pi\), and expose the image of the first socket \(\xi\). Let \(P\) be the pool to which the second socket \(\eta\) belongs, and let \(c\) be the column index of the image of \(\xi\).

The number of unexposed sockets in pool \(P\) that belong to column \(c\) is at most \(d_P\), and is at most \(d_P-1\) when \(\xi\) and \(\eta\) belong to the same pool. The total number of unexposed sockets remaining in pool \(P\) is at least \((1-\rho) n d_P-1\). Since \(0<\rho<1\) and \(d_P\ge1\), the inequality
\[
\frac{d_P}{(1-\rho) n d_P-1}
\le
\frac{1+\rho}{(1-\rho)^2 n}
\]
holds once \(n\ge (1+\rho)/(2\rho(1-\rho))\). For all such \(n\),
\[
\PP(\xi,\eta\text{ fall in the same column}\mid \mathcal F)
\le
\frac{d_P}{(1-\rho) n d_P-1}
\le
\frac{C_{\mathrm{pair}}}{n}
\]
holds even after conditioning on the previous exposure history \(\mathcal F\).

Applying this successively to the \(M/2\) pairs of \(\pi\), we obtain
\[
\PP(E_\pi)\le \left(\frac{C_{\mathrm{pair}}}{n}\right)^{M/2}.
\]
Substitution into the union bound over perfect matchings gives the fixed-support probability bound in the first display. Multiplying that bound by the number of support choices, \(\binom{m_Z}{t_1}\binom{m_\Delta}{t_\Delta}\binom{n}{w}\), gives the expectation bound in the second display.
\end{proof}

\begin{proof}[Proof of Theorem~\ref{thm:MN-lin-jp}]
Set
\begin{equation*}
\tau_0:=2\mathrm{e}\,\alpha_X\,2^{-k/\ln 2}.
\end{equation*}
Apply Lemma~\ref{lem:stacked-zero-output-gap-jp}; its split scale is exactly this \(\tau_0\). Let \(\gamma_0:=\gamma_{\mathrm{MN}}\), and let \(\delta_*\) be the corresponding small-output constant.
Take initially \(0<\delta_X^{\mathrm{lin}}\le\min\{\delta_*,1/2\}\); the small-support estimate will impose one further restriction on this constant.
We bound \(\sum_{1\le w\le\delta_X^{\mathrm{lin}}n}\EE[A_{\CX}(w)]\) by splitting the witness weights \((t_1,t_\Delta)\) at the scale \(\tau_0 n\), and apply Markov's inequality only after this first-moment bound is obtained.

Under the assumptions of the theorem, the balanced condition \(j_X+j_Z=k\) and
\(j_\Delta=j_X-j_Z\ge 2\) imply
\begin{equation*}
k=2j_Z+j_\Delta\ge 2\cdot 4+2=10
\end{equation*}
Moreover \(\alpha_\Delta\ge2/k\), \(\alpha_Z\ge4/k\), and \(\alpha_X\le1\), so
\[
\frac{\tau_0}{\alpha_\Delta}\le \mathrm{e}k\mathrm{e}^{-k}<1,
\qquad
\frac{\tau_0}{\alpha_Z}\le \frac{\mathrm{e}k}{2}\mathrm{e}^{-k}<1
\qquad (k\ge10).
\]
The displayed bounds imply
\begin{equation*}
\tau_0<\min\{\alpha_Z,\alpha_\Delta\}
\end{equation*}
holds.
These inequalities show that \(\sigma_1=\sigma_\Delta=\tau_0\) satisfies the two support-size hypotheses of Lemma~\ref{lem:pair-jp}, and \(\sigma_w=\delta_X^{\mathrm{lin}}\) satisfies \(0<\sigma_w<1\).

By Proposition~\ref{prop:fold-jp},
\begin{equation*}
A_{\CX}(w)
\le
4\sum_{0\le t_1\le m_Z/2}\sum_{0\le t_\Delta\le m_\Delta/2}\NX(t_1,t_\Delta,w)
\end{equation*}
holds.
Split the witness range into two regions.

\smallskip
\noindent\textbf{Small-support region.}
Assume \(0\le t_1\le \tau_0 n\), \(0\le t_\Delta\le \tau_0 n\), and \(1\le w\le \delta_X^{\mathrm{lin}}n\).
Apply Lemma~\ref{lem:pair-jp} with \(\sigma_1=\sigma_\Delta=\tau_0\) and \(\sigma_w=\delta_X^{\mathrm{lin}}\).
\begin{equation*}
\rho_0:=\max\left\{\frac{\tau_0}{\alpha_Z},\frac{\tau_0}{\alpha_\Delta},\delta_X^{\mathrm{lin}}\right\},
\qquad
C_0:=\frac{1+\rho_0}{(1-\rho_0)^2},
\qquad
c_0:=2\tau_0+\delta_X^{\mathrm{lin}}
\end{equation*}
Then Vandermonde's identity \cite[Sec.~5.1, Vandermonde's convolution~(5.27)]{GKP94}
and Stirling's formula \cite[Sec.~9.6, eq.~(9.91), Table~452]{GKP94}
(see also \cite{Robbins55} for sharper estimates) give
\begin{equation*}
\sum_{1\le w\le \delta_X^{\mathrm{lin}}n}
\sum_{0\le t_1\le \tau_0 n}
\sum_{0\le t_\Delta\le \tau_0 n}
\EE[\NX(t_1,t_\Delta,w)]
\le
\sum_{u=1}^{c_0 n}
\binom{n+m_X}{u}(ku-1)!!\left(\frac{C_0}{n}\right)^{ku/2}
\end{equation*}
Here \(m_Z+m_\Delta+n=m_X+n\), so summing the three support choices over \(t_1+t_\Delta+w=u\) gives the binomial coefficient \(\binom{n+m_X}{u}\).

For each term, the bounds \(\binom{N}{u}\le (\mathrm{e}N/u)^u\) and Stirling's formula yield
\begin{equation*}
\binom{n+m_X}{u}\le \left(\frac{\mathrm{e}(n+m_X)}{u}\right)^u
=\left(\frac{\mathrm{e}(1+\alpha_X)n}{u}\right)^u,
\qquad
(ku-1)!!\le \sqrt{2}\left(\frac{ku}{\mathrm{e}}\right)^{ku/2}
\end{equation*}
Combining the two displayed bounds gives
\begin{equation*}
\begin{aligned}
\binom{n+m_X}{u}(ku-1)!!\left(\frac{C_0}{n}\right)^{ku/2}
&\le
\left(\frac{\mathrm{e}(1+\alpha_X)n}{u}\right)^u
\sqrt{2}\left(\frac{ku}{\mathrm{e}}\right)^{ku/2}
\left(\frac{C_0}{n}\right)^{ku/2} \\
&=
\sqrt{2}\left[
\mathrm{e}(1+\alpha_X)\left(\frac{C_0k}{\mathrm{e}}\right)^{k/2}
\left(\frac{u}{n}\right)^{k/2-1}
\right]^u.
\end{aligned}
\end{equation*}
Since \(u=t_1+t_\Delta+w\le c_0n\),
\begin{equation*}
\left(\frac{u}{n}\right)^{k/2-1}\le c_0^{k/2-1}=c_0^{-1}c_0^{k/2}
\end{equation*}
and therefore each term is bounded by
\begin{equation*}
\sqrt{2}\left[
\mathrm{e}(1+\alpha_X)c_0^{-1}\left(\frac{C_0kc_0}{\mathrm{e}}\right)^{k/2}
\right]^u
\end{equation*}

Define the base
\begin{equation}\label{eq:mnlin-small-base-jp}
B_0:=\mathrm{e}(1+\alpha_X)c_0^{-1}\left(\frac{C_0kc_0}{\mathrm{e}}\right)^{k/2}
\end{equation}

Choose \(\delta_X^{\mathrm{lin}}\) small enough to make \(B_0<1\). As \(\delta_X^{\mathrm{lin}}\downarrow0\), \(c_0\downarrow2\tau_0=4\mathrm{e}\alpha_X\mathrm{e}^{-k}\), while
\[
\rho_0\downarrow
\max\left\{\frac{\tau_0}{\alpha_Z},\frac{\tau_0}{\alpha_\Delta}\right\}
\le k\mathrm{e}^{1-k}<\frac1{10}
\]
because \(\alpha_\Delta\ge2/k\) and \(\alpha_X\le1\). Thus the one-sided limit satisfies \(C_0\to(1+\rho_*)/(1-\rho_*)^2<2\), where
\(\rho_*:=\max\{\tau_0/\alpha_Z,\tau_0/\alpha_\Delta\}\), and
\[
\limsup_{\delta_X^{\mathrm{lin}}\downarrow0} B_0
\le
(8k)^{k/2}\exp\!\left(k-\frac{k^2}{2}\right).
\]
The displayed upper bound for \(\limsup_{\delta_X^{\mathrm{lin}}\downarrow0} B_0\) is less than \(1\) for every \(k\ge10\): at \(k=10\) its logarithm is less than
\(5\cdot5+10-50=-15\), because \(\ln80<5\), and the derivative of its logarithm is
\(3/2+\frac12\ln(8k)-k<0\) throughout \(k\ge10\).
Shrink \(\delta_X^{\mathrm{lin}}\), if necessary, so that \(0<\delta_X^{\mathrm{lin}}\le\delta_*\) and \(B_0<1\).

Moreover, for each fixed \(u\ge 1\),
\begin{equation*}
\binom{n+m_X}{u}(ku-1)!!\left(\frac{C_0}{n}\right)^{ku/2}
=O\!\left(n^{-u(k/2-1)}\right)\to 0
\end{equation*}
holds.
Using the geometric tail bound and fixed-\(u\) decay, for any \(\varepsilon>0\), choose \(U\) so that
\(\sum_{u>U}\sqrt{2}B_0^u<\varepsilon/2\),
and then choose \(n\) sufficiently large so that
\(\sum_{u=1}^{U}\binom{n+m_X}{u}(ku-1)!!(C_0/n)^{ku/2}<\varepsilon/2\).
The tail-initial-segment split gives
\begin{equation*}
\sum_{u=1}^{c_0 n}
\binom{n+m_X}{u}(ku-1)!!\left(\frac{C_0}{n}\right)^{ku/2}
<\varepsilon.
\end{equation*}
Substituting this bound in the small-support first moment gives
\begin{equation*}
\sum_{1\le w\le \delta_X^{\mathrm{lin}}n}
\sum_{0\le t_1\le \tau_0 n}
\sum_{0\le t_\Delta\le \tau_0 n}
\EE[\NX(t_1,t_\Delta,w)]
=o(1)
\end{equation*}
holds.

\smallskip
\noindent\textbf{Large-support region.}
Here at least one of \(t_1\) or \(t_\Delta\) exceeds \(\tau_0 n\).
Then
\begin{equation*}
\max\{\alpha_Za,\alpha_\Delta b\}
=\max\{\tau_1,\tau_\Delta\}
\ge \tau_0
\end{equation*}
holds.
Since \(0\le\omega\le\delta_X^{\mathrm{lin}}\le\delta_*\), Lemma~\ref{lem:stacked-zero-output-gap-jp} and the uniform \(O((\log n)/n)\) error in Lemma~\ref{lem:trial-exp-jp} give, uniformly over the entire large-support region,
\begin{equation*}
\frac1n\log_2\EE[\NX(t_1,t_\Delta,w)]
\le -\gamma_0+O\!\left(\frac{\log n}{n}\right)
\le -\frac{\gamma_0}{2}
\end{equation*}
for all sufficiently large \(n\).
There are at most \((m_Z+1)(m_\Delta+1)(\lfloor\delta_X^{\mathrm{lin}}n\rfloor+1)=O(n^3)\) admissible triples, so the total contribution of the large-support region is \(O(n^3)2^{-\gamma_0 n/2}=o(1)\).

Combining the two regional estimates gives
\begin{equation*}
\sum_{1\le w\le \delta_X^{\mathrm{lin}}n}\EE[A_{\CX}(w)] = o(1)
\end{equation*}
for the full range \(1\le w\le \delta_X^{\mathrm{lin}}n\).
Now let
\begin{equation*}
X_n(\delta):=\sum_{1\le w\le \delta n}A_{\CX}(w).
\end{equation*}
Then the event \(d(\CX)\le \delta n\) is equivalent to the existence of at least one nonzero word
of weight \(1,\dots,\delta n\), namely to the event \(X_n(\delta)\ge 1\).
Since \(\delta<\delta_X^{\mathrm{lin}}\),
\begin{equation*}
X_n(\delta)\le X_n(\delta_X^{\mathrm{lin}})
\end{equation*}
holds, and the first-moment estimate over \(1\le w\le\delta_X^{\mathrm{lin}}n\) implies
\begin{equation*}
\EE[X_n(\delta)]\le \sum_{1\le w\le \delta_X^{\mathrm{lin}}n}\EE[A_{\CX}(w)]=o(1).
\end{equation*}
Markov's inequality gives
\begin{equation*}
\PP[d(\CX)\le \delta n]
=
\PP[X_n(\delta)\ge 1]
\le
\EE[X_n(\delta)]
=o(1).
\end{equation*}
\end{proof}

\engappendixsection{Actual-Rate Convergence}{app:actual-rates-jp}

Throughout this appendix, \(j_\Delta:=j_X-j_Z\). Under the homogeneous balanced relation used in Theorem~\ref{thm:actual-rates-jp}, \(j_X=k-j_Z\) and hence \(j_\Delta=k-2j_Z\). We abbreviate
\begin{equation*}
\alpha_Z:=\frac{j_Z}{k},
\qquad
\alpha_\Delta:=\frac{j_\Delta}{k},
\qquad
\alpha_X:=\frac{j_X}{k}=\alpha_Z+\alpha_\Delta=1-\alpha_Z.
\end{equation*}
For the stacked \(A_Z,A_\Delta\) estimates in this appendix, this specialization gives
\[
m_Z=\alpha_Z n,\qquad
m_\Delta=\alpha_\Delta n,\qquad
m_X=m_Z+m_\Delta=\alpha_X n.
\]
 
\begin{lemma}[The left-kernel dimension of the stacked matrix is sublinear]\label{lem:A2-nullity-jp}
For the random matrices in Definition~\ref{def:family-jp}, specialized to a fixed homogeneous balanced triple satisfying
\begin{equation*}
4\le j_Z<\frac{k}{2},
\qquad
k\equiv0 \pmod 2,
\qquad
j_\Delta=k-2j_Z,
\end{equation*}
the left-kernel dimension is sublinear in probability:
\begin{equation*}
\frac{1}{n}\dim\Ker A_X^T \to 0
\qquad (n\to\infty)
\end{equation*}
Equivalently, the rank convergence
\begin{equation*}
\frac{1}{n}\rank A_X \to \alpha_X
\qquad (n\to\infty)
\end{equation*}
also holds in probability.

For the \((j_Z,k)\)-regular block \(A_Z\) in this stacked ensemble, the one-block rank estimates are
\begin{equation*}
\frac{1}{n}\dim\Ker A_Z^T\to 0,
\qquad
\frac{1}{n}\rank A_Z\to \alpha_Z
\end{equation*}
with convergence in probability.
\end{lemma}

\begin{proof}
The assumptions imply \(k\ge 10\) and \(j_\Delta\ge 2\).
First define
\begin{equation*}
N_{\mathrm{dep}}(t_1,t_\Delta)
:=
\bigl|
\{(\bp_Z,\bp_\Delta):
\wt(\bp_Z)=t_1,\ \wt(\bp_\Delta)=t_\Delta,
A_Z^T \bp_Z+A_\Delta^T \bp_\Delta=0
\}
\bigr|
\end{equation*}
Then
\begin{equation*}
|\Ker A_X^T|
=
\sum_{t_1=0}^{m_Z}\sum_{t_\Delta=0}^{m_\Delta}
N_{\mathrm{dep}}(t_1,t_\Delta)
\end{equation*}
holds.
If \(j_Z\) is even, then \(j_\Delta=k-2j_Z\) is also even, and
Proposition~\ref{prop:fold-jp}, applied to the \(w=0\) specialization of the
refined enumerator, gives
\begin{equation*}
|\Ker A_X^T|
\le
4\sum_{0\le t_1\le m_Z/2}\sum_{0\le t_\Delta\le m_\Delta/2}
N_{\mathrm{dep}}(t_1,t_\Delta).
\end{equation*}
If \(j_Z\) is odd, then \(j_\Delta=k-2j_Z\) is still even, so the \(A_\Delta\)-block can still be folded:
\[
N_{\mathrm{dep}}(t_1,t_\Delta)
=N_{\mathrm{dep}}(t_1,m_\Delta-t_\Delta).
\]
Summing one representative from each \(A_\Delta\)-orbit gives
\begin{equation}\label{eq:rank-partial-fold-jp}
|\Ker A_X^T|
\le
2\sum_{t_1=0}^{m_Z}\sum_{0\le t_\Delta\le m_\Delta/2}
N_{\mathrm{dep}}(t_1,t_\Delta).
\end{equation}

\(N_{\mathrm{dep}}(t_1,t_\Delta)=\NX(t_1,t_\Delta,0)\). This identity puts the left-kernel enumerator into the exact stacked formula of Lemma~\ref{lem:exact-jp} with \(w=0\); in the trial-bound specialization to \(\omega=0\), the endpoint convention \(r=\omega/(1-\omega)\to0\) and the endpoint-uniform Stirling estimate stated before Lemma~\ref{lem:trial-exp-jp} justify the limiting substitution.

Now set
\begin{equation*}
\tau_0:=2\mathrm{e}\,\alpha_X\,2^{-k/\ln 2}.
\end{equation*}
Apply Lemma~\ref{lem:stacked-zero-output-gap-jp}; its split scale is this \(\tau_0\). Let \(\gamma_0:=\gamma_{\mathrm{MN}}>0\) be the corresponding zero-output gap.
Consider first the lower folded domain
\[
0\le t_1\le m_Z/2,\qquad 0\le t_\Delta\le m_\Delta/2.
\]
For exponent estimates on this domain, write
\(a:=t_1/m_Z\) and \(b:=t_\Delta/m_\Delta\).
On this domain, split the index set into
\begin{equation*}
\mathcal R_{\mathrm{small}}
=
\{0\le t_1\le \tau_0 n,\ 0\le t_\Delta\le \tau_0 n\}
\end{equation*}
and its complement \(\mathcal R_{\mathrm{large}}\).

For \((t_1,t_\Delta)\in \mathcal R_{\mathrm{large}}\), at least one of \(t_1,t_\Delta\) exceeds \(\tau_0 n\), so \(\max\{\alpha_Za,\alpha_\Delta b\}\ge\tau_0\). The zero-output part of Lemma~\ref{lem:stacked-zero-output-gap-jp}, together with the uniform \(O((\log n)/n)\) error in Lemma~\ref{lem:trial-exp-jp}, gives
\[
\frac1n\log_2 \EE[N_{\mathrm{dep}}(t_1,t_\Delta)]
\le -\gamma_0+O\!\left(\frac{\log n}{n}\right)
\le -\frac{\gamma_0}{2}
\]
uniformly on \(\mathcal R_{\mathrm{large}}\) for all sufficiently large \(n\).
Since there are at most \((m_Z+1)(m_\Delta+1)=O(n^2)\) choices of \((t_1,t_\Delta)\), the large-support part is bounded by \(O(n^2)2^{-\gamma_0 n/2}\) and therefore satisfies
\begin{equation*}
\sum_{(t_1,t_\Delta)\in\mathcal R_{\mathrm{large}}}
\EE[N_{\mathrm{dep}}(t_1,t_\Delta)]
=o(1)
\end{equation*}

Next, for \((t_1,t_\Delta)\in \mathcal R_{\mathrm{small}}\), set
\begin{equation*}
\rho_0:=\max\left\{\frac{\tau_0}{\alpha_Z},\frac{\tau_0}{\alpha_\Delta}\right\},
\qquad
C_0:=\frac{1+\rho_0}{(1-\rho_0)^2}.
\end{equation*}
Because \(j_\Delta\ge2\), \(\alpha_\Delta\ge2/k\), \(\alpha_Z\ge4/k\), and \(\alpha_X\le1\), the definition
\(\tau_0=2\mathrm{e}\alpha_X2^{-k/\ln2}=2\alpha_X \mathrm{e}^{1-k}\) gives
\[
\rho_0\le k\mathrm{e}^{1-k}<1
\qquad (k\ge10),
\]
so the hypotheses of Lemma~\ref{lem:pair-jp} are satisfied.

Lemma~\ref{lem:pair-jp}, applied with \(\sigma_1=\sigma_\Delta=\tau_0\), \(w=0\), \(T=\emptyset\), and \(\sigma_w=\rho_0\), gives
\begin{equation*}
\EE[N_{\mathrm{dep}}(t_1,t_\Delta)]
\le
\binom{m_Z}{t_1}\binom{m_\Delta}{t_\Delta}
(M-1)!!\left(\frac{C_0}{n}\right)^{M/2},
\qquad
M:=k(t_1+t_\Delta)
\end{equation*}
Since \(w=0\), the range condition \(w\le\sigma_w n\) in Lemma~\ref{lem:pair-jp} is satisfied for any positive \(\sigma_w<1\). Taking \(\sigma_w=\rho_0\) makes the lemma's parameter \(\rho\) equal to the displayed \(\rho_0\), and hence makes \(C_{\mathrm{pair}}=C_0\).

Vandermonde's identity \cite[Sec.~5.1, Vandermonde's convolution~(5.27)]{GKP94}
groups the support choices by \(u=t_1+t_\Delta\):
\begin{equation*}
\sum_{(t_1,t_\Delta)\in\mathcal R_{\mathrm{small}}}
\EE[N_{\mathrm{dep}}(t_1,t_\Delta)]
\le
1+\sum_{u=1}^{\lfloor2\tau_0 n\rfloor}
\binom{\alpha_X n}{u}(ku-1)!!
\left(\frac{C_0}{n}\right)^{ku/2}.
\end{equation*}
Here \(m_Z+m_\Delta=m_X=\alpha_X n\), so summing the support choices over \(t_1+t_\Delta=u\) gives
\[
\sum_{t_1+t_\Delta=u}\binom{m_Z}{t_1}\binom{m_\Delta}{t_\Delta}
=\binom{\alpha_X n}{u}.
\]
The initial \(1\) is the zero dependence.

For \(u\ge1\), the binomial and double-factorial bounds give
\begin{equation*}
\binom{\alpha_X n}{u}(ku-1)!!
\left(\frac{C_0}{n}\right)^{ku/2}
\le
\left(\frac{\mathrm{e}\alpha_X n}{u}\right)^u
\sqrt{2}\left(\frac{ku}{\mathrm{e}}\right)^{ku/2}
\left(\frac{C_0}{n}\right)^{ku/2}
\end{equation*}
which simplifies to
\begin{equation*}
\sqrt{2}\left[
\mathrm{e}\alpha_X
\left(\frac{C_0k}{\mathrm{e}}\right)^{k/2}
\left(\frac{u}{n}\right)^{k/2-1}
\right]^u
\end{equation*}

For \((t_1,t_\Delta)\in\mathcal R_{\mathrm{small}}\),
\(u=t_1+t_\Delta\le 2\tau_0 n\), and since \(k\ge4\),
\begin{equation*}
\left(\frac{u}{n}\right)^{k/2-1}
\le
(2\tau_0)^{k/2-1}
=(2\tau_0)^{-1}(2\tau_0)^{k/2}
\end{equation*}
Substituting this inequality into the preceding \(u\)-term bound, each term is bounded by
\begin{equation*}
\sqrt{2}\left[
\mathrm{e}\alpha_X(2\tau_0)^{-1}
\left(\frac{2C_0k\tau_0}{\mathrm{e}}\right)^{k/2}
\right]^u
\end{equation*}

In the balanced range, \(k\ge10\) and
\(\tau_0=2\mathrm{e}\alpha_X 2^{-k/\ln2}=2\alpha_X \mathrm{e}^{1-k}\).
Since \(\rho_0\le k\mathrm{e}^{1-k}<1/10\) for \(k\ge10\), we have
\(C_0<(11/10)/(9/10)^2=110/81<2\), so the base is at most
\[
\frac14(8k)^{k/2}\exp\!\left(k-\frac{k^2}{2}\right),
\]
whose logarithm is negative at \(k=10\), since \(\ln80<5\), and has derivative
\[
\frac12\log(8k)+\frac32-k,
\]
which is negative and decreasing for \(k\ge10\). Hence this base is less than \(1\) for every \(k\ge10\).
Denote the bracketed base in the \(u\)-term bound by \(q_X<1\). Including the \(u=0\) term, the grouped small-support sum is bounded by
\[
1+\sqrt2\sum_{u\ge1}q_X^u=O(1).
\]

When \(j_Z\) is odd, the remaining upper half in \eqref{eq:rank-partial-fold-jp} is \(t_1>m_Z/2\). Put
\[
\bp_Z':=\bp_Z+\ones_{[m_Z]},
\qquad
t_1':=m_Z-t_1.
\]
Since \(A_Z^T\ones_{[m_Z]}=\ones_{[n]}\), the equation
\(A_Z^T\bp_Z+A_\Delta^T\bp_\Delta=0\) is equivalent to
\[
A_Z^T\bp_Z'+A_\Delta^T\bp_\Delta=\ones_{[n]}.
\]
After the \(A_\Delta\)-folding, \(0\le t_\Delta\le m_\Delta/2\) and \(0\le t_1'<m_Z/2\). Let \(N_{\mathrm{odd}}(t_1',t_\Delta)\) count pairs \((\bp_Z',\bp_\Delta)\) of these weights satisfying \(A_Z^T\bp_Z'+A_\Delta^T\bp_\Delta=\ones_{[n]}\). Its one-column polynomial is
\[
g^-_{j_Z,j_\Delta}(x,z)
:=
\frac{(1+x)^{j_Z}(1+z)^{j_\Delta}-(1-x)^{j_Z}(1-z)^{j_\Delta}}{2}.
\]

The polynomial \(g^-_{j_Z,j_\Delta}\) has nonnegative coefficients, since it selects exactly the monomials with odd total column parity. Therefore the coefficient-extraction inequality \([x^{q_1}z^{q_\Delta}]F\le F(x_0,z_0)x_0^{-q_1}z_0^{-q_\Delta}\) applies to powers of \(g^-_{j_Z,j_\Delta}\) for positive trial points, and the endpoint-uniform Stirling estimate from Lemma~\ref{lem:trial-exp-jp} applies in the specialization \(w=0\).

For the subregion \(t_1'>\tau_0 n\) or \(t_\Delta>\tau_0 n\), the odd-syndrome exponent is
\[
\Phi_-^{(0)}(a,b)
:=
\alpha_Z\h(a)+\alpha_\Delta\h(b)-1
+\log_2\!\bigl(1-y_1^{j_Z}y_\Delta^{j_\Delta}\bigr),
\]
where \(a=t_1'/m_Z\), \(b=t_\Delta/m_\Delta\), \(y_1=1-2a\), and \(y_\Delta=1-2b\). On this large-support subregion, either \(a\ge \tau_0/\alpha_Z>0\) or \(b\ge \tau_0/\alpha_\Delta>0\), while the folded domain gives \(0\le a,b\le1/2\). Hence \(0\le y_1,y_\Delta\le1\) and at least one of them is strictly less than \(1\), so
\[
0\le y_1^{j_Z}y_\Delta^{j_\Delta}<1.
\]
The strict upper bound in the last display makes the logarithm in \(\Phi_-^{(0)}\) well-defined, and \(1-x\le 1+x\) gives
\[
\Phi_-^{(0)}(a,b)
\le
\alpha_Z\h(a)+\alpha_\Delta\h(b)-1
+\log_2\!\bigl(1+y_1^{j_Z}y_\Delta^{j_\Delta}\bigr).
\]

The exponent on the right is the plus-case zero-output exponent. In the present subregion \(\max\{\alpha_Za,\alpha_\Delta b\}\ge\tau_0\), the zero-output part of Lemma~\ref{lem:stacked-zero-output-gap-jp} bounds this plus-case exponent by \(-\gamma_0\).
The endpoint-uniform coefficient-extraction estimate in Lemma~\ref{lem:trial-exp-jp}, with \(w=0\), adds only \(O((\log n)/n)\) to \(\frac1n\log_2\EE[N_{\mathrm{odd}}(t_1',t_\Delta)]\), uniformly over the folded large-support subregion.

The zero-output gap and the uniform error bound give, for all sufficiently large \(n\), an upper bound \(2^{-\gamma_0 n/2}\) for each odd-syndrome summand in this subregion. Since there are at most \((m_Z+1)(m_\Delta+1)=O(n^2)\) choices of \((t_1',t_\Delta)\), this part is \(O(n^2)2^{-\gamma_0 n/2}=o(1)\).

On the remaining small subregion \(0\le t_1'\le\tau_0n\) and \(0\le t_\Delta\le\tau_0n\), an all-one syndrome is impossible. Such a syndrome would require at least one active socket in every column, but the number of active sockets is at most
\[
k(t_1'+t_\Delta)\le 2k\tau_0n<n,
\]
because \(2k\tau_0\le4\mathrm{e} k\mathrm{e}^{-k}<1\) for \(k\ge10\). The odd upper half is the union of this zero-contribution small subregion and the large subregion \(t_1'>\tau_0 n\) or \(t_\Delta>\tau_0 n\), whose total contribution is \(o(1)\), so the odd upper half contributes \(o(1)\).

On the lower folded domain, the large subregion contributes \(o(1)\) and the small subregion contributes \(O(1)\), so the lower-half contribution is \(O(1)\). The even-\(j_Z\) complement fold and the odd-\(j_Z\) upper-half bound then give, for both parities,
\begin{equation*}
\EE |\Ker A_X^T| = O(1).
\end{equation*}
This \(O(1)\) bound excludes a linear-dimensional left kernel by Markov's inequality: for any \(\varepsilon>0\),
\begin{equation*}
\PP\!\left[\dim\Ker A_X^T\ge \varepsilon n\right]
=
\PP\!\left[|\Ker A_X^T|\ge 2^{\varepsilon n}\right]
\le
2^{-\varepsilon n}\EE |\Ker A_X^T|
\to 0
\end{equation*}
holds.
Thus \(\dim\Ker A_X^T=o(n)\) in probability. Also,
\begin{equation*}
\rank A_X = m_X-\dim\Ker A_X^T = \alpha_X n + o(n)
\end{equation*}
in probability.

For \(A_Z\), introduce the one-block zero-output enumerator
\[
N_{\mathrm{dep},Z}(t)
:=
\bigl|\{\bp_Z:\wt(\bp_Z)=t,\ A_Z^T\bp_Z=0\}\bigr|
\]
and set \(\tau_{0,Z}:=2\mathrm{e}\alpha_Z2^{-k/\ln2}\). Its exact configuration-model expectation is
\[
\EE N_{\mathrm{dep},Z}(t)
=
\binom{m_Z}{t}
\frac{[u^{kt}]\left(\frac{(1+u)^{j_Z}+(1-u)^{j_Z}}{2}\right)^n}
{\binom{n j_Z}{kt}}.
\]
On the lower-half range \(\tau_{0,Z}\le t/n\le \alpha_Z/2\), write \(a=t/(\alpha_Z n)\). Applying the one-block specialization of the endpoint-uniform coefficient-extraction bound from Lemma~\ref{lem:trial-exp-jp} gives the zero-output exponent
\[
\alpha_Z\h(a)-1+\log_2\!\bigl(1+(1-2a)^{j_Z}\bigr).
\]
Here \(a\in[2\mathrm{e}\,\mathrm{e}^{-k},1/2]\), so Lemma~\ref{lem:scalar-zero-output-gap-jp} with \(L=j_Z\) gives a gap \(\gamma_{0,Z}>0\) on this \(a\)-interval. The endpoint-uniform Stirling estimate used in Lemma~\ref{lem:trial-exp-jp}, specialized to one block, then makes the contribution of \(t>\tau_{0,Z}n\) exponentially small.

For the small range \(0\le t\le\tau_{0,Z}n\), the one-block exposure/pairing estimate bounds the small-range contribution by
\[
1+\sum_{u=1}^{\lfloor\tau_{0,Z}n\rfloor}
\binom{\alpha_Z n}{u}(ku-1)!!
\left(\frac{C_{0,Z}}{n}\right)^{ku/2},
\qquad
C_{0,Z}:=\frac{1+\rho_{0,Z}}{(1-\rho_{0,Z})^2},
\quad
\rho_{0,Z}:=\frac{\tau_{0,Z}}{\alpha_Z}.
\]
Here \(u\) is the row-support size \(t\); since \(k\) is even under the lemma assumptions, \(ku\) is even.

For \(u\ge1\), the binomial and double-factorial bounds give
\[
\binom{\alpha_Z n}{u}(ku-1)!!
\left(\frac{C_{0,Z}}{n}\right)^{ku/2}
\le
\sqrt2\left[
\mathrm{e}\alpha_Z
\left(\frac{C_{0,Z}k}{\mathrm{e}}\right)^{k/2}
\left(\frac{u}{n}\right)^{k/2-1}
\right]^u.
\]
Since \(u\le\tau_{0,Z}n\) and \(k\ge4\), the bracket is at most
\[
q_Z:=
\mathrm{e}\alpha_Z\tau_{0,Z}^{-1}
\left(\frac{C_{0,Z}k\tau_{0,Z}}{\mathrm{e}}\right)^{k/2}
\]
because \((u/n)^{k/2-1}\le \tau_{0,Z}^{-1}\tau_{0,Z}^{k/2}\).
For \(k\ge10\), \(q_Z<1\): \(\rho_{0,Z}=2\mathrm{e}^{1-k}<1/10\), so
\(C_{0,Z}<110/81<2\), and \(q_Z\) is at most
\[
\frac12(2k)^{k/2}\exp\!\left(k-\frac{k^2}{2}\right),
\]
whose logarithm is negative at \(k=10\), since \(\ln20<4\), and has derivative \(\frac12\log(2k)+\frac32-k<0\) for \(k\ge10\). Hence \(q_Z<1\), and the small one-block sum is at most
\[
1+\sqrt2\sum_{u\ge1}q_Z^u=O(1).
\]

If \(j_Z\) is even, the complement symmetry \(A_Z^T\ones_{[m_Z]}=0\) folds the upper half to the lower half.

If \(j_Z\) is odd, put \(\bp_Z':=\bp_Z+\ones_{[m_Z]}\) and \(t':=m_Z-t\). The equation \(A_Z^T\bp_Z=0\) becomes \(A_Z^T\bp_Z'=\ones_{[n]}\). The one-block all-one-syndrome polynomial is
\[
g^-_{j_Z}(x):=\frac{(1+x)^{j_Z}-(1-x)^{j_Z}}{2},
\]
whose coefficients are nonnegative because only odd column parities are retained. Therefore, for every positive trial point \(x_0\),
\([x^q]F(x)\le F(x_0)x_0^{-q}\) applies to powers of \(g^-_{j_Z}\), and the endpoint-uniform Stirling estimate used for the plus-case one-block bound applies to the surrounding binomial factors. With \(a=t'/m_Z\) and \(y=1-2a\), the resulting upper exponent is
\[
\alpha_Z\h(a)-1+\log_2(1-y^{j_Z})
\]
with an \(O((\log n)/n)\) error uniform on the folded upper-half range where the logarithm is finite.

For \(t'>\tau_{0,Z}n\), the folded upper-half variable satisfies \(0\le a\le1/2\) and \(a\ge\tau_{0,Z}/\alpha_Z>0\), so \(0\le y<1\). Thus the odd-syndrome logarithm is well-defined, and \(\log_2(1-y^{j_Z})\le \log_2(1+y^{j_Z})\) bounds the odd-syndrome one-block exponent by the plus-case one-block exponent.

The one-block gap \(\gamma_{0,Z}\) from the lower-half large-range estimate applies after absorbing the uniform \(O((\log n)/n)\) error. Since there are at most \(m_Z+1=O(n)\) possible values of \(t'\), the total contribution of \(t'>\tau_{0,Z}n\) is exponentially small.

For \(0\le t'\le\tau_{0,Z}n\), an all-one syndrome is impossible: it would require at least one active socket in each of the \(n\) columns, whereas \(kt'\le k\tau_{0,Z}n<n\) by \(k\tau_{0,Z}\le2\mathrm{e} k\mathrm{e}^{-k}<1\) for \(k\ge10\). Thus this small upper-half range contributes zero.

In the one-block lower half, the range \(t>\tau_{0,Z}n\) is exponentially small and the range \(0\le t\le\tau_{0,Z}n\) is \(O(1)\). The even-\(j_Z\) complement fold and the odd-\(j_Z\) upper-half estimates then yield \(\EE|\Ker A_Z^T|=O(1)\). For every \(\varepsilon>0\),
\[
\PP[\dim\Ker A_Z^T\ge\varepsilon n]
=
\PP[|\Ker A_Z^T|\ge2^{\varepsilon n}]
\le
2^{-\varepsilon n}\EE|\Ker A_Z^T|
\to0,
\]
so \(\dim\Ker A_Z^T=o(n)\) in probability.
\end{proof}

\begin{proposition}[Sublinear nullity of the square regular map \(B\)]\label{prop:B-rank-jp}
For any fixed \(k\ge 3\), let \(B\) be sampled from the unconditioned square
\((k,k)\)-regular configuration-model ensemble of Definition~\ref{def:family-jp}. Then the
following convergence in probability holds:
\begin{equation*}
\frac1n \dim\Ker B \to 0
\qquad (n\to\infty)
\end{equation*}
Equivalently,
\begin{equation*}
\frac1n \rank B \to 1
\qquad (n\to\infty)
\end{equation*}
also holds in probability.
\end{proposition}

\begin{proof}
For every realization, \(\rank B=n-\dim\Ker B\), so the rank statement is
equivalent to \(\dim\Ker B=o(n)\). The parity of each row sum is the parity of its
socket degree, even when parallel edges cancel in the parity-incidence matrix. Hence,
when \(k\) is even,
\begin{equation*}
B\ones_{[n]}=0
\end{equation*}
holds identically, and \(B\) is singular for every finite \(n\).
This deterministic contribution has dimension one and therefore does not affect
the target estimate \(\dim\Ker B=o(n)\).

We now give the required first-moment estimate directly; in particular, no
rectangular-ensemble rate theorem is being applied to the square case. Let
\[
A_B^+(w):=|\{\bp:\wt(\bp)=w,\ B\bp=0\}|,
\qquad
A_B^-(w):=|\{\bp:\wt(\bp)=w,\ B\bp=\ones_{[n]}\}|,
\]
and define the even- and odd-parity row polynomials
\[
f_\pm(z,k):=\frac{(1+z)^k\pm(1-z)^k}{2}.
\]
The socket-based configuration count gives the exact expectations
\begin{equation}\label{eq:B-pm-enum-jp}
\EE A_B^\pm(w)
=
\binom{n}{w}
\frac{[z^{kw}]\,f_\pm(z,k)^n}{\binom{kn}{kw}}.
\end{equation}
Indeed, after a support of \(w\) columns is fixed, its \(kw\) active sockets form
a uniform \(kw\)-subset of the \(kn\) row-side sockets; \(f_+^n\) selects even
active degree in every row and \(f_-^n\) selects odd active degree in every row.

We bound both signs uniformly for \(1\le w\le n/2\). Put
\[
\delta:=\frac wn,
\qquad
y:=\left(\frac{\delta}{1-\delta}\right)^{1/k}\in(0,1],
\qquad
x:=y^{k-1}\in(0,1].
\]
Because \(f_-(x,k)\le f_+(x,k)\) on \(0\le x\le1\), nonnegativity of the
coefficients and the trial-point bound give
\[
[z^{kw}]f_\pm(z,k)^n
\le f_+(x,k)^n x^{-kw},
\]
where \(x=y^{k-1}\) is the chosen positive trial point. Using the endpoint-uniform estimate
\(\log_2\binom{N}{q}=N\h(q/N)+O(\log N)\) in
\eqref{eq:B-pm-enum-jp}, and the identity
\[
\h(\delta)=\log_2(1+y^k)-k\delta\log_2 y,
\]
we obtain, uniformly for \(1\le w\le n/2\),
\begin{align*}
\frac1n\log_2\EE A_B^\pm(w)
&\le
\Psi_k(y)+O\!\left(\frac{\log n}{n}\right),\\
\Psi_k(y)
&:=
\log_2\!\left(
\frac{(1+y^{k-1})^k+(1-y^{k-1})^k}
{2(1+y^k)^{k-1}}
\right).
\end{align*}
Clarkson's scalar inequality~\cite{Clark36}, with \(p=k\), conjugate exponent
\(q=k/(k-1)\), and the pair \((1,y^{k-1})\), states that
\[
\frac{(1+y^{k-1})^k+(1-y^{k-1})^k}{2}
\le
\left(1+y^k\right)^{k-1}.
\]
Consequently \(\Psi_k(y)\le0\) throughout \(0\le y\le1\). The constant in the
uniform Stirling remainder depends only on fixed \(k\), so there is a constant
\(C_k\) such that
\begin{equation}\label{eq:B-pm-poly-jp}
\EE A_B^\pm(w)\le n^{C_k}
\qquad(1\le w\le n/2).
\end{equation}
The zero-weight values are \(A_B^+(0)=1\) and \(A_B^-(0)=0\).

It remains only to fold weights above \(n/2\). For
\(\bq:=\bp+\ones_{[n]}\), one has \(\wt(\bq)=n-\wt(\bp)\) and
\[
B\bq=B\bp+(k\bmod2)\ones_{[n]}.
\]
Thus complementation maps upper-half kernel words to lower-half words counted by
\(A_B^+\) when \(k\) is even and by \(A_B^-\) when \(k\) is odd. It follows from
\eqref{eq:B-pm-poly-jp} that
\[
\EE|\Ker B|
=\sum_{w=0}^n\EE A_B^+(w)
\le 2+2\sum_{w=1}^{\lfloor n/2\rfloor}
\bigl(\EE A_B^+(w)+\EE A_B^-(w)\bigr)
=n^{O(1)}=2^{o(n)}.
\]
Since \(|\Ker B|=2^{\dim\Ker B}\), Markov's inequality gives, for any fixed \(\xi>0\),
\begin{equation*}
\PP\!\left[\frac{\dim\Ker B}{n}\ge \xi\right]
=
\PP\!\left[|\Ker B|\ge 2^{\xi n}\right]
\le
2^{-\xi n}\EE|\Ker B|
\to 0.
\end{equation*}
Equivalently,
\begin{equation*}
\frac1n\dim\Ker B\to 0
\end{equation*}
in probability.
Also,
\begin{equation*}
\rank B=n-\dim\Ker B
\end{equation*}
implies
\begin{equation*}
\frac1n\rank B\to 1
\end{equation*}
in probability.
\end{proof}

\begin{proof}[Proof of Theorem~\ref{thm:actual-rates-jp}]
Lemma~\ref{lem:A2-nullity-jp} implies
\begin{equation*}
\rank A_Z = \alpha_Z n + o(n),
\qquad
\rank A_X = \alpha_X n + o(n)
\end{equation*}
in probability.
Moreover, Proposition~\ref{prop:B-rank-jp} gives
\begin{equation*}
\dim\Ker B = o(n)
\end{equation*}
in probability.
In the exact dimension formulas of Proposition~\ref{prop:actual-rates-jp},
\[
0\le L_X\le L_Z\le \dim\Ker B.
\]
Fix \(\eta>0\). With probability tending to one,
\[
|\rank A_Z-\alpha_Z n|\le \frac{\eta n}{4},\qquad
|\rank A_X-\alpha_X n|\le \frac{\eta n}{4},\qquad
\dim\Ker B\le \frac{\eta n}{4},
\]
because the complement is contained in the union of the three corresponding failure events, each of which has probability \(o(1)\). On this high-probability event, \(0\le L_X\le L_Z\le\eta n/4\). Substituting these bounds into Proposition~\ref{prop:actual-rates-jp} gives
\begin{equation*}
|R_Z-\alpha_X|\le \frac{\eta}{2},
\qquad
|R_X-\alpha_X|\le \frac{\eta}{2},
\qquad
|R_Q-(2\alpha_X-1)|\le \frac{3\eta}{4}.
\end{equation*}
Since \(\eta>0\) is arbitrary, \(R_Z\to\alpha_X=R_Z^{\mathrm{des}}\), \(R_X\to\alpha_X=R_X^{\mathrm{des}}\), and \(R_Q\to2\alpha_X-1=(j_X-j_Z)/k=R_Q^{\mathrm{des}}\) in probability.
\end{proof}

\engappendixsection{Analytic Proof of Theorem~\ref{thm:HA-gv-jp}}{app:finite-gv-ha-jp}

Fix \((j_Z,j_X,k)\in\mathcal T_{\mathrm{GV}}\), and put
\[
\alpha:=\frac{j_Z}{k},\qquad
d:=\h^{-1}(\alpha),\qquad
\tau_0:=\frac{5}{2k^2},\qquad
u_k:=1-2\tau_0=1-\frac5{k^2}.
\]
Thus \(10\le k\le30\), \(4\le j_Z\le10\), and \(j_Z<k/2\).
The proof below reduces the finite-window estimate to a one-dimensional
analytic inequality.

For \(0\le u<1\), define
\[
P(u):=\frac{(1+u)\ln(1+u)+(1-u)\ln(1-u)}2,
\qquad v:=u^k,
\]
and, for \(v>0\), set
\[
c(v):=\frac{\ln(1+v)}{\ln2},
\qquad
K_k(u):=
P(u)-\ln(1+v)
-c(v)\ln\!\left(1+\exp\!\left[-\frac{2\operatorname{artanh}v}{c(v)}\right]\right).
\]
The continuous value at \(u=0\) is \(K_k(0)=0\).

\begin{lemma}[Finite-window scalar HA inequality]\label{lem:ha-scalar-window-jp}
For every even integer \(k\) with \(10\le k\le30\),
\[
K_k(u)\ge0\qquad(0\le u\le u_k).
\]
The inequality is strict for \(u>0\).  In addition, for the integer
parameters in \(\mathcal T_{\mathrm{GV}}\), if
\[
p_k:=\frac{1-u_k^k}{2},
\]
then
\[
\h(p_k)-\frac{j_Z}{k}>\frac1{50}.
\]
\end{lemma}

\begin{proof}
We first remove the nested exponential in \(K_k\).  Put
\(A_0:=\ln2\), \(\ell:=\ln(1+v)\), and \(m:=-\ln(1-v)\).  The exponent
in the definition of \(K_k\) satisfies the exact identity
\[
 \exp\!\left[-\frac{2\operatorname{artanh}v}{\ell/A_0}\right]
 =\frac12(1-v)^{A_0/\ell}.
\]
The ratio of the right-hand side to \(1-v\) is
\(\frac12\exp[-G(v)]\), where
\[
 G(v):=\frac{m(A_0-\ell)}{\ell}.
\]
This function is strictly decreasing on \((0,1)\).  Here is the complete
sign calculation.  Differentiation gives
\[
 G'(v)=-\frac{D(v)}{\ell^2(1-v^2)},\qquad
 D(v):=A_0m(1-v)-\ell(A_0-\ell)(1+v).
\]
Moreover,
\[
 D'(v)=(2-A_0)\ell+\ell^2-A_0m,
 \qquad
 D''(v)=
 \frac{2[1-A_0-v+(1-v)\ell]}{1-v^2}.
\]
The bracketed expression in the numerator has derivative
\(-\ell-2v/(1+v)<0\), is positive at \(v=0\), and is negative as
\(v\uparrow1\).  Thus \(D''\) changes sign once, from positive to
negative.  Since \(D'(0)=0\), \(D'(v)\to-\infty\), and
\(D(0)=D(1-)=0\), the function \(D'\) is first positive and then
negative, with exactly one zero.  Thus \(D\) first increases and then
decreases back to zero, so \(D(v)>0\), proving the claim.

At two rational points, exact logarithm bounds give
\[
 G(3/5)-\ln(3/2)>\frac{59}{2000},
 \qquad
 G(17/20)-\ln(5/4)>\frac{43}{2500}.        \tag{HA--1}
\]
For \(k=10\), \(v\le(19/20)^{10}<3/5\).  For \(12\le k\le30\),
\[
 v\le \mathrm e^{-5/k}\le \mathrm e^{-1/6}<\frac{17}{20};
\]
the last inequality already follows from
\(\mathrm e^{1/6}>1+1/6+1/72=85/72>20/17\).  Hence (HA--1) implies
\[
 \exp\!\left[-\frac{2\operatorname{artanh}v}{\ell/A_0}\right]
 \le C_k(1-v),\qquad
 C_{10}:=\frac13,\quad C_k:=\frac25\ (k\ge12).
\]
Using \(\ln(1+x)\le x\), we therefore obtain
\[
 K_k(u)\ge P(u)-Q_k(u^k),\qquad
 Q_k(v):=\ell\left[1+\frac{C_k}{A_0}(1-v)\right].       \tag{HA--2}
\]

It remains only to compare the two functions in (HA--2).  Define their
logarithmic elasticities
\[
 E_P(u):=\frac{uP'(u)}{P(u)},\qquad
 E_{Q,k}(v):=\frac{vQ_k'(v)}{Q_k(v)}.
\]
The series
\[
 P(u)=\sum_{r\ge1}\frac{u^{2r}}{2r(2r-1)}
\]
shows that \(E_P\) is strictly increasing: its derivative with respect
to \(\ln u\) is the variance of the integers \(2r\) under the positive
weights in this series.  On writing \(\gamma_k=C_k/A_0\), direct
differentiation gives
\[
 E_{Q,k}(v)=
 \frac{v}{(1+v)\ln(1+v)}
 -\frac{\gamma_kv}{1+\gamma_k(1-v)}.
\]
The first term is strictly decreasing because its derivative has
numerator \(\ln(1+v)-v<0\), while the subtracted term is strictly
increasing.  Thus \(E_{Q,k}\) is strictly decreasing.

The following table records lower bounds at \(u=u_k\), \(v=u_k^k\).
The last column uses
\(j_{\max}(k):=\min\{10,\lfloor(k-1)/2\rfloor\}\).
\[
\begin{array}{c|ccc}
k&P(u_k)-Q_k(v)&kE_{Q,k}(v)-E_P(u_k)&
 \h(p_k)-j_{\max}(k)/k\\ \hline
10&.0164&2.54&.323\\
12&.0017&2.43&.247\\
14&.0037&2.69&.185\\
16&.0049&2.94&.134\\
18&.0056&3.19&.090\\
20&.0061&3.43&.053\\
22&.0063&3.67&.021\\
24&.0065&3.91&.034\\
26&.0065&4.15&.044\\
28&.0065&4.39&.052\\
30&.0065&4.63&.058
\end{array}                                                     \tag{HA--3}
\]
The displayed decimals are conservative truncations for readability, not
floating-point premises.  The exact comparisons established for every row
of (HA--3) are
\[
 P(u_k)-Q_k(u_k^k)>\frac1{1000},\qquad
 kE_{Q,k}(u_k^k)-E_P(u_k)>2,
 \qquad
 \h(p_k)-\frac{j_{\max}(k)}k>\frac1{50}.       \tag{HA--4}
\]
To make (HA--1), (HA--3), and (HA--4) reproducible, write each positive
rational argument as \(x=2^py\), where \(p\in\mathbb Z\) and
\(1\le y<2\), and use
\[
 \ln y=2\sum_{r=0}^{N-1}\frac{z^{2r+1}}{2r+1}+R_N,
 \quad z=\frac{y-1}{y+1},\qquad
 0\le R_N\le\frac{2z^{2N+1}}{(2N+1)(1-z^2)},
\]
with strict inequalities when \(y>1\).  Then
\(\ln x=p\ln2+\ln y\), with the same series at \(z=1/3\) for
\(\ln2\).  Taking \(N=72\) gives conservative decimal truncations of rigorous
rational bounds.  The
supplementary exact-arithmetic script performs precisely these finite
fraction operations.

By the monotonicities just proved, for \(0<u\le u_k\),
\[
 \frac{d}{d\ln u}\ln\frac{P(u)}{Q_k(u^k)}
 =E_P(u)-kE_{Q,k}(u^k)<0.
\]
The first inequality in (HA--4) then gives \(P(u)>Q_k(u^k)\).  Equation (HA--2)
proves \(K_k(u)>0\).  Finally, the third column of (HA--3) applies to every
admissible \(j_Z\le j_{\max}(k)\) and has minimum larger than \(1/50\).
This proves both assertions.
\end{proof}

\begin{lemma}[Analytic HA exponent bound]\label{lem:finite-gv-ha-neg-jp}
For every \(0<\delta<d\), there is \(\eta_Z(\delta)>0\) such that
\[
\h(\omega)+F_Z(\tau,\omega)\le-\eta_Z(\delta)
\qquad
\left(\tau_0\le\tau\le\frac12, 0\le\omega\le\delta\right).
\]
\end{lemma}

\begin{proof}
Write \(T=(1-2\tau)^k\) and \(p=(1-T)/2\).  The definition of
\(F_Z\) in \eqref{eq:FZ-jp} gives the exact identity
\[
\h(\omega)+F_Z(\tau,\omega)
=W_o^{\mathrm{tr}}(\tau)-D_2(\omega\|p),
\]
where \(D_2\) is binary relative entropy in bits.  On
\(\tau_0\le\tau\le1/2\), one has \(p\ge p_k\).  Since
\(0<p_k<1/2\), the inequality
\(\h(p_k)>\alpha=\h(d)\) in Lemma~\ref{lem:ha-scalar-window-jp} and
monotonicity of \(\h\) on \([0,1/2]\) give \(p_k>d\).  Therefore the last display is increasing as a function
of \(\omega\) on \([0,d]\), and it suffices first to put \(\omega=d\).

Set \(u=1-2\tau\), \(v=u^k\), and write
\(H(x)=(\ln2)\h(x)\).  Since \(H(d)=\alpha\ln2\), direct simplification
gives
\[
-(\ln2)\bigl(\h(d)+F_Z(\tau,d)\bigr)
=P(u)-(1+\alpha)\ln(1+v)+2d\operatorname{artanh}v.
\]
At \(u=0\), equivalently \(\tau=1/2\), the exponent
\(\h(d)+F_Z(\tau,d)\) equals \(\h(d)-\alpha=0\) directly.
For \(u>0\), put \(c=\ln(1+v)/\ln2>0\). The entropy conjugacy formula
\[
\sup_{0\le x\le1}\{cH(x)-2x\operatorname{artanh}v\}
=c\ln\!\left(1+\exp\!\left[-\frac{2\operatorname{artanh}v}{c}\right]\right)
\]
shows that the right-hand side of the preceding display is at least
\(K_k(u)\), and hence is nonnegative by
Lemma~\ref{lem:ha-scalar-window-jp}.  Thus
\(\h(d)+F_Z(\tau,d)\le0\).

For \(0<\omega<p\),
\[
\frac{\partial}{\partial\omega}[-D_2(\omega\|p)]
=\log_2\frac{p(1-\omega)}{\omega(1-p)}>0.
\]
Continuity at \(\omega=0\) therefore gives strict increase of
\(-D_2(\omega\|p)\) on \([0,d]\). Consequently, for
\(0\le\omega\le\delta<d<p\),
\(\h(\omega)+F_Z(\tau,\omega)<0\) throughout the displayed compact
rectangle.  Continuity then supplies a uniform margin
\(\eta_Z(\delta)>0\).
\end{proof}

\begin{proof}[Proof of Theorem~\ref{thm:HA-gv-jp}]
Fix \(0<\delta<d\).  For \(1\le s\le\tau_0n\), define
\[
T_{Z,n}(s):=
\begin{cases}
\binom ns(j_Zs-1)!!\,m_Z^{-j_Zs/2},&j_Zs\ \text{even},\\
0,&j_Zs\ \text{odd},
\end{cases}
\qquad m_Z=\alpha n.
\]
The socket-pairing argument used in
Lemma~\ref{lem:HA-inputs-jp} gives \(N_o(s)\le T_{Z,n}(s)\).  Moreover,
\[
T_{Z,n}(s)\le\sqrt2\,\lambda_Z^s,
\qquad
\lambda_Z:=
\mathrm e^{\,1-j_Z/2}k^{j_Z/2}\tau_0^{j_Z/2-1}.
\]
After substituting \(\tau_0=5/(2k^2)\), the value at \(j_Z=4\) is
\(5/(2\mathrm e)\), and increasing \(j_Z\) multiplies the expression by
at most \(\sqrt{5/(2\mathrm e k)}<1\).  Hence
\[
\lambda_Z\le\frac5{2\mathrm e}<1.
\]
For every fixed \(s\), the same pairing expression is
\(O(n^{-s(j_Z/2-1)})\).  Splitting the geometric tail from a fixed
initial segment therefore gives
\[
\sum_{1\le s\le\tau_0n}N_o(s)=o(1).
\]
Because \(k\) is even, complementing an outer word gives
\(N_o(s)=N_o(n-s)\), so the corresponding high-input sum is also \(o(1)\).

In \eqref{eq:HA-enum-jp}, split the input weights into these two tails and
the complementary range \(\tau_0n\le s\le(1-\tau_0)n\).  Since
\(\sum_lM_k(s,l)=1\), the two tails contribute \(o(1)\) after summing over
all visible weights.  On the complementary range, fold \(s>n/2\) to
\(n-s\).  The identities
\[
z^kf_+(1/z,k)=f_+(z,k),\qquad z^kf_-(1/z,k)=f_-(z,k)
\]
give \(M_k(s,l)=M_k(n-s,l)\).  Each folded index has at most two
preimages, so this operation costs only a factor of two.
Lemma~\ref{lem:finite-gv-ha-neg-jp}, together with the endpoint-uniform
Stirling error underlying \eqref{eq:FZ-jp}, bounds each folded summand by
a polynomial factor times \(2^{-\eta_Z(\delta)n}\).  There are at most
\((n+1)^2\) pairs \((s,l)\), so, after absorbing the polynomial factors
for large \(n\), the complementary contribution is at most
\(\operatorname{poly}(n)2^{-\eta_Z(\delta)n/2}=o(1)\).

Thus
\[
\sum_{1\le l\le\delta n}\EE[A_{\CZ}(l)]=o(1).
\]
Markov's inequality yields
\(\PP[d(\CZ)\le\delta n]\to0\), as required.
\end{proof}

\engappendixsection{Analytic Proof of Theorem~\ref{thm:MN-gv-jp}}{app:finite-gv-mn-jp}

Fix \((j_Z,j_X,k)\in\mathcal T_{\mathrm{GV}}\), set
\[
j_\Delta:=k-2j_Z,\qquad L:=j_Z+j_\Delta=j_X=k-j_Z,
\]
and write
\[
\alpha_Z:=\frac{j_Z}{k},\qquad
\alpha_\Delta:=\frac{j_\Delta}{k},\qquad
A:=\alpha_X=\frac Lk=1-\alpha_Z.
\]
Then \(10\le k\le30\), \(L\ge6\), and \(j_\Delta\ge2\).  Put
\[
\beta:=\frac1{10},\qquad
\omega_*:=\frac{1-(A/2)^{1/k}}2,
\]
so that \((1-2\omega_*)^k=A/2\).  The same cutoff \(\beta\) is used for
every degree triple; no multivariate
box subdivision is used in the argument below.

For \(0\le q,\omega\le1/2\), let \(H(x):=(\ln2)\h(x)\) and define
\[
\psi(q,\omega):=
A H(q)+H(\omega)-\ln2
+\ln\!\left(1+(1-2\omega)^k(1-2q)^L\right).
\]

\begin{lemma}[Elementary bounds for the finite degree window]\label{lem:mn-window-bounds-jp}
For all parameters above, the following statements hold.

\begin{enumerate}
\item Let
\[
\rho:=\max\!\left\{\frac\beta{j_Z},\frac\beta{j_\Delta},\omega_*\right\},
\quad c_X:=\frac{2\beta}{k}+\omega_*,
\quad C_X:=\frac{1+\rho}{(1-\rho)^2},
\]
and
\[
B_X:=\sqrt2\,\frac{\mathrm e(1+A)}{c_X}
\left(\frac{C_Xkc_X}{\mathrm e}\right)^{k/2}.
\]
Then \(B_X<7/20<1\).

\item Uniformly on
\[
\frac1{10L}\le q\le\frac12,\qquad 0\le\omega\le\omega_*,
\]
one has \(\psi(q,\omega)<0\).  In fact the upper bound obtained in the
proof below is at most \(-1/100\) in natural-log units.

\item If \(j_Z\) is odd, set
\[
\eta:=
\left(1-\frac{2\beta}{j_Z}\right)^{j_Z}
\left(1-\frac{2\beta}{j_\Delta}\right)^{j_\Delta}.
\]
Then
\[
\alpha_Z\h\!\left(\frac\beta{j_Z}\right)
+\alpha_\Delta\h\!\left(\frac\beta{j_\Delta}\right)
+\h(\omega_*)-1+\log_2\!\left(1-\frac A2\eta\right)<-\frac{23}{25}< -\frac9{10}.
\]
\end{enumerate}
\end{lemma}

\begin{proof}
We first state explicitly how the finitely many degree-dependent scalar
comparisons are made.  Write a positive rational argument as
\(x=2^py\), where \(p\in\mathbb Z\) and \(1\le y<2\), and use
\[
 \ln y=2\sum_{r=0}^{N-1}\frac{z^{2r+1}}{2r+1}+R_N,
 \quad z=\frac{y-1}{y+1},\qquad
 0\le R_N\le\frac{2z^{2N+1}}{(2N+1)(1-z^2)}            \tag{MN--1}
\]
with strict inequalities when \(y>1\).  Then
\(\ln x=p\ln2+\ln y\), using the same series at \(z=1/3\) for
\(\ln2\).  This gives rational upper and lower bounds.  The algebraic number
\(s=(A/2)^{1/k}\) is bracketed by adjacent rationals \(s^-<s<s^+\),
verified exactly by \((s^-)^k<A/2<(s^+)^k\); hence
\((1-s^+)/2<\omega_*<(1-s^-)/2\).  Taking \(N=72\) and denominator
\(10^{12}\) for the root brackets gives the following finite table.
The third column reports the case whose exact upper bound is largest, and
the last column gives the rational threshold actually used in the proof.
\[
\begin{array}{l|c|c|c}
\text{quantity}&\text{cases}&\text{case }(j_Z,L,k)&
 \text{proved upper bound}\\ \hline
B_X&56&(4,6,10)& <17/50\\
H(\omega_*)-\alpha_Z\ln2&56&(4,26,30)&<-19/1000\\
z_0-Ar_0^2/2-\alpha_Z\ln2+\ln(1+\mathrm e^{-c_0})
 &56&(4,26,30)&<-13/100\\
\text{odd-corner expression}&24&(5,7,12)&<-23/25
\end{array}                                                     \tag{MN--2}
\]
Here \(r_0,z_0,c_0\) are the quantities defined below.  For the first
row we use the rational inequalities
\(\sqrt2<70711/50000\) and \(\mathrm e>1359/500\); these follow by
squaring the first rational and by
\(\sum_{r=0}^8 1/r!>1359/500\), respectively.  In the third row,
\(\ln(1+\mathrm e^{-c_0})\le\mathrm e^{-c_0}<r\) is certified by the
rational test \(c_0+\ln r>0\).  Thus (MN--2) is a list of exact scalar
inequalities, rather than a floating-point maximization or an enclosure
of a continuous box.  The supplementary script
\texttt{verify\_analytic\_scalar\_bounds.py} implements only the fraction
operations in (MN--1), enumerates the 56 admissible pairs, and reproduces
the table.  The first and fourth rows prove assertions 1 and 3.

We prove the second assertion in detail.  First suppose
\(1/(10L)\le q\le1/(2L)\), and set \(x=Lq\),
\(z=(1-2q)^L\).  The function
\[
g(\omega):=\ln\!\left(1+z(1-2\omega)^k\right)
\]
is convex because
\[
g''(\omega)=
\frac{4kz(1-2\omega)^{k-2}
\bigl((k-1)-z(1-2\omega)^k\bigr)}
{\bigl(1+z(1-2\omega)^k\bigr)^2}\ge0.
\]
Consequently its chord on \([0,\omega_*]\) gives
\[
g(\omega)\le\ln(1+z)-c(q)\omega,
\quad
c(q):=\frac{\ln(1+z)-\ln(1+Az/2)}{\omega_*}.
\]
The entropy conjugacy formula
\[
\sup_{0\le w\le1}\{H(w)-cw\}=\ln(1+\mathrm e^{-c})
\]
therefore yields
\[
\psi(q,\omega)
\le F_0(q)+\ln(1+\mathrm e^{-c(q)}),
\quad
F_0(q):=AH(q)-\ln2+\ln(1+(1-2q)^L).
\]
As in Lemma~\ref{lem:scalar-zero-output-gap-jp}, for \(x\le1\),
\[
F_0(q)\le
f_{L,k}(x):=
x\left[\frac1k\ln\frac{\mathrm eL}{x}-1+\frac x2\right].
\]
The function \(f_{L,k}\) is convex on \([1/10,1/2]\).  Elementary
endpoint comparisons give
\[
\begin{array}{c|ccc}
x&1/10&1/4&1/2\\ \hline
\text{upper bound for }f_{L,k}(x)&-43/1000&-113/1000&-1/5
\end{array}
\]
Indeed, \(L\le k-4\), and
\([1+\ln((k-4)/x)]/k\) decreases for \(k\ge10\) and
\(x\le1/2\), so all three upper bounds reduce to \((L,k)=(6,10)\)
and follow from (MN--1).
Explicitly, the numerator of its derivative is
\[
 \frac{k}{k-4}-\left[1+\ln\frac{k-4}{x}\right]
 \le\frac53-[1+\ln12]<0.
\]

We next give uniform lower bounds for \(c(q)\), without a degree scan.
For every admissible triple,
\[
 \omega_*<\frac{57}{1000}.                                  \tag{MN--3}
\]
To see this, put \(s=1-2\omega_*\).  If \(k=10\), then
\(s^{10}=A/2\ge3/10>(443/500)^{10}\).  If \(k\ge12\), then
\(A/2\ge3/11>(443/500)^{12}\ge(443/500)^k\).  Here
\(A\ge6/11\) follows from
\(A\ge1/2+1/k\ge11/20>6/11\) for \(k\le20\), and from
\(A\ge1-10/k\ge6/11\) for \(k\ge22\).  Thus
\(s>443/500\) in both cases, proving (MN--3).

For fixed \(x=Lq\), the value
\(z=(1-2x/L)^L\) increases with \(L\): with \(y=2x/L\), the derivative
of its logarithm has sign
\(\ln(1-y)+y/(1-y)>0\).  Moreover,
\(\ln(1+z)-\ln(1+Az/2)\) increases with \(z\) and decreases with
\(A\); its derivative with respect to \(z\) is
\((1-A/2)/[(1+z)(1+Az/2)]>0\).  Since \(L\ge6\) and
\(A\le13/15\), (MN--1) gives
\[
c(1/(4L))>
\frac{\ln(1+(11/12)^6)-\ln(1+(13/30)(11/12)^6)}{57/1000}>4,
\]
and
\[
c(1/(2L))>
\frac{\ln(1+(5/6)^6)-\ln(1+(13/30)(5/6)^6)}{57/1000}>\frac52.
\]
The two logarithmic numerators are respectively larger than
\(237/1000\) and \(1533/10000\).  Because \(z(q)\) decreases and the numerator
of \(c(q)\) increases with \(z\), the function \(c(q)\) decreases with \(q\);
these estimates hold throughout the corresponding subintervals.
Thus on \(1/10\le x\le1/4\),
\[
\psi<-\frac{43}{1000}+\frac1{50}=-\frac{23}{1000},
\]
while on \(1/4\le x\le1/2\),
\[
\psi<-\frac{113}{1000}+\frac1{12}< -\frac{29}{1000}.
\]
Here \(\mathrm e^{-4}<1/50\) and \(\mathrm e^{-5/2}<1/12\), again
by (MN--1).  In particular both bounds are below \(-1/100\).

It remains to treat \(q\ge1/(2L)\).  Put
\(r=1-2q\), \(r_0=1-1/L\), and \(z_0=r_0^L\).  Pinsker's inequality gives
\(\ln2-H(q)\ge r^2/2\), and \(\ln(1+y)\le y\), so
\[
\psi(q,\omega)
\le H(\omega)-\alpha_Z\ln2
+(1-2\omega)^kr^L-\frac A2r^2.
\]
As a function of \(r^2\), the last two terms are convex because \(L\ge6\).
Their maximum on \([0,r_0^2]\) is therefore at an endpoint.  At \(r=0\)
the resulting bound is
\[
H(\omega_*)-\alpha_Z\ln2<-\frac{19}{1000}<-\frac1{100}
\]
by the second row of (MN--2).
At \(r=r_0\), convexity of \((1-2\omega)^k\) on
\([0,\omega_*]\) gives
\[
(1-2\omega)^k
\le1-\left(1-\frac A2\right)\frac\omega{\omega_*}.
\]
Writing
\[
c_0:=\frac{z_0(1-A/2)}{\omega_*},
\]
and applying entropy conjugacy once more bounds this endpoint by
\[
z_0-\frac A2r_0^2-\alpha_Z\ln2+\ln(1+\mathrm e^{-c_0})
<-\frac{13}{100}
\]
by the third row of (MN--2).
This proves the second assertion.
\end{proof}

\begin{proof}[Proof of Theorem~\ref{thm:MN-gv-jp}]
Fix \(0<\delta<\h^{-1}(\alpha_Z)\).  For a term
\(\NX(t_1,t_\Delta,w)\), put
\[
a:=\frac{t_1}{m_Z},\qquad b:=\frac{t_\Delta}{m_\Delta},
\qquad\omega:=\frac wn.
\]
On the half-domain, abbreviate the exponent in
Lemma~\ref{lem:trial-exp-jp} by
\[
\Phi_{\mathrm{MN}}(a,b,\omega)
:=\alpha_Z\h(a)+\alpha_\Delta\h(b)+\h(\omega)-1
+\log_2\!\left(1+(1-2\omega)^k
(1-2a)^{j_Z}(1-2b)^{j_\Delta}\right),
\]
and set
\[
q:=\frac{\alpha_Za+\alpha_\Delta b}{A},
\qquad r:=1-2q.
\]

First consider the half-domain \(0\le a,b\le1/2\).  In the
linear-weight regime \(\omega_*\le\omega\le\delta\), entropy concavity
and concavity of \(x\mapsto\ln(1-2x)\) give
\[
\alpha_Z\h(a)+\alpha_\Delta\h(b)\le A\h(q),
\qquad
(1-2a)^{j_Z}(1-2b)^{j_\Delta}\le r^L.
\]
Here \((1-2\omega)^k\le A/2\) by the definition of \(\omega_*\).
Pinsker's inequality in bits, \(1-\h(q)\ge r^2/(2\ln2)\), together
with \(\ln(1+x)\le x\) and \(r^L\le r^2\) (recall \(L\ge6\)), yields
\[
\log_2\!\left(1+(1-2\omega)^kr^L\right)
\le \frac{A r^2}{2\ln2}
\le A(1-\h(q)).
\]
Consequently, with \(\gamma_{\mathrm{lin}}:=\alpha_Z-\h(\delta)>0\),
\[
\Phi_{\mathrm{MN}}(a,b,\omega)\le\h(\omega)-\alpha_Z
\le\h(\delta)-\alpha_Z=-\gamma_{\mathrm{lin}}.
\]
The endpoint-uniform error in Lemma~\ref{lem:trial-exp-jp} is
\(O((\log n)/n)\).  Hence every such summand is at most
\(2^{-\gamma_{\mathrm{lin}}n/2}\) for all sufficiently large \(n\).
There are \(O(n^3)\) choices of \((t_1,t_\Delta,w)\), so the total
linear-weight contribution is \(o(1)\).

For the low-weight small-support region, impose
\[
0<\omega\le\omega_*,\qquad
t_1,t_\Delta\le\frac{\beta n}{k}.
\]
Apply Lemma~\ref{lem:pair-jp} with
\(\sigma_1=\sigma_\Delta=\beta/k\) and
\(\sigma_w=\omega_*\), and group by \(u=t_1+t_\Delta+w\).  Since
\(u\le c_Xn\), define the grouped contribution
\[
T_{X,n}(u):=
\sum_{\substack{t_1+t_\Delta+w=u\\
0\le t_1,t_\Delta\le\beta n/k\\1\le w\le\omega_*n}}
\EE[\NX(t_1,t_\Delta,w)].
\]
Vandermonde's identity and Lemma~\ref{lem:pair-jp} give
\[
T_{X,n}(u)
\le \binom{(1+A)n}{u}(ku-1)!!
       \left(\frac{C_X}{n}\right)^{ku/2}.
\]
Using \(\binom{N}{u}\le(\mathrm eN/u)^u\) and
\((ku-1)!!\le\sqrt2(ku/\mathrm e)^{ku/2}\), the right-hand side is at
most
\[
\sqrt2\left[
\frac{\mathrm e(1+A)}{c_X}
\left(\frac{C_Xkc_X}{\mathrm e}\right)^{k/2}
\right]^u
\le B_X^u,
\]
where the last inequality uses \(u\ge1\) and the definition of \(B_X\).
For each fixed \(u\), the unsimplified binomial--double-factorial bound is
\(O(n^{-u(k/2-1)})\).  Thus
\[
T_{X,n}(u)\le B_X^u,
\qquad
T_{X,n}(u)=O(n^{-u(k/2-1)})\quad(u\ \text{fixed}).
\]
Lemma~\ref{lem:mn-window-bounds-jp} gives \(B_X<7/20\).  Splitting a
geometric tail from a fixed initial segment therefore shows that the whole
small-support contribution is \(o(1)\).

Now suppose a half-domain low-weight term is outside the small-support
region.  Recall \(q=(\alpha_Za+\alpha_\Delta b)/A\).
Concavity of \(H\) and of \(x\mapsto\ln(1-2x)\) gives
\[
(\ln2)\Phi_{\mathrm{MN}}(a,b,\omega)\le\psi(q,\omega).
\]
Moreover, either \(t_1>\beta n/k\) or
\(t_\Delta>\beta n/k\), and hence
\[
q>\frac{\beta}{kA}=\frac1{10L}.
\]
The second part of Lemma~\ref{lem:mn-window-bounds-jp} supplies a uniform
negative exponent.  After summing the \(O(n^3)\) possible triples, the
low-weight large-support contribution is \(o(1)\).

If \(j_Z\) is even, then \(j_\Delta\) is even, and the folding symmetry in
Proposition~\ref{prop:fold-jp} transfers these three half-domain estimates
to the full domain.  Markov's inequality completes the proof in this case.

Suppose finally that \(j_Z\) is odd.  Since \(j_\Delta=k-2j_Z\) is even,
the \(A_\Delta\)-block still folds.  To make the remaining half explicit,
let \(N_X^-(t_1,t_\Delta,w)\) count triples of the indicated weights that
satisfy
\[
A_Z^T\bp_Z+A_\Delta^T\bp_\Delta+B^T\bv=\ones_{[n]}.
\]
Its one-column polynomial is
\[
g^-_{j_Z,j_\Delta,k}(u,v,r)
:=\frac{(1+u)^{j_Z}(1+v)^{j_\Delta}(1+r)^k
-(1-u)^{j_Z}(1-v)^{j_\Delta}(1-r)^k}{2},
\]
whose coefficients are nonnegative because it selects odd total parity in
each column.  The same socket count as in Lemma~\ref{lem:exact-jp} gives
the exact identity
\[
\EE[N_X^-(t_1,t_\Delta,w)]
=\binom{m_Z}{t_1}\binom{m_\Delta}{t_\Delta}\binom nw
\frac{[u^{kt_1}v^{kt_\Delta}r^{kw}]
      g^-_{j_Z,j_\Delta,k}(u,v,r)^n}
{\binom{nj_Z}{kt_1}\binom{nj_\Delta}{kt_\Delta}\binom{nk}{kw}}.
\tag{MN--4}
\]

After folding the \(A_\Delta\)-block, partition the \(A_Z\)-witness
weights at \(m_Z/2\).  Complementing an upper-half witness uses
\(A_Z^T\ones_{[m_Z]}=\ones_{[n]}\) and maps it bijectively to the
odd-parity count in (MN--4).  Consequently, for every \(w\),
\[
A_{\CX}(w)\le
2\sum_{0\le t_\Delta\le m_\Delta/2}
\left[
 \sum_{0\le t_1\le m_Z/2}\NX(t_1,t_\Delta,w)
 +\sum_{0\le t_1<m_Z/2}N_X^-(t_1,t_\Delta,w)
\right].                                                     \tag{MN--5}
\]

Applying the same trial points as in Lemma~\ref{lem:trial-exp-jp} to
(MN--4) gives, uniformly on \(0\le a,b\le1/2\) and \(1\le w\le n/2\),
\[
\frac1n\log_2\EE[N_X^-(t_1,t_\Delta,w)]
\le \Phi_{\mathrm{MN}}^-(a,b,\omega)
   +O\!\left(\frac{\log n}{n}\right),
\]
where \(\mu:=(1-2\omega)^k\) and
\[
\begin{aligned}
\Phi_{\mathrm{MN}}^-(a,b,\omega)
:={}&\alpha_Z\h(a)+\alpha_\Delta\h(b)+\h(\omega)-1\\
&+\log_2\!\left(1-\mu(1-2a)^{j_Z}(1-2b)^{j_\Delta}\right).
\end{aligned}                                                \tag{MN--6}
\]
The logarithm is finite because \(w\ge1\) implies \(0\le\mu<1\).
Since \(1-x\le1+x\), one has
\(\Phi_{\mathrm{MN}}^-\le\Phi_{\mathrm{MN}}\); hence the preceding
linear-weight and low-weight large-support estimates apply unchanged.

The only remaining part of the odd count is the small corner
\[
0\le a\le\frac\beta{j_Z},\qquad
0\le b\le\frac\beta{j_\Delta},\qquad
0<\omega\le\omega_*.
\]
Here \(\mu\ge A/2\), while
\((1-2a)^{j_Z}(1-2b)^{j_\Delta}\ge\eta\).  Monotonicity of entropy
therefore gives
\[
\Phi_{\mathrm{MN}}^-(a,b,\omega)
\le
\alpha_Z\h\!\left(\frac\beta{j_Z}\right)
+\alpha_\Delta\h\!\left(\frac\beta{j_\Delta}\right)
+\h(\omega_*)-1+\log_2\!\left(1-\frac A2\eta\right)
<-\frac{23}{25}< -\frac9{10}
\]
by Lemma~\ref{lem:mn-window-bounds-jp}.  After absorbing the uniform
\(O((\log n)/n)\) term, each odd small-corner summand is exponentially
small.  There are only \(O(n^3)\) summands in (MN--5), so the odd upper
half contributes \(o(1)\).

In both parity cases,
\[
\sum_{1\le w\le\delta n}\EE[A_{\CX}(w)]=o(1),
\]
and Markov's inequality gives
\(\PP[d(\CX)\le\delta n]\to0\).
\end{proof}

\engappendixsection{Proof of Proposition~\ref{prop:zero-rate-boundary-jp}}{app:zero-rate-boundary-jp}

Fix \((j_Z,j_X,k)=(j,j,2j)\), where \(4\le j\le15\). Here
\(A_X=A_Z\), \(C_X=C_Z^\perp\), and the quantum dimension is exactly zero.
The following argument concerns the two classical constituent codes.
Write \(H(x):=(\ln2)\h(x)\), and put
\[
\bar d:=\frac{11002787}{10^8},\qquad
\tau_0:=\frac1{4k},\qquad
\omega_*:=\frac{1-4^{-1/k}}2.
\]
Set \(\beta=3/25\) when \(j=4\), and \(\beta=1/5\) otherwise.
Unlike the positive-rate proof, this boundary proof includes \(k=8\)
and uses \(\alpha_Z=1/2\) throughout.

\begin{lemma}[Exact rational boundary certificates]\label{lem:boundary-certificates}
For all 12 boundary triples, the following inequalities hold:
\[
\h(\bar d)>\frac12,\qquad
p(\tau_0)>\bar d,\qquad
p(\tau):=\frac{1-(1-2\tau)^k}{2},
\]
\[
\lambda_Z:=k\mathrm e^{1-j/2}(1/4)^{j/2-1}<\frac34.
\]
Define
\[
\rho:=\max\{\beta/j,\omega_*\},\quad
c_X:=\beta/k+\omega_*,\quad
C:=\frac{1+\rho}{(1-\rho)^2},\quad
B_X:=\sqrt2\,\frac{3\mathrm e}{2c_X}
\left(\frac{Ckc_X}{\mathrm e}\right)^j.
\]
Then \(\rho<1\), \(B_X<97/100\), and
\[
\sup_{\tau_0\le\tau\le49/100}
\left\{H(\tau)-\tfrac12\ln2+
\tfrac12\ln(1+(1-2\tau)^k)-D_{\rm nat}(\bar d\|p(\tau))\right\}
<-\frac1{100000},                                      \tag{F--1}
\]
\[
\sup_{\substack{\beta/j\le a\le1/2\\0\le w\le\omega_*}}
\left\{\tfrac12 H(a)+H(w)-\ln2+
\ln(1+(1-2w)^k(1-2a)^j)\right\}
<-\frac1{100000}.                                      \tag{F--2}
\]
Here \(D_{\rm nat}\) is Bernoulli relative entropy using natural logarithms.
For odd \(j\), one additionally has
\[
\tfrac12H(\beta/j)+H(\omega_*)-\ln2+
\ln\left(1-\tfrac14(1-2\beta/j)^j\right)<-\tfrac12.       \tag{F--3}
\]
\end{lemma}

\begin{proof}
All comparisons are verified by the accompanying standard-library script
\path{verify_boundary_bounds.py}. We specify the bounds and interval
coverage so that the computation is reproducible. Every endpoint and
operation is rational. Logarithms use the series and explicit remainder
of (MN--1), with 24 terms and conservative rational enclosure afterwards.
The root \(4^{-1/k}\) is bracketed between adjacent rationals of denominator
\(10^{12}\), checked by exact integer powers. Denote the induced bracket
for \(\omega_*\) by \([\omega^-,\omega^+]\).

The small-support comparisons use \(\sqrt2<70711/50000\) and
\(\mathrm e>1359/500\), with the integer and factorial checks of
Appendix~\ref{app:finite-gv-mn-jp}. In particular,
\[
\lambda_Z^2< k^2(1/4)^{j-2}(1359/500)^{2-j}<9/16.
\]
For \(B_X\), replace \(\rho,c_X\) by the upper bounds obtained using
\(\omega^+\), and rewrite its dependence on \(\mathrm e\) as
\(\mathrm e^{1-j}\) before using this rational lower bound.

For (F--1), on an interval \([a,b]\subset[\tau_0,49/100]\), an upper
bound is
\[
\begin{split}
U_Z(a,b):={}&H(b)-\tfrac12\ln2+
\tfrac12\ln(1+(1-2a)^k)+H(\bar d)\\
&+\bar d\ln p(a)+(1-\bar d)\ln(1-p(a)).
\end{split}                                            \tag{F--4}
\]
Indeed, \(H\) increases on the half interval, the parity factor decreases,
and \(-D_{\rm nat}(\bar d\|p)\) decreases for \(p>\bar d\).

For (F--2), fix \(a\), write \(z=(1-2a)^j\), and use convexity of
\(w\mapsto\ln(1+z(1-2w)^k)\). Its chord on \([0,\omega_*]\) gives
\[
\ln(1+z(1-2w)^k)\le\ln(1+z)-c(z)w,\qquad
c(z):=\frac{\ln(1+z)-\ln(1+z/4)}{\omega_*}.
\]
For an interval \([a,b]\subset[\beta/j,1/2]\), put
\(z^+=(1-2a)^j\), \(z^-=(1-2b)^j\).
The numerator of \(c\) increases with \(z\). A rational lower bound
\(c^-\ge0\) is therefore obtained from
\([\ln(1+z^-)-\ln(1+z^-/4)]/\omega^+\).
An upper bound for the interval is
\[
U_X(a,b):=\tfrac12 H(b)-\ln2+\ln(1+z^+)
+\sup_{0\le w\le\omega^+}\{H(w)-c^-w\}.                \tag{F--5}
\]
If \(c^-\le\ln((1-\omega^+)/\omega^+)\), the supremum is
\(H(\omega^+)-c^-\omega^+\). Otherwise it is bounded by
\(\ln(1+\exp(-c^-))\). To enclose the exponential without floating
point, for a rational \(0\le c_0\le c^-\) use
\[
\exp(-c^-)\le\exp(-c_0)
\le\left(\sum_{r=0}^{48}\frac{c_0^r}{r!}\right)^{-1}.
\]
The script uses only certified branch comparisons and conservative
rational logarithm bounds in (F--4) and (F--5).

Each initial interval is bisected until its certified upper bound is
below \(-1/100000\). The saved leaf intervals have exactly matching
adjacent endpoints and cover the entire initial interval. The numbers
of leaves for \(j=4,5,\ldots,15\), respectively, are
\[
\begin{array}{c|rrrrrrrrrrrr}
j&4&5&6&7&8&9&10&11&12&13&14&15\\\hline
\text{HA}&11&9&8&7&6&6&6&6&6&6&6&6\\
\text{MN}&8&5&4&4&4&4&4&4&4&4&4&4
\end{array}
\]
The exact rational endpoints and certified upper bounds are saved in
\path{boundary_certificates.json}. Direct rational logarithm bounds
also verify (F--3), using \(\omega^+\) in its entropy term.
\end{proof}

\begin{proof}[Proof of Proposition~\ref{prop:zero-rate-boundary-jp}]
Fix \(0<\delta<d:=\h^{-1}(1/2)<\bar d\).
The coefficient-extraction and transition identities in
Appendix~\ref{app:HA-proofs-jp} remain valid for \(j=k/2\) and \(k=8\).
For \(1\le s\le\tau_0n\), their socket-pairing proof gives
\[
N_o(s)\le\binom ns(js-1)!!(n/2)^{-js/2}
\le\sqrt2\,\lambda_Z^s
\]
when \(js\) is even, and \(N_o(s)=0\) otherwise. For fixed \(s\),
the unsimplified bound is \(O(n^{-s(j/2-1)})\). Thus a fixed initial
segment and a geometric tail give \(\sum_{s\le\tau_0n}N_o(s)=o(1)\).
Complement symmetry handles \(s\ge(1-\tau_0)n\). The identity
\(\sum_lM_k(s,l)=1\) makes both tails \(o(1)\) after the visible
weights have been summed.

On \(\tau_0\le\tau\le49/100\), (F--1) bounds the HA exponent
uniformly below zero for every \(0\le\omega\le\delta\), since
\(p(\tau)>\bar d>\delta\) and \(-D_{\rm nat}(\omega\|p)\) increases
with \(\omega<p\). For \(49/100\le\tau\le1/2\), put
\(u=1-2\tau\le1/50\). Pinsker's inequality and \(k\ge8\) give
\[
1-\h(\tau)\ge\frac{u^2}{2\ln2},\qquad
\frac32\log_2(1+u^k)\le\frac{3u^8}{2\ln2}
\le\frac{u^2}{2\ln2}.
\]
Therefore the exponent is at most \(\h(\delta)-1/2<0\).
Complementing the input preserves both \(N_o(s)\) and \(M_k(s,l)\),
so the upper half has the same bound. Uniform Stirling errors are
\(O(\log n)\); there are only \(O(n^2)\) pairs \((s,l)\).
The total first moment is consequently \(o(1)\), proving the HA assertion.

On the MN side, the two-pool version of the exact socket formula gives
\[
\EE N^{(0)}(t,w)=\binom{n/2}{t}\binom nw
\frac{[u^{kt}r^{kw}]
\left(\frac{(1+u)^j(1+r)^k+(1-u)^j(1-r)^k}{2}\right)^n}
{\binom{nj}{kt}\binom{nk}{kw}}.
\]
There is no absent-pool denominator in this formula. With
\(a=2t/n\) and \(\omega=w/n\), the trial exponent in bits is
\[
\Phi_{+,0}(a,\omega)=\tfrac12\h(a)+\h(\omega)-1+
\log_2(1+(1-2\omega)^k(1-2a)^j).
\]
For \(\omega\ge\omega_*\), the parity factor is at most \(1/4\),
so Pinsker's inequality and \((1-2a)^j\le(1-2a)^2\) give
\(\Phi_{+,0}\le\h(\delta)-1/2<0\).

For \(0<\omega\le\omega_*\) and \(t\le\beta n/k\), the exposure
proof of Lemma~\ref{lem:pair-jp} applies to the two nonempty pools.
Grouping by \(u=t+w\le c_Xn\) gives
\[
T_n(u)\le\binom{3n/2}{u}(ku-1)!!(C/n)^{ku/2}
\le B_X^u,
\qquad
T_n(u)=O(n^{-u(k/2-1)})\quad(u\text{ fixed}).
\]
Since \(B_X<1\), the grouped sum is \(o(1)\). The complementary
half-domain has \(a>\beta/j\) and is covered by (F--2).

If \(j\) is even, complementing the witness folds the full domain to
\(a\le1/2\). If \(j\) is odd, complementing an upper-half witness
changes the syndrome to \(\ones\). The odd-parity generating polynomial
has nonnegative coefficients and gives
\[
\Phi_{-,0}(a,\omega)=\tfrac12\h(a)+\h(\omega)-1+
\log_2(1-(1-2\omega)^k(1-2a)^j).
\]
For \(w\ge1\), the logarithm is finite, and \(\Phi_{-,0}\le\Phi_{+,0}\)
handles the linear-weight and large-support regions. In the remaining
corner \(a\le\beta/j\), \(0<\omega\le\omega_*\), (F--3) bounds the
odd exponent uniformly below zero. All coefficient upper bounds have
uniform \(O(\log n)\) Stirling errors, and at most \(O(n^2)\) terms
are summed. Thus the MN first moment is \(o(1)\). Markov's inequality
proves the second distance assertion.

The same estimates also establish the actual constituent rates. Set
\(w=0\) in the one-block enumerator above. For nonzero witnesses in the
lower half, the pairing estimate bounds the small-support sum by \(o(1)\),
and (F--2) bounds the remaining sum exponentially. If \(j\) is even,
complement symmetry gives at most twice the lower-half contribution,
including the zero witness. If \(j\) is odd, the upper half has odd
syndrome and is bounded by (F--2) and (F--3); the zero-support odd term
vanishes. Hence
\[
\EE|\Ker A_Z^T|=O(1),\qquad
\dim\Ker A_Z^T=o(n)\quad\text{in probability}.
\]
Thus \(\dim\Ker A_Z=n/2+o(n)\). Proposition~\ref{prop:B-rank-jp}
applies also at \(k=8\) and gives \(\dim\Ker B=o(n)\) in probability.
Since \(\CZ=B(\Ker A_Z)\),
\[
\dim\Ker A_Z-\dim\Ker B\le\dim\CZ\le\dim\Ker A_Z.
\]
Consequently \(R_Z\to1/2\), and \(\CX=\CZ^\perp\) gives
\(R_X\to1/2\), both in probability.
\end{proof}

\bibliographystyle{IEEEtran}
\bibliography{references}

\begin{thebibliography}{10}
\providecommand{\url}[1]{#1}
\csname url@samestyle\endcsname
\providecommand{\newblock}{\relax}
\providecommand{\bibinfo}[2]{#2}
\providecommand{\BIBentrySTDinterwordspacing}{\spaceskip=0pt\relax}
\providecommand{\BIBentryALTinterwordstretchfactor}{4}
\providecommand{\BIBentryALTinterwordspacing}{\spaceskip=\fontdimen2\font plus
\BIBentryALTinterwordstretchfactor\fontdimen3\font minus
  \fontdimen4\font\relax}
\providecommand{\BIBforeignlanguage}[2]{{%
\expandafter\ifx\csname l@#1\endcsname\relax
\typeout{** WARNING: IEEEtran.bst: No hyphenation pattern has been}%
\typeout{** loaded for the language `#1'. Using the pattern for}%
\typeout{** the default language instead.}%
\else
\language=\csname l@#1\endcsname
\fi
#2}}
\providecommand{\BIBdecl}{\relax}
\BIBdecl

\bibitem{PK21}
P.~Panteleev and G.~Kalachev, ``Asymptotically good quantum and locally
  testable classical {LDPC} codes,'' in \emph{Proc. 54th Annual ACM SIGACT
  Symposium on Theory of Computing (STOC)}, 2022, pp. 375--388.

\bibitem{DHLV22}
I.~Dinur, M.-H. Hsieh, T.-C. Lin, and T.~Vidick, ``Good quantum {LDPC} codes
  with linear time decoders,'' in \emph{Proc. 55th Annual ACM Symposium on
  Theory of Computing (STOC)}, 2023, pp. 905--918.

\bibitem{LZ22}
A.~Leverrier and G.~Z{\'e}mor, ``Quantum tanner codes,'' in \emph{Proc. 2022
  IEEE 63rd Annual Symposium on Foundations of Computer Science (FOCS)}, 2022,
  pp. 872--883.

\bibitem{Tanner81}
R.~M. Tanner, ``A recursive approach to low complexity codes,'' \emph{IEEE
  Transactions on Information Theory}, vol.~27, no.~5, pp. 533--547, 1981.

\bibitem{MN96}
D.~J.~C. MacKay and R.~M. Neal, ``Near shannon limit performance of low density
  parity check codes,'' \emph{Electronics Letters}, vol.~32, no.~18, pp.
  1645--1646, 1996.

\bibitem{HA05}
C.-H. Hsu and A.~Anastasopoulos, ``Capacity-achieving codes with bounded
  graphical complexity and maximum likelihood decoding,'' \emph{IEEE
  Transactions on Information Theory}, vol.~56, no.~3, pp. 992--1006, 2010.

\bibitem{RU08}
T.~Richardson and R.~Urbanke, \emph{Modern Coding Theory}.\hskip 1em plus 0.5em
  minus 0.4em\relax Cambridge University Press, 2008.

\bibitem{SipserS94}
M.~Sipser and D.~A. Spielman, ``Expander codes,'' in \emph{Proc. 35th Annual
  Symposium on Foundations of Computer Science (FOCS)}, 1994, pp. 566--576.

\bibitem{TZ09}
J.-P. Tillich and G.~Z{\'e}mor, ``Quantum {LDPC} codes with positive rate and
  minimum distance proportional to the square root of the blocklength,''
  \emph{IEEE Transactions on Information Theory}, vol.~60, no.~2, pp.
  1193--1202, 2014.

\bibitem{LeverrierTZ15}
A.~Leverrier, J.-P. Tillich, and G.~Z{\'e}mor, ``Quantum expander codes,'' in
  \emph{Proc. 2015 IEEE 56th Annual Symposium on Foundations of Computer
  Science (FOCS)}, 2015, pp. 810--824.

\bibitem{EvraKZ20}
S.~Evra, T.~Kaufman, and G.~Z{\'e}mor, ``Decodable quantum {LDPC} codes beyond
  the square root distance barrier using high dimensional expanders,'' in
  \emph{Proc. 2020 IEEE 61st Annual Symposium on Foundations of Computer
  Science (FOCS)}, 2020, pp. 218--227.

\bibitem{BreuckmannE21}
N.~P. Breuckmann and J.~N. Eberhardt, ``Balanced product quantum codes,''
  \emph{IEEE Transactions on Information Theory}, vol.~67, no.~10, pp.
  6653--6674, 2021.

\bibitem{MKC09}
R.~E. Moore, R.~B. Kearfott, and M.~J. Cloud, \emph{Introduction to Interval
  Analysis}.\hskip 1em plus 0.5em minus 0.4em\relax SIAM, 2009.

\bibitem{Tucker11}
W.~Tucker, \emph{Validated Numerics: A Short Introduction to Rigorous
  Computations}.\hskip 1em plus 0.5em minus 0.4em\relax Princeton University
  Press, 2011.

\bibitem{CS96}
A.~R. Calderbank and P.~W. Shor, ``Good quantum error-correcting codes exist,''
  \emph{Physical Review A}, vol.~54, no.~2, pp. 1098--1105, 1996.

\bibitem{Ste96}
A.~M. Steane, ``Multiple-particle interference and quantum error correction,''
  \emph{Proceedings of the Royal Society of London. Series A}, vol. 452, pp.
  2551--2577, 1996.

\bibitem{BayatiMS09}
\BIBentryALTinterwordspacing
M.~Bayati, A.~Montanari, and A.~Saberi, ``Generating random graphs with large
  girth,'' in \emph{Proceedings of the Twentieth Annual {ACM-SIAM} Symposium on
  Discrete Algorithms ({SODA} 2009)}, 2009, pp. 566--575. [Online]. Available:
  \url{https://dl.acm.org/doi/10.5555/1496770.1496833}
\BIBentrySTDinterwordspacing

\bibitem{BayatiKMOS09}
M.~Bayati, R.~Keshavan, A.~Montanari, S.~Oh, and A.~Saberi, ``Generating random
  {Tanner}-graphs with large girth,'' in \emph{2009 {IEEE} Information Theory
  Workshop}, 2009, pp. 154--157.

\bibitem{AlonHL02}
N.~Alon, S.~Hoory, and N.~Linial, ``The moore bound for irregular graphs,''
  \emph{Graphs and Combinatorics}, vol.~18, no.~1, pp. 53--57, 2002.

\bibitem{AmirzadePS24}
F.~Amirzade, D.~Panario, and M.~Sadeghi, ``Quantum quasi-cyclic {LDPC} codes
  with column weight at least 3 have girth at most 6,'' \emph{Problems of
  Information Transmission}, vol.~60, no.~2, pp. 3--11, 2024.

\bibitem{Gal63}
R.~G. Gallager, \emph{Low-Density Parity-Check Codes}.\hskip 1em plus 0.5em
  minus 0.4em\relax MIT Press, 1963.

\bibitem{Gil52}
E.~N. Gilbert, ``A comparison of signalling alphabets,'' \emph{Bell System
  Technical Journal}, vol.~31, no.~3, pp. 504--522, 1952.

\bibitem{Var57}
R.~R. Varshamov, ``The evaluation of signals in codes with correction of
  errors,'' \emph{Doklady Akademii Nauk SSSR}, vol. 117, no.~5, pp. 739--741,
  1957.

\bibitem{DIRU06}
C.~Di, T.~J. Richardson, and R.~L. Urbanke, ``Weight distribution of
  low-density parity-check codes,'' \emph{IEEE Transactions on Information
  Theory}, vol.~52, no.~11, pp. 4839--4855, 2006.

\bibitem{KasaiMET10}
K.~Kasai, T.~Awano, D.~Declercq, C.~Poulliat, and K.~Sakaniwa, ``Weight
  distributions of multi-edge type {LDPC} codes,'' \emph{IEICE Transactions on
  Fundamentals of Electronics, Communications and Computer Sciences}, vol.
  E93-A, no.~11, pp. 1942--1948, 2010.

\bibitem{Mitchell12}
D.~G.~M. Mitchell, K.~Kasai, M.~Lentmaier, and J.~D.~J.~Costello, ``Asymptotic
  analysis of spatially coupled {MacKay-Neal} and {Hsu-Anastasopoulos} {LDPC}
  codes,'' in \emph{Proc. 2012 International Symposium on Information Theory
  and its Applications (ISITA)}, 2012, pp. 337--341.

\bibitem{Zahr25}
A.~Zahr, E.~B. Yacoub, B.~Matuz, and G.~Liva, ``Rate-adaptive protograph-based
  {MacKay-Neal} codes,'' \emph{IEEE Transactions on Information Theory},
  vol.~71, no.~2, pp. 914--929, 2025.

\bibitem{KasaiGVData26}
K.~Kasai, ``Finite-degree {GV} certificates: Data and verification code,''
  \url{https://github.com/kasaikenta/finite-degree-gv-certificates}, 2026,
  accessed 2026-09-13.

\bibitem{ObataJKP13}
N.~Obata, Y.-Y. Jian, K.~Kasai, and H.~D. Pfister, ``Spatially-coupled
  multi-edge type {LDPC} codes with bounded degrees that achieve capacity on
  the {BEC} under {BP} decoding,'' in \emph{Proc. IEEE International Symposium
  on Information Theory (ISIT)}, 2013, pp. 2433--2437.

\bibitem{KasaiSak11}
K.~Kasai and K.~Sakaniwa, ``Spatially-coupled {MacKay-Neal} codes and
  {Hsu-Anastasopoulos} codes,'' \emph{IEICE Transactions on Fundamentals of
  Electronics, Communications and Computer Sciences}, vol. E94-A, no.~11, pp.
  2161--2168, 2011.

\bibitem{FukushimaOK15}
M.~Fukushima, T.~Okazaki, and K.~Kasai, ``Spatially-coupled {MacKay-Neal} codes
  universally achieve the symmetric information rate of arbitrary generalized
  erasure channels with memory,'' in \emph{Proc. IEEE International Symposium
  on Information Theory (ISIT)}, 2015, pp. 899--903.

\bibitem{HagiwaraKIS11}
M.~Hagiwara, K.~Kasai, H.~Imai, and K.~Sakaniwa, ``Spatially coupled
  quasi-cyclic quantum {LDPC} codes,'' in \emph{Proc. IEEE International
  Symposium on Information Theory (ISIT)}, 2011, pp. 638--642.

\bibitem{KomotoKasai25}
D.~Komoto and K.~Kasai, ``Quantum error correction near the coding theoretical
  bound,'' \emph{npj Quantum Information}, vol.~11, p. 154, 2025.

\bibitem{GKP94}
R.~L. Graham, D.~E. Knuth, and O.~Patashnik, \emph{Concrete Mathematics: A
  Foundation for Computer Science}, 2nd~ed.\hskip 1em plus 0.5em minus
  0.4em\relax Addison-Wesley Professional, 1994.

\bibitem{Robbins55}
H.~Robbins, ``A remark on {Stirling's} formula,'' \emph{American Mathematical
  Monthly}, vol.~62, no.~1, pp. 26--29, 1955.

\bibitem{Clark36}
J.~A. Clarkson, ``Uniformly convex spaces,'' \emph{Transactions of the American
  Mathematical Society}, vol.~40, no.~3, pp. 396--414, 1936.

\end{thebibliography}

\end{document}